\documentclass[11pt]{article}
\usepackage{microtype}
\usepackage{algorithm}
\usepackage{algpseudocode}

\usepackage{mathtools} 
\usepackage{color}
\usepackage{amsmath,amssymb,amsthm}
\usepackage{enumitem}
\usepackage{mathtools} 
\usepackage{relsize}
\usepackage{euscript}
\usepackage{subcaption}

\usepackage{hyperref}

\usepackage{tikz}
\usepackage{bm}

\usepackage{esvect}
\usepackage{geometry}
\usepackage{stackengine}
\newcommand{\circled}[1]{\stackMath\stackinset{c}{}{c}{}{\scriptscriptstyle #1}{\scalebox{0.65}{$\bigcirc$}}}

\newtheorem{theorem}{Theorem}[section]
\newtheorem{ther}{Theorem}
\newtheorem{lemma}[theorem]{Lemma}

\newtheorem{proposition}[theorem]{Proposition}
\newtheorem{prop}{Proposition}
\newtheorem*{proposition-no}{Proposition}

\newtheorem{remark}[theorem]{Remark}
\newtheorem{example}[theorem]{Example}

\let\al=\alpha

\let\Gm=\Gamma

\let\vf=\varphi

\let\sg=\sigma

\let\th=\theta

\def\ovarrow{}
\def\circp{\circled{p}}

\def\rel{\EuScript{R}}
\def\relo{\EuScript{Q}}
\def\rela{\EuScript{S}}

\def\cE{\mathcal E}
\def\cF{\mathcal F}
\def\cG{\mathcal G}
\def\cH{\mathcal H}
\def\cI{\mathcal I}

\def\cK{\mathcal K}
\def\cP{\mathcal P}

\def\cT{\mathcal T}

\def\cC{\mathcal C}

\def\ba{\mathbf{a}}
\def\bb{\mathbf{b}}
\def\bc{\mathbf{c}}
\def\bd{\mathbf{d}}
\def\bs{\mathbf{s}}
\def\bx{\mathbf{x}}
\def\by{\mathbf{y}}
\def\bz{\mathbf{z}}

\def\bv{\mathbf{v}}
\def\sH{\mathsf{H}}
\def\sG{\mathsf{G}}
\def\sK{\mathsf{K}}
\def\sM{\mathsf{M}}

\def\const{\mathsf{c}}
\def\zz{{\underline{0}}}

\def\WT#1{\widetilde #1}
\def\RT{{\mathfrak{R}}}

\def\CSP{\mathsf{CSP}}

\def\NCSP{\mathsf{\#CSP}}
\def\NpCSP{\#_p\mathsf{CSP}}

\def\Aut{\mathsf{Aut}}
\def\Fix{\mathsf{Fix}}

\def\Part{\mathsf{Part}}

\def\PAR #1{\mathsf{PAR}_{#1}}

\def\ghom#1{\#\mathsf{Hom}(#1)}
\def\ghomk#1#2{\#_{#1}\mathsf{Hom}(#2)}
\def\CSPp#1{\#_p\mathsf{CSP}(#1)}

\def\sgn{\mathsf{sgn}}
\def\bbZ{\mathbb{Z}}

\def\type{\mathsf{type}}

\def\pr{\mathsf{pr}}

\def\Hom{\mathsf{Hom}}

\def\CSP{\mathsf{CSP}}

\def\Aut{\mathsf{Aut}}
\def\Fix{\mathsf{Fix}}

\def\proj{\mathsf{pr}}
\def\prp{\mathsf{pr}^p}
\def\prpar{\mathsf{pr}^2}
\def\ext{\mathsf{\# ext}}
\def\extp{\mathsf{\#_p ext}}

\def\hom{\mathsf{hom}}
\def\inj{\mathsf{inj}}

\def\Con{\mathsf{Con}}

\def\cp#1#2{\arraycolsep0pt
\left(\begin{array}{c} #1\\ #2 \end{array}\right)}
\def\cpp#1#2#3{\arraycolsep0pt \left(\begin{array}{c} #1\\ #2\\
#3 \end{array}\right)}

\let\eps=\emptyset
\let\ov=\overline

\def\ang#1{\langle#1\rangle}
\let\sse=\subseteq

\let\tm=\times
\def\tms{\tm\dots\tm}
\def\zd{,\dots,}
\def\vv#1{\mathbf{#1}}
\def\vc#1#2{#1_1,\dots,#1_{#2}}
\def\modp{$\#_p$P}
\def\modk{$\#_k$P}
\def\ar{\mathsf{ar}}
\def\nat{\mathbb{N}}
\def\red#1{\vrule height7pt depth3pt width.4pt
\lower3pt\hbox{$\scriptstyle #1$}}
\let\meet=\wedge

\def\existspr{\exists^{\equiv p}}

\newcommand{\FUNC}[3]{#1 : #2 \rightarrow #3}

\newcommand{\constCSP}[2]{\langle #1,#2 \rangle}

\newcommand{\SP}[1]{\langle #1 \rangle}

\newcommand{\colect}[2]{\{ #1 \}_{{#2}} }

\usepackage{scalerel}
\usepackage{stackengine}
\stackMath

\newcommand{\vvecad}[2]{
    \begin{pmatrix}
    {#1} \\
    \vdots \\
    {#2}
    \end{pmatrix}
}

\begin{document}
\date{}

\title{Modular Counting CSP: Reductions and Algorithms}
%\titlerunning{Modular Counting CSP}
\author{Amirhossein Kazeminia and Andrei A.Bulatov \\ Simon Fraser University}
%\author{Anonymous authors}{Some affiliation}{}{}{}
%\authorrunning{Anonymous authors}
%\keywords{Constraint Satisfaction Problem, Modular Counting}
\maketitle

\begin{abstract}
The Constraint Satisfaction Problem (CSP) is ubiquitous in various areas of mathematics and computer science. Many of its variations have been studied including the Counting CSP, where the goal is to find the number of solutions to a CSP instance. The complexity of finding the exact number of solutions of a  CSP is well understood (Bulatov, 2013, and Dyer and Richerby, 2013) and the focus has shifted to other variations of the Counting CSP such as counting the number of solutions modulo an integer. This problem has attracted considerable attention recently. In the case of CSPs based on undirected graphs Bulatov and Kazeminia (STOC 2022) obtained a complexity classification for the problem of counting solutions modulo $p$ for arbitrary prime $p$. In this paper we report on the progress made towards a similar classification for the general CSP, not necessarily based on graphs.

We identify several features that make the general case very different from the graph case such as a stronger form of rigidity and the structure of automorphisms of powers of relational structures. We provide a solution algorithm in the case $p=2$ that works under some additional conditions and prove the hardness of the problem under some assumptions about automorphisms of the powers of the relational structure. We also reduce the general CSP to the case that only uses binary relations satisfying strong additional conditions.  
\end{abstract}

\newpage

%%%%%%%%%%%%%%%%%%%%%%%%%%%%%%%%%%%%%%%
%%%%%%%%%%%%%%%%%%%%%%%%%%%%%%%%%%%%%%%
\section{Introduction}
%% \input{Data/b-intro}

% general stuff about counting and CSP
Counting problems in general have been intensively studied since the pioneering work by Valiant \cite{Valiant79:complexity,Valiant79:computing}. One of the most interesting and well studied problems in this area is the Counting Constraint Satisfaction Problem ($\#\CSP$), which  provides a generic framework for a wide variety of counting combinatorial problems that arise frequently in multiple disciplines such as logic, graph theory, and artificial intelligence. The counting $\CSP$ also allows for generalizations including weighted CSPs and partition functions \cite{Barvinok16:combinatorics,Bulatov05:partition} that yield connections with areas such as statistical physics, see, e.g.\ \cite{Jerrum93:polynomial,Lieb81:general}. While the complexity of exact counting solutions of a $\CSP$ is now well-understood \cite{Dyer00:complexity,DBLP:journals/jacm/Bulatov13,effective-Dyer-doi:10.1137/100811258,Dalmau04:complexity}, modular counting such as finding the parity of the number of solutions remains widely open. 

\vspace{3mm}

\noindent
\textbf{Homomorphisms and the Constraint Satisfaction Problem.}
The most convenient way to introduce CSPs is through homomorphisms of relational structures. A \emph{relational signature} $\sg$ is a collection of \emph{relational symbols} each of which is assigned a positive integer, the \emph{arity} of the symbol. A \emph{relational structure} $\cH$ with signature $\sg$ is a set $H$ and an \emph{interpretation} $\rel^\cH$ of each $\rel\in\sg$, where $\rel^\cH$ is a relation or a predicate on $H$ whose arity equals that of $\rel$. The set $H$ is said to be the \emph{base set} or the \emph{universe} of $\cH$. We will use the same letter for the base set as for the structure, only in the regular font. A structure with signature $\sg$ is often called a \emph{$\sg$-structure}. Structures with the same signature are called \emph{similar}. A $\sg$-structure $\cH$ is \emph{finite} if both, $H$ and $\sigma$ are finite. In this paper all structures are finite.

Let $\cG,\cH$ be similar structures with signature $\sg$. A \emph{homomorphism} from $\cG$ to $\cH$ is a mapping $\vf:G\to H$ such that for any $\rel\in\sg$, say, of arity $r$, if $\rel^\cG(\vc ar)$ is true for $\vc ar\in G$, then $\rel^\cH(\vf(a_1)\zd\vf(a_r))$ is also true. The set of all homomorphisms from $\cG$ to $\cH$ is denoted $\Hom(\cG,\cH)$. The cardinality of $\Hom(\cG,\cH)$ is denoted by $\hom(\cG,\cH)$. A homomorphism $\vf$ is an \emph{isomorphism} if it is bijective and the inverse mapping $\vf^{-1}$ is a homomorphism from $\cH$ to $\cG$. A homomorphism of a structure to itself is called an \emph{endomorphism}, and an isomorphism to itself is called an \emph{automorphism}.

Following Feder and Vardi \cite{Feder98:computational}, in a $\CSP$, the goal is, given similar relational structure $\cG,\cH$, to decide whether there is a homomorphism from $\cG$ to $\cH$. The restricted problem in which $\cH$ is fixed and only $\cG$ is given as an input is denoted by $\CSP(\cH)$. The complexity of such problems is well understood \cite{Bulatov17:dichotomy,Zhuk20:dichotomy}.

\vspace{3mm}

\noindent
\textbf{Counting CSP.}
In the (exact) Counting $\CSP$ the goal is to find the number $\hom(\cG,\cH)$ of homomorphisms from a  structure $\cG$ to a similar  structure $\cH$. Restricted versions of the Counting CSP can be introduced in the same way as for the decision one. In the counting version of $\CSP(\cH)$ denoted $\NCSP(\cH)$ the goal is to find $\hom(\cG,\cH)$ for a given structure $\cG$.

% exact counting: results (vaguely), complexity classes
The complexity class the Counting $\CSP$ belongs to is \#P, the class of problems of counting accepting paths of polynomial time nondeterministic Turing machines. There are several ways to define reductions between counting problems, but the most widely used ones are parsimonious reductions and Turing reductions. A \emph{parsimonious reduction} from a counting problem $\#A$ to a counting problem $\#B$ is an algorithm that, given an instance $I$ of $\#A$ produces (in polynomial time) an instance $I'$ of $\#B$ such that the answers to $I$ and $I'$ are equal. A \emph{Turing reduction} is a polynomial time algorithm that solves $\#A$ using $\#B$ as an oracle. 
%% The fact that $\#A$ is Turing reducible to $\#B$ is denoted by $\#A\le_T\# B$. 
Completeness in \#P is then defined in the standard way. This paper and all the papers we cite predominantly use Turing reductions.

A complexity classification of counting $\CSP$s of the form $\NCSP(\cH)$ was obtained by Bulatov~\cite{DBLP:journals/jacm/Bulatov13} and was further improved and simplified by Dyer and Richerby  \cite{effective-Dyer-doi:10.1137/100811258}. Bulatov’s proof makes heavy use of techniques of universal algebra. Dyer and Richerby's proof, on the other hand, uses combinatorial and structural properties of relational structures. The more general version of the counting CSP, the weighted CSP, has also been thoroughly studied. Cai and Chen~\cite{cai-complex} obtained a complexity classification for weighted CSP, where each homomorphism has a complex weight. One of the main features of counting with complex weights is the phenomenon of cancellation, when complex weights of homomorphisms cancel each other rather than add up. This, of course, never happens in exact unweighted counting problems, but is frequently encountered in modular counting. 
% modular counting, classes

\vspace{3mm}

\noindent
\textbf{Modular Counting.}
Another natural variation of counting problems is counting modulo some integer. In this paper we consider the problem of computing the number of solutions of a $\CSP$ modulo a prime number $p$. If a relational structure $\cH$ is fixed, this problem is denoted $\NpCSP(\cH)$. More precisely, in $\NpCSP(\cH)$ the objective is, given a relational structure $\cG$, to find the number of homomorphisms from $\cG$ to $\cH$ modulo $p$. 
%% In a more general setting we need to find the number $\hom(\cG,\cH)$ modulo some fixed number $p$. This problem is denoted by $\#_p\CSP(\cH)$.

There are several complexity classes related to modular counting. The more established type of classes is $\mathsf{Mod}_k$P, the class of problems of deciding whether the number of accepting paths of a polynomial time nondeterministic Turing machine is \emph{not} divisible by $k$, \cite{Cai89:power,Hertrampf90:relations}. In particular, if $k=2$ then $\mathsf{Mod}_k$P is the class $\oplus$P. However, problems of counting accepting paths naturally belong to classes of the form \modk,  introduced by Faben in \cite{faben2008complexity} that contain problems of counting accepting paths modulo $k$. The standard notion of reduction is again Turing reduction. Faben in \cite{faben2008complexity} studied the basic properties of such classes, in particular, he identified several \modk-complete problems.

% problem definition
In the case of the $\CSP$, the research has mostly been focused on graph homomorphisms. The only exceptions we are aware of are a result of Faben \cite{faben2008complexity}, who characterized the complexity of counting the solutions of a Generalized Satisfiability problem modulo an integer, and a generalization of \cite{faben2008complexity} to problems with weights by  Guo et al.\ \cite{guo_et_al:LIPIcs:2011:3015}. The study of modular counting of graph homomorphisms has been much more vibrant. 

Before discussing the existing results on modular counting and the results of this study, we need to mention some features of the \emph{automorphism group} of a relational structure. The automorphisms of $\cH$ form a group with respect to composition denoted $\Aut(\cH)$. The \emph{order} of an automorphism $\pi\in\Aut(\cH)$ is the smallest number $k$ such that $\pi^k$ is the identity permutation. An element $a\in H$ is a fixed point of $\pi\in\Aut(\cH)$ if $\pi(a)=a$. The set of all fixed points of $\pi$ is denoted by $\Fix(\pi)$.

A systematic study of counting homomorphisms in graphs was initiated by Faben and Jerrum in \cite{faben2008complexity}. They observed that the automorphism group $\Aut(\cH)$, particularly the automorphisms of order $p$, plays a crucial role in the complexity of the problem $\#_p\Hom(\cH)$. This insight extends to relational structures, as discussed in \cite{DBLP:conf/stoc/BulatovK22}. Specifically, for a homomorphism $\varphi$ from a relational structure $\cG$ to $\cH$, composing $\varphi$ with an automorphism from $\Aut(\cH)$ yields another homomorphism from $\cG$ to $\cH$. Thus, any automorphism of $\cH$ acts on the set $\Hom(\cG, \cH)$ of all homomorphisms from $\cG$ to $\cH$. If $\Aut(\cH)$ contains an automorphism $\pi$ of order $p$, the size of the orbit of $\varphi$ is divisible by $p$ unless $\pi \circ \varphi = \varphi$. In this case, the range of $\varphi$ lies within the set of fixed points $\Fix(\pi)$ of $\pi$, making this orbit contribute $0$ modulo $p$ to the total homomorphism count from $\cG$ to $\cH$. This observation motivates the following construction: let $\cH^\pi$ denote the substructure of $\cH$ induced by $\Fix(\pi)$. We denote by $\cH \rightarrow_p \cH'$ the existence of a $\pi \in \Aut(\cH)$ such that $\cH'$ is isomorphic to $\cH^\pi$. Furthermore, we write $\cH \rightarrow^*_p \cH'$ if there exist structures $\vc \cH k$ such that $\cH$ is isomorphic to $\cH_1$, $\cH'$ is isomorphic to $\cH_k$, and $\cH_1 \rightarrow_p \cH_2 \rightarrow_p \cdots \rightarrow_p \cH_k$.

Relational structures without order $p$ automorphisms will be called \emph{$p$-rigid}. 

\begin{lemma}[\cite{DBLP:conf/stoc/BulatovK22,ref:CountingMod2Ini}]\label{lem:aut-reduction}
Let $\cH$ be a relational structure and $p$ a prime. Then \\[1mm]
(a) Up to an isomorphism there exists a unique $p$-rigid structure $\cH^{*p}$ such that $\cH\rightarrow_p^* \cH^{*p}$.\\[1mm]
(b) For any relational structure $\cG$ it holds  that 
$
\hom(\cG,\cH)\equiv\hom(\cG,\cH^{*p})\pmod p.
$
\end{lemma}
By Lemma~\ref{lem:aut-reduction} it suffices to determine the complexity of $\NpCSP(\cH)$ for $p$-rigid structures $\cH$.

% existing results 
\vspace{3mm}

\noindent
\textbf{The existing results on modular counting.}
As we mentioned before, the research in modular counting $\CSP$s has mostly been aimed at counting graph homomorphisms. The complexity of the problem $\#_p\Hom(H)$ of counting homomorphism to a fixed graph $H$ modulo a prime number $p$ has received significant attention in the last ten years. Faben and Jerrum in \cite{ref:CountingMod2Ini} posed a conjecture that up to an order $p$ automorphism reduction $\rightarrow_p$ the complexity of this problem is the same as that for exact counting. More precisely, they conjectured that $\#_p\Hom(H)$ is solvable in polynomial time if and only if $\Hom(H^{*p})$ is. By the result of Dyer and Greenhill \cite{Dyer00:complexity} $\ghomk pH$ is solvable in polynomial time if and only if every connected component of $H^{*p}$ is a complete graph with all loops present or a complete bipartite graph. Therefore, proving that if a $p$-rigid $H$ does not satisfy these conditions then $\#_p\Hom(H)$ is $\#_pP$-hard suffices to confirm the conjecture of Faben and Jerrum. Over several years the hardness of $\#_p\Hom(H)$ was established for various graph classes \cite{ref:CountingMod2Ini,ref:CountingMod2ToCactus,ref:CountingMod2ToSquarefree,ref:CountingModPToTrees_gbel_et_al_LIPIcs, ref:CountingModPToSquarefree, Focke21:counting,ref:CountingModPToK33free}. Finally, it was proved for arbitrary graphs by Bulatov and Kazeminia \cite{DBLP:conf/stoc/BulatovK22}.

\begin{theorem}[\cite{DBLP:conf/stoc/BulatovK22}]\label{the:CountingGraphHom}
For any prime $p$ and any graph $H$ the problem $\ghomk pH$ is solvable in polynomial time if and only if $\Hom(H^{*p})$ is solvable in polynomial time. Otherwise it is \modp-complete. 
\end{theorem}

\vspace{3mm}

\noindent
\textbf{Our Results.}
In this paper we begin a systematic study of the problem $\#_p\CSP(\cH)$ for general relational structures $\cH$. Note that to the best of our knowledge, it is the first paper attempting at such a general modular counting problem. The ultimate goal is to obtain a complexity classification similar to Theorem~\ref{the:CountingGraphHom} for arbitrary relational structures. The contribution of the paper is twofold. First, we analyse the existing techniques and their applicability to the general case. We conclude that few of them work. More specifically, Theorem~\ref{the:CountingGraphHom} asserts that the complexity of modular counting for $p$-rigid graphs is the same as of exact counting. We, however, suggest a relational structure, a digraph $\cT_p$, that is $p$-rigid, its exact counting problem is hard, but modular counting is easy, see Section~\ref{sec:rigid-non-rigid}. Another important ingredient of the proof of Theorem~\ref{the:CountingGraphHom} is a structural theorem on automorphisms of products of graphs \cite{ref:ProductOfGraphs}. No such result exists for products of relational structures. Moreover, in Section~\ref{sec:poduct-auto} we suggest an example (again, a digraph) showing that nothing similar to such a structural result can be true. Some of the standard techniques in counting CSPs involve properties of relations and relational structures such as rectangularity, permutability, balancedness, the presence of a Mal'tsev polymorphism. In the case of exact counting these concepts are closely related to each other and make efficient algorithms possible. Unfortunately, these concepts are of little help for modular counting. We introduce their modular equivalents, but then a series of examples show that no tight connections are possible in this case, see Section~\ref{sec:rectangularity-main}. This makes algorithm design very difficult.

On the positive side, we obtain some results to somewhat remedy the situation. The first step is to convert the problem into a richer framework of multi-sorted relational structures and CSPs. The main idea is, given a CSP instance $\cG$, to try to identify the possible images of each vertex of $\cG$, and then treat vertices with different ranges as having different types and map them to different disjoint domains. In Section~\ref{sec:refinement} we call this process refinement and propose two types of refinement, one is based on local propagation techniques, and the other on solving the decision version of the problem. The main benefit of using multi-sorted structures and CSPs is the richer structure of their automorphisms. This often allows for stronger reductions than single-sorted structures do. In particular, if the digraph $\cT_p$ mentioned above is subjected to this process, it results in a multi-sorted structure that is no longer $p$-rigid, the corresponding reduced structure is very simple and can easily be solved. We are not aware of any example of a structure whose multi-sorted refinement is $p$-rigid, but that would give rise to an easy modular counting problem.

In the next line of research we follow the approach of  \cite{DBLP:conf/stoc/BulatovK22} to expand the relational structure $\cH$ by adding relations to $\cH$ that are primitive positive (pp-)definable in $\cH$, that is, relations that can be derived from the relations of $\cH$ using equality, conjunction and existential quantifiers. However, expanding the general relational structure by pp-definable relations does not work as well as for graphs. To overcome this obstacle, we introduce a new form of expansion which uses $p$-modular quantifiers instead of regular existential quantifiers. The semantics of a \emph{$p$-modular quantifier} is "there are non-zero modulo $p$ values of a variable" rather than "there exists at least one value of a variable" as the regular existential quantifier asserts. Every relational structure is associated with a \emph{relational clone} $\SP\cH$ that consists of all relations pp-definable in $\cH$.  The new concept gives rise to new definitions of pp-formulas and relational clones. If regular existential quantifiers in pp-formulas are replaced with $p$-modular quantifiers, we obtain $p$-modular primitive positive formulas ($p$-mpp-formulas, for short). Then, similar to pp-definitions, a relation $\rel$ is said to be \emph{$p$-mpp-definable} in a structure $\cH$ if there is a $p$-mpp-formula in $\cH$ expressing $\rel$. The $p$-modular clone $\SP\cH_p$ associated with $\cH$ is the set of all relations $p$-mpp-definable in $\cH$. We show in Theorem~\ref{pro:GadgetExists} (see also Theorem~\ref{the:ConstantCSP}) that, similar to the result of Bulatov and Dalmau \cite{ref:BULATOV_TowardDichotomy}, expanding a structure by a  $p$-mpp-definable relation does not change the complexity of the problem $\#_p\CSP(\cH)$.

\begin{theorem}\label{the:CloneTheorem-intro}
Let $p$ be a prime number and $\cH$ a $p$-rigid relational structure. If $\rel$ is $p$-mpp-definable in $\cH$, then $\#_p\CSP(\cH+\rel)$ is polynomial time reducible to $\#_p\CSP(\cH)$.
\end{theorem}

In the remaining part of the paper we identify a number of conditions under which it is possible to design an algorithm or to prove the hardness of the problem. One such case is $\#_2\CSP(\cH)$ when $\cH$ satisfies both 2-rectangularity and the usual strong rectangularity conditions (or, equivalently, has a Mal'tsev polymorphism). It starts with applying the known techniques \cite{Bulatov06:simple,effective-Dyer-doi:10.1137/100811258} to find a concise representation or a frame of the set of solutions of a given  CSP. However, such a representation cannot be used directly to find the parity of the number of solutions. The algorithm performs multiple recomputations of the frame to exclude the parts that produce an even number of solutions. Unfortunately, this algorithm does not generalize to larger $p$ even under very strong assumptions, because the structure of finite fields start playing a role.

While studying the structure of automorphisms of products of relational structures such as $\Aut(\cH^n)$ may be a difficult problem, in Section~\ref{sec:binarization} we make a step forward by reducing the class of structures $\cH$ for which such  structural results are required. More precisely, for any relational structure $\cH=(H;\vc\rel k)$ we construct its \emph{binarization} $b(\cH)$ whose relations are binary and rectangular. This makes such structures somewhat closer to graphs and the hope is that it will be easier to study the structure of $\Aut(b(\cH)^n)$ than $\Aut(\cH^n)$ itself. We prove that $\cH$ and $b(\cH)$ share many important properties.

\begin{theorem}
Let $\cH$ be a relational structure. Then $\cH$ is strongly rectangular ($p$-strongly rectangular, $p$-rigid, has a Mal'tsev polymorphism) if and only if $b(\cH)$ is strongly rectangular ($p$-strongly rectangular, $p$-rigid, has a Mal'tsev polymorphism).
\end{theorem}

In Section~\ref{sec:automorphisms} we explore what implications of a structural results about $\Aut(\cH^n)$ can be. We show that assuming certain such a result has similar consequences as those in the case of graphs, they may be sufficient to prove the hardness of $\#_p\CSP(\cH)$.

We provide a more detailed overview of the main results in Section~\ref{sec:overview}.

%%%%%%%%%%%%%%%%%%%%%%%%%%%%%%%%%%%%%%%
%%%%%%%%%%%%%%%%%%%%%%%%%%%%%%%%%%%%%%%
\section{Preliminaries}\label{sec:preliminaries}

Let $[n]$ denote the set $\{1, 2, \ldots, n\}$. Let $H^n$ be the Cartesian product of the set $H$ with itself $n$ times and $H_1\tms H_k$ the Cartesian product of sets $\vc Hk$. We denote the members of $H^n$ and $H_1\tms H_k$ using bold font, $\mathbf{a} \in H^n$, $\ba\in H_1\tms H_k$. The $i$-th element of $\mathbf{a}$ is denoted by $\mathbf{a}[i]$. 

For $I=\{\vc ik\}\sse[n]$ we use $\pr_I\ba$ to denote $(\ba[i_1]\zd \ba[i_k])$; and for $\rel\sse H_1\tms H_k$ we use $\pr_I\rel$ to denote $\{\pr_I\ba\mid\ba\in\rel\}$. 
For $\ba\in H^s$ by $\ext_\rel(\ba)$ we denote the number of assignments $\bb\in H^{n-s}$ such that $(\ba,\bb)\in\rel$. (Note that the order of elements in $\ba$ and $\bb$ and $\rel$ might differ, hence we slightly abuse the notation here.) We denote the number of assignments mod $p$ by $\extp_\rel(\ba)$.
Moreover, $\prp_{I}\rel$ denotes the set $\{\pr_{I} \ba\mid\ba \in \rel \text{ and } \extp_{\rel}(\pr_{I} \ba) \neq 0  \}$. Often, we treat relations $\rel\sse H_1\tms H_k$ as predicates $\rel: H_1\tms H_k\to \{0,1\}$.

%%%%%%%%%%%%%%%%%%%%%%%%%%%%%%%%%%%
\subsection{Multi-Sorted sets and relational structures.}\label{sec:multi-sorted}

We begin with introducing \emph{multi-sorted} or \emph{typed} sets. Let $H =\colect{H_i}{i\in [k]}= \{ H_1\zd H_k \}$ be a collection of sets. We will assume that the sets $\vc Hk$ are disjoint. 
The cardinality of a multi-sorted set $H$ equals $|H| = \sum_{i\in[k]} |H_i|$.
A mapping $\vf$ between two multi-sorted sets $G= \colect{G_i}{i\in[k]}$ and $H= \colect{H_i}{i\in[k]}$ is defined as a collection of mappings $\ovarrow\vf=\colect{\vf_i}{i\in[k]}$, where $\vf_i:G_i\to H_i$, that is, $\vf_i$ maps elements of the sort $i$ in $G$ to elements of the sort $i$ in $H$. 
A mapping $\ovarrow\vf = \colect{\vf_i}{i\in[k]}$ from $\colect{G_i}{i\in[k]}$ to $\colect{H_i}{i\in [k]}$ is injective (bijective), if for all $i\in [k]$, $\vf_i$ is injective (bijective).

A \emph{multi-sorted relational signature} $\sg$ over a set of types $\{1\zd k\}$ is a collection of \emph{relational symbols}, a symbol $\rel\in\sg$ is assigned a positive integer $\ell_\rel$, the \emph{arity} of the symbol, and a tuple $(\vc i{\ell_\rel})$, the \emph{type} of the symbol. A \emph{multi-sorted relational structure} $\ovarrow\cH$ with signature $\sg$ is a multi-sorted set $\{H_i\}_{i \in [k]}$ and an \emph{interpretation} $\rel^{\ovarrow\cH}$ of each $\rel\in\sg$, where $\rel^{\ovarrow\cH}$ is a relation or a predicate on $H_{i_1} \times ... \times H_{i_{\ell_\rel}}$. 
The multi-sorted structure $\ovarrow\cH$ is finite if $H$ and $\sg$ are finite. All structures in this paper are finite. The set $H$ is said to be the \emph{base set} or the \emph{universe} of $\ovarrow\cH$. For the base set we will use the same letter as for the structure, only in the regular font. 
Multi-sorted structures with the same signature and type are called \emph{similar}.

Let $\ovarrow\cG,\ovarrow\cH$ be similar multi-sorted structures with signature $\sg$. A \emph{homomorphism} $\ovarrow\vf$ from $\ovarrow\cG$ to $\ovarrow\cH$ is a collection of mappings $\vf_i:G_i\to H_i$, $i\in[k]$, from $G$ to $H$ such that for any $\rel\in\sg$ with type $(\vc i{\ell_\rel})$, if $\rel^{\ovarrow\cG}(\vc a{\ell_\rel})$ is true for $(\vc a{\ell_\rel})\in G_{i_1} \times ... \times G_{i_{\ell_\rel}}$, then $\rel^{\ovarrow\cH}(\vf_{i_1}(a_1)\zd\vf_{i_{\ell_\rel}}(a_{\ell_\rel}))$ is true as well. 
The set of all homomorphisms from $\ovarrow\cG$ to $\ovarrow\cH$ is denoted $\Hom(\ovarrow\cG,\ovarrow\cH)$. The cardinality of $\Hom(\ovarrow\cG,\ovarrow\cH)$ is denoted by $\hom(\ovarrow\cG,\ovarrow\cH)$. For a multi-sorted structure $\ovarrow\cH$, the corresponding \emph{counting CSP}, $\#_p\CSP(\ovarrow\cH)$, is the problem of computing $\hom(\ovarrow\cG,\ovarrow\cH)$ for a given structure $\ovarrow\cG$. A homomorphism $\ovarrow\vf$ is an \emph{isomorphism} if it is bijective and the inverse mapping $\ovarrow\vf^{-1}$ is a homomorphism from $\ovarrow\cH$ to $\ovarrow\cG$. If $\ovarrow\cH$ and $\ovarrow\cG$ are isomorphic, we write $\ovarrow\cH\cong\ovarrow\cG$. A homomorphism of a structure to itself is called an \emph{endomorphism}, and an isomorphism to itself is called an \emph{automorphism}. As is easily seen, automorphisms of a structure $\ovarrow\cH$ form a group denoted $\Aut(\ovarrow\cH)$.

The \emph{direct product} of multi-sorted $\sg$-structures $\ovarrow\cH,\ovarrow\cG$, denoted $\ovarrow\cH\tm\ovarrow\cG$ is the multi-sorted $\sg$-structure with the base set $H\tm G=\colect{H_i \times G_i}{i\in[k]}$, the interpretation of $\rel\in\sg$ is given by $\rel^{\ovarrow\cH\tm\ovarrow\cG}((a_1,b_1)\zd(a_k,b_k))=1$ if and only if $\rel^{\ovarrow\cH}(\vc ak)=1$ and $\rel^{\ovarrow\cG}(\vc bk)=1$. By $\ovarrow\cH^\ell$ we will denote the \emph{$\ell$th power} of $\ovarrow\cH$, that is, the direct product of $\ell$ copies of $\ovarrow\cH$. 

Let $\ovarrow\cH=(\colect{H_i}{i\in[k]};\vc\rel m)$ be a multi-sorted relational structure and $\ovarrow\pi$ an automorphism of $\ovarrow\cH$. 
For $\relo\sse H_{i_1}\tm\dots\tm H_{i_\ell}$, we use $\ovarrow\pi(\relo)$ to denote the set $\{\ovarrow\pi(\ba)\mid \ba\in\relo\}$, where $\ovarrow\pi(\ba)=(\pi_{i_1}(\ba[1])\zd\pi_{i_\ell}(\ba[\ell]))$.
For a prime number $p$ we say that $\pi$ has order $p$ if $\pi$ is not the identity in $\Aut(\cH)$ and has order $p$ in $\Aut(\cH)$. In other words, each of the $\pi_j$'s is either the identity mapping or has order $p$, and at least one of the $\pi_j$'s is not the identity mapping. Structure $\ovarrow\cH$ is said to be \emph{$p$-rigid} if it has no automorphism of order $p$. Similar to regular relational structures we can introduce reductions of multi-sorted structures by their automorphisms of order $p$. 

A \emph{substructure} $\ovarrow\cH'$ of $\ovarrow\cH$ \emph{induced} by a collection of sets $\{H'_i\}_{i\in[k]}$, where $H'_i\sse H_i$ is the relational structure given by $(\colect{H'_i}{i\in[k]};\vc{\rel'}m)$, where 
%% if $\rel_j\sse H_{i_1}\tm\dots\tm H_{i_\ell}$, then 
$\rel'_j=\rel_j\cap(H'_{i_1}\tm\dots\tm H'_{i_\ell})$ and $(\vc i\ell)$ is the type of $\rel_j$. 
By $\Fix(\ovarrow\pi)$ we denote the collection $\{\Fix(\pi_i) \}_{j\in[k]}$ of sets of fixed points of the $\pi_i$'s. Let $\ovarrow\cH^{\ovarrow\pi}$ denote the substructure of $\ovarrow\cH$ induced by $\Fix(\ovarrow\pi)$. We write $\ovarrow\cH\rightarrow_p \ovarrow\cH'$ if there is $\ovarrow\pi\in\Aut(\ovarrow\cH)$ of order $p$ such that $\ovarrow\cH'$ is isomorphic to $\ovarrow\cH^{\ovarrow\pi}$. We also write $\ovarrow\cH\rightarrow^*_p \ovarrow\cH'$ if there are structures $\vc {\ovarrow\cH} {t}$ such that $\ovarrow\cH$ is isomorphic to $\ovarrow\cH_1$, $\ovarrow\cH'$ is isomorphic to $\ovarrow\cH_t$, and $\ovarrow\cH_1\rightarrow_p \ovarrow\cH_2\rightarrow_p\dots\rightarrow_p \ovarrow\cH_t$. Let also $\ovarrow\cH^{*p}$ be a $p$-rigid structure such that $\ovarrow\cH\rightarrow^*_p \ovarrow\cH^{*p}$. 

\begin{proposition}\label{pro:mult-uniqueness-main}
Let $\ovarrow\cH$ be a multi-sorted  structure and $p$ a prime. Then up to an isomorphism there exists a unique $p$-rigid multi-sorted  structure $\ovarrow\cH^{*p}$ such that $\ovarrow\cH\rightarrow_p^* \ovarrow\cH^{*p}$.
\end{proposition}

The proof of Proposition~\ref{pro:mult-uniqueness-main} follows the same lines as the analogous statement in the single-sorted case, and is moved to Appendix~\ref{sec:appendix-new}.

The following statements are well-known for single-sorted structures, see \cite{ref:CountingMod2Ini,DBLP:conf/stoc/BulatovK22}. 

\begin{lemma}\label{lem:aut-reduction-multi-sorted-structures-prelim} 
Let $\ovarrow\cH$ be a multi-sorted relational structure. If $\ovarrow\pi$ is a $p$-automorphism of $\ovarrow\cH$, then for any $\ovarrow\cG$,
\[
\hom(\ovarrow\cG,\ovarrow\cH)\equiv \hom(\ovarrow\cG, \ovarrow\cH^{\ovarrow\pi}) \pmod{p}.
\]
\end{lemma}

\begin{proof}
Let $H = \colect{H}{i\in[k]}$ and $H^{\ovarrow\pi} = \colect{H_i^{\pi_i}}{i\in[k]}$ denote the universes of $\ovarrow\cH$ and $\ovarrow\cH^{\ovarrow\pi}$ respectively. For a similar multi-sorted structure $\cG$ with universe $G$, we show that the number of homomoprhisms which use at least one element of $H- H^{\ovarrow\pi}$ is $0\pmod p$. 

Given any homomorphism $\FUNC{\ovarrow\vf}{\ovarrow\cG}{\ovarrow\cH}$, where $\vf = \{\vf_i\}_{i\in [k]}$, consider the  homomorphism  $\ovarrow{\pi \circ \vf} = \{\pi_i \circ \vf_i\}_{i \in [k]}$. This is still a homomorphism  which is different from $\ovarrow\vf$ as there is some element $v \in G_i$ such that $\vf_i(v) \in H_i- H_i^{\pi_i}$ for all $i \in [k]$, and so $\pi_i(\vf_i(v)) \neq \vf_i(v)$. On the other hand $\pi_i^p \circ \vf_i$ is just $\vf_i$, as $\pi_i$ is a $p$-automorphism. So $\ovarrow\pi$ acts as a permutation of order $p$ on the set $\Hom(\ovarrow\cG, \ovarrow\cH)$. Moreover, the orbit of $\ovarrow\pi$ containing a homomorphism that has at least one element from $H_i- H_i^\pi$ in its range has size 0 modulo $p$. 
\end{proof}

The following statement is a straightforward implication of Lemma~\ref{lem:aut-reduction-multi-sorted-structures-prelim}.

\begin{proposition}\label{pro:mult-aut-reduction}
Let $\ovarrow\cH$ be a multi-sorted relational structure and $p$ a prime. Then for any relational structure $\cG$ it holds that
\[
\hom(\cG,\cH)\equiv\hom(\cG,\cH^{*p})\pmod p.
\]
\end{proposition}

We complete this section with a definition of polymorphisms. Let $\rel\sse H^n$ be a relation over a set $H$. A $k$-ary \emph{polymorphism} of $\rel$ is a mapping $f:H^k\to H$ such that for any choice of $\ba_1\zd\ba_k\in \rel$, it holds that $f(\ba_1\zd\ba_k) \in \rel$ (computed component-wise). The mapping $f$ is a polymorphism of a (single-sorted) relational structure $\cH=(H,\vc\rel m)$ if it is a polymorphism of each relation $\vc\rel m$. In the multi-sorted case the definitions are a bit more complicated. Let $\rel\sse H_1\tms H_n$ be a multi-sorted relation. Instead of a single mapping we consider a family of mappings $f=\{f_i\}_{i\in[n]}$, $f_i:H^k_i\to H_i$. The family $f$ is said to be a polymorphism of $\rel$ if it satisfies the same condition: for any choice of $\ba_1\zd\ba_k\in \rel$, it holds that $f(\ba_1\zd\ba_k) \in \rel$, only in this case the mapping $f_i$ is applied in the $i$th coordinate position, $i\in[n]$.  A polymorphism of a multi-sorted structure $\ovarrow\cH=(\{H_i\}_{i\in[q]};\vc\relo m)$ is again a family $f=\{f_i\}_{i\in q}:H^k_i\to H_i$ such that for each $j\in[m]$, $f$ (or rather its appropriate subfamily) is a polymorphism of $\relo_j$. For a complete introduction into polymorphisms the reader is referred to \cite{ref:POlymorphismAndUsethem_barto2017polymorphisms}.

The following special type of polymorphisms often occurs in the study of counting CSPs. For a set $H$ a mapping $\FUNC{\vf}{H^3}{H}$ is said to be a \emph{Mal'tsev} operation if for any $a,b\in H$ it satisfies the conditions $f(a,a,b)=f(b,a,a)=b$. In the multi-sorted case a family $f=\{f_i\}_{i\in[n]}$ is Mal'tsev, if every $f_i$ is Mal'tsev. A Mal'tsev operation (or a family of operations) that is a polymorphism of a relation $\rel$ or a relational structure $\cH$ is said to be a Mal'tsev polymorphism of $\rel$ or $\cH$. One of the useful properties of Mal'tsev polymorphisms is that it enforces the rectangularity of relations. A binary relation $\rel\sse H^2$ is rectangular if whenever $((a,c),(a,d),(b,c)\in\rel$ it also holds that $(b,d)\in\rel$. If $\rel$ has a Mal'tsev polymorphism $f$, it is rectangular. Indeed,
\[
f\left(\begin{array}{ccc} a&a&b\\ d&c&c \end{array}\right)=\left(\begin{array}{c} f(a,a,b)\\f(d,c,c) \end{array}\right) = \left(\begin{array}{c} b\\  d \end{array}\right)\in\rel.
\]
See Section~\ref{sec:rectangularity-main} for further details.

%%%%%%%%%%%%%%%%%%%%%%
\subsection{The two views on the CSP}\label{sec:two-views}

In the Introduction, we defined the CSP as the problem of deciding the existence of a homomorphism between two relational structures. From a technical perspective, however, a more traditional approach is often more useful. 

Let $H$ be a (finite) set and $\Gm$ a set of relations on $H$. Such a set of relations is called a \emph{constraint language}, and the set $H$ is often referred to as the \emph{domain}. 

The Constraint Satisfaction Problem $\CSP(\Gm)$ is the combinatorial problem with:\\[3mm]
\textbf{\emph{Instance:}} a pair $\cP=(V,\cC)$ where $V$ is a finite set of variables and $\cC$ is a finite set of \emph{constraints}. Each constraint $C \in \cC$ is a pair $\ang{\bs,\rel}$ where 
\begin{itemize}
    \item $\bs = (v_1, v_2,... , v_m )$ is a tuple of variables from $V$ of length $m$, called the \emph{constraint scope};
    \item $\rel\in\Gm$ is an $m$-ary relation, called the \emph{constraint relation}.
\end{itemize}
\textbf{\emph{Objective:}} Decide whether there is a solution of $\cP$, that is, a mapping $\vf:V\to H$ such that for each constraint $\ang{\bs,\rel}\in \cC$ with $\bs = (\vc vm)$ the tuple $(\vf(v_1)\zd\vf(v_m))$ belongs to~$\rel$. 

\smallskip

In the (modular) counting version of $\CSP(\Gm)$ denoted $\NCSP(\Gm)$ ($\NpCSP(\Gm)$) the objective is to find the number of solutions of instance $\cP$ (modulo $p$). We will refer to this definition as the \emph{standard} definition of the CSP.

To relate this definition to the homomorphism definition of the CSP, note that for a $\sg$-structure $\cH$ the collection of interpretations $\rel^\cH$, $\rel\in\sg$, is just a set of relations, that is a constraint language over $H$. Thus, for every relational structure $\cH$ there is an associated constraint language $\Gm_\cH$. Conversely, every (finite) constraint language $\Gm$ can be converted into a relational structure $\cH_\Gm$ such that $\Gm_{\cH_\Gm}=\Gm$ in a straightforward way, although in this case there is much room for the choice of a signature. A language $\Gm$ is said to be \emph{rigid} or \emph{$p$-rigid}, if $\cH_\Gm$ is rigid or $p$-rigid.

It is well known, see e.g.\ \cite{Feder98:computational} and \cite{ref:POlymorphismAndUsethem_barto2017polymorphisms} that the problems $\CSP(\cH)$ and $\CSP(\Gm_\cH)$ can be easily translated into each other. The same is true for $\NCSP(\cH)$ and $\NCSP(\Gm_\cH)$. The conversion procedure goes as follows. Let $\cG$ be an instance of $\CSP(\cH)$ and $\sg$ is the signature of $\cH$. Create an instance $\cP=(V,\cC)$ of $\CSP(\Gm_\cH)$ by setting $V=G$, the base set of $\cG$, and for every $\rel\in\sg$ and every $\bs\in\rel^\cG$, include the constraint $\ang{\bs,\rel^\cH}$ into $\cC$. The transformation of instances of $\CSP(\Gm)$ to an instance of $\CSP(\cH_\Gm)$ is the reverse of the procedure above. This transformation is clearly extended to counting CSPs, as well. In this paper we often use the standard definition of the CSP inside proofs assuming that an instance of $\NpCSP(\cH)$ is given by variables and constraints.

The definition above can be easily extended to multi-sorted relational structures and languages. If $\ovarrow\cH$ is a multi-sorted relational structure, $\CSP(\cH)$ ($\#\CSP(\cH)$) asks, given a similar structure $\cG$, whether there exists a homomorphism of multi-sorted structure $\cG$ to $\cH$ (asks to find the number of such homomorphisms). If $\Gm$ is a set of multi-sorted relations, every variable $v\in V$ for an instance $(V,\cC)$ of $\CSP(\Gamma)$ is assigned a type $\tau(v)$ and for every constraint $\ang{(\vc vk),\rel}$, the sequence $(\tau(v_1)\zd\tau(v_k))$ must match the type of $\rel$.

%%%%%%%%%%%%%%%%%%%%%%%%%%%%%%%%%%%%%%%%%%%%%%%%%%%%%%%%%%%%%%%%%
\subsection{Expansion of relational structures}\label{sec:pp-defs}

One of the standard techniques when studying constraint problems is to identify ways to expand the target relational structure or constraint language with additional relations without changing the complexity of the problem.

Let $\cH$ be a relational structure with signature $\sg$ and $\cH^=$ its expansion by adding a binary relational symbol $=$ interpreted as $=_H$, the equality relation on $H$. The following reduction is straightforward.

\begin{lemma}[\cite{DBLP:conf/stoc/BulatovK22}]\label{lem:adding-equality}
For any relational structure $\cH$ and any prime $p$, $\NpCSP(\cH^=)\le_T\NpCSP(\cH)$. Similarly, for a constraint language $\Gm$ on $H$ the problem $\NpCSP(\Gm\cup\{=_H\})$ is polynomial time reducible to $\NpCSP(\Gm)$.
\end{lemma}

A constant relation over a set $H$ is a unary relation $C_a=\{a\}$, $a\in H$. For a relational structure $\cH$ by $\cH^\const$ we denote the expansion of $\cH$ by all the constant relations $C_a$, $a\in H$. Theorem~\ref{the:ConstantCSP} was proved for exact counting in \cite{ref:BULATOV_TowardDichotomy}, for modular counting of graph homomorphisms in \cite{ref:CountingMod2Ini, ref:CountingMod2ToSquarefree, ref:CountingModPToTrees_gbel_et_al_LIPIcs}, and for general modular counting CSP in \cite{DBLP:conf/stoc/BulatovK22}. 

\begin{theorem}[\cite{DBLP:conf/stoc/BulatovK22}]\label{the:ConstantCSP}
Let $\cH$ be a $p$-rigid $\sg$-structure. Then $\NpCSP(\cH^\const)$ is polynomial time reducible to $\NpCSP(\cH)$. 

For a $p$-rigid constraint language $\Gm$ on a set $H$ the problem $\NpCSP(\Gm\cup\{C_a\mid a\in H\})$ is polynomial time reducible to $\NpCSP(\Gm)$.
\end{theorem}

Lemma~\ref{lem:adding-equality} and Theorem~\ref{the:ConstantCSP} can be generalized to the multi-sorted case. Let $\ovarrow\cH=\{\cH_i\}_{i\in[k]}$ be a multi-sorted structure with signature $\sg$ and $\ovarrow\cH^=$ its expansion by adding a family of binary relational symbols $=_{H_i}$ (one for each type) interpreted as the equality relation on $H_i$, $i\in[k]$. 

A constant relation over a set $\{H_i\}_{i\in[k]}$ is a unary relation $C_a=\{a\}$, $a\in H_i,i\in[k]$ (such a predicate can only be applied to variables of type $i$). For a structure $\ovarrow\cH$ by $\ovarrow\cH^\const$ we denote the expansion of $\ovarrow\cH$ by all the constant relations $C_a$, $a\in H_i,i\in[k]$. 

\begin{theorem}[see also \cite{DBLP:conf/stoc/BulatovK22}]\label{the:ConstantCSP-main}
Let $\cH$ be a multi-sorted relational structure and $p$ prime.\\[1mm]
(1) $\NpCSP(\ovarrow\cH^=)$ is Turing reducible to $\NpCSP(\ovarrow\cH)$;\\[1mm]
(2) Let $\ovarrow\cH$ be $p$-rigid. Then $\NpCSP(\ovarrow\cH^\const)$ is Turing reducible to $\NpCSP(\ovarrow\cH)$. 
\end{theorem}

Yet another way to expand a relational structure is by primitive positive definable (pp-definable for short) relations. Primitive-positive definitions have played a major role in the study of the CSP. It has been proved in multiple circumstances that expanding a structure with pp-definable relations does not change the complexity of the corresponding CSP. This has been proved for the decision CSP in \cite{Jeavons:algebraic,Bulatov05:classifying} and the exact Counting CSP \cite{ref:BULATOV_TowardDichotomy}. The reader is referred to \cite{ref:POlymorphismAndUsethem_barto2017polymorphisms} for details about pp-definitions and their use in the study of the CSP.

Conjunctive definitions are a special case of primitive positive definitions that do not use quantifiers. Let $\cH$ be a structure with signature $\sg$. A \emph{conjunctive formula} $\Phi$ over variables $\{\vc xk\}$ is a conjunction of atomic formulas of the form $\rel(\vc y\ell)$, where $\rel\in\sg$ is an ($\ell$-ary) symbol and $\vc y\ell\in\{\vc xk\}$. A $k$-ary predicate $\relo$ is conjunctive definable in $\cH$ if $(\vc ak)\in\relo$ if and only if $\Phi(\vc ak)$ is true.

\begin{lemma}[\cite{DBLP:conf/stoc/BulatovK22}]\label{lem:conjunctive}
Let $\cH$ be a relational structure with signature $\sg$, $\rel$ be conjunctive definable in $\cH$, and $\cH+\rel$ denote the expansion of $\cH$ by a predicate symbol $\rel$ that is interpreted as the relation $\rel$ in $\cH$. Then $\NpCSP(\cH+\rel)\le_T\NpCSP(\cH)$.

Similarly, for a constraint language $\Gm$ if a relation $\rel$ is conjunctive definable in $\Gm$, then $\NpCSP(\Gm\cup\{\rel\})\le_T\NpCSP(\Gm)$.
\end{lemma}

We now extend the concept of primitive-positive definability to the multi-sorted case.
Let $\ovarrow\cH$ be a multi-sorted relational structure with the base set $H$. As before \emph{primitive positive} (pp-) formula in $\ovarrow\cH$ is a first-order formula 
\[
\exists \vc ys\Phi(\vc xk,\vc ys),
\]
where $\Phi$ is a conjunction of atomic formulas of the form $z_1=_Hz_2$ or $\rel(\vc z\ell)$, $\vc z\ell\in\{\vc xk,\vc ys\}$, and $\rel$ is a predicate of $\ovarrow\cH$. However, every variable in $\Phi$ is now assigned a type $\tau(x_i),\tau(y_j)$ in such a way that for every atomic formula $z_1=_Hz_2$ it holds that $\tau(z_1)=\tau(z_2)$, and for any atomic formula $\rel(\vc z\ell)$ the sequence $(\tau(z_1)\zd\tau(z_\ell))$ matches the type of $\rel$. We say that $\ovarrow\cH$ \emph{pp-defines} a predicate $\relo$ if there exists a pp-formula such that \
\[
\relo(\vc xk)=\exists \vc ys\Phi(\vc xk,\vc ys).
\]

For $\ba\in\rel$ by $\ext_\Phi(\ba)$ we denote the number of assignments $\bb\in H_{\tau(y_1)}\tms H_{\tau(y_s)}$ to $\vc ys$ such that $\Phi(\ba,\bb)$ is true. We denote the number of such assignments mod $p$ by $\extp_\Phi(\ba)$.

While when $\cH$ is a graph it is possible to prove a statement similar to Lemma~\ref{lem:conjunctive} for pp-definable relations \cite{DBLP:conf/stoc/BulatovK22}, we will later see that it is unlikely to be true for general relational structures.

Finally, for a relational structure $\cH$ (single- or multi-sorted) $\SP\cH$ denotes the relational clone of $\cH$, that is, the set of all relations pp-definable in $\cH$

%%%%%%%%%%%%%%%%%%%%%%%%%%%%%%%%%%%%%%%
%%%%%%%%%%%%%%%%%%%%%%%%%%%%%%%%%%%%%%%

\section{Overview of the Results}\label{sec:overview}
The results of this paper can be grouped in two categories. The results from the first category analyse the methods used to obtain a classification of the complexity of modular counting CSPs over graphs, Theorem~\ref{the:CountingGraphHom}, \cite{DBLP:conf/stoc/BulatovK22}, and the methods used in exact counting \cite{DBLP:journals/jacm/Bulatov13,effective-Dyer-doi:10.1137/100811258,cai-complex} trying to evaluate which of them are applicable for modular counting CSPs over general relational structures. Unfortunately, the majority of those methods fail, as we show by constructing a series of counterexamples. The second category comprises a number of results aiming at refining the framework, modifying the existing methods, and proposing new ones in order to overcome those difficulties. We now discuss the results in some detail.

%%%%%%%%%%%%%%%%%%%%%%%%%%%%%%%%%%%%%%%%%%%%%%%%%%%%%
\subsection{The Failure}\label{sec:failure}

\textbf{$p$-rigidity.}
The classification result in Theorem~\ref{the:CountingGraphHom} asserts that $\ghomk pH$ for a $p$-rigid graph $H$ is hard whenever the exact counting problem is hard. In Example~\ref{exa:bad-triangle} we show that this is no longer the case for general relational structures. Let $p$ be a prime number and $T_p=\{0\zd p-1,p,p+1\}$. A relational structure $\cT_p$ has the base set $T_p$ and predicates $\rel,C_0\zd C_{p+1}$, where $C_a$ is the constant relation $\{(a)\}$ and $\rel=T^2_p-\{(0,p)\zd(p-1,p)\}$. The structure $\cT_p$ is $p$-rigid as it contains all the constant relations and every automorphism must preserve them, implying that every element of $T_p$ is a fixed point. By \cite{Bulatov17:dichotomy,Zhuk20:dichotomy} the decision $\CSP(\cT_p)$ is solvable in polynomial time, while the exact counting problem $\#\CSP(\cT_p)$ is $\#$P-complete by \cite{ref:BULATOV_TowardDichotomy}, as it does not have a Mal'tsev polymorphism and $\rel$ is not rectangular (see below). However, if $\cG$ is a structure similar to $\cT_p$ and there is a homomorphism $\vf:\cG\to\cT_p$ such that some vertex $v$ of $\cG$ is mapped to $a\in\{0\zd p-1\}$ then, unless $v$ is bound by $C_a$, the mapping that differs from $\vf$ only at $v$ by sending it to any $b\in\{0\zd p-1\}$ is also a homomorphism. Since $|\{0\zd p-1\}|=p$, this means that the elements $0\zd p-1$ can be effectively eliminated from $\cT_p$, and the resulting structure is somewhat trivial. Therefore $\NpCSP(\cT_p)$ can be solved in polynomial time. 

\smallskip

\textbf{Automorphisms of direct products.}
Another crucial component of Theorem~\ref{the:CountingGraphHom} is the structure of automorphisms of direct products of graphs \cite{ref:ProductOfGraphs}. It essentially asserts that every automorphism of a direct product $H_1\tms H_n$ can be thought of as a composition of a permutation of factors in the product and automorphisms of those factors. In Example~\ref{exa:bad-graph} we show that this breaks down already for digraphs. Indeed, let $\cH = (V, E)$ be a directed graph where $V = \{a, b,c ,d \}$ with (directed) edge set $E = \{ (b, a), (b,c), (c,d) \}$. This digraph is rigid. However, the automorphism group of $\cH^2$, see Figure~\ref{fig:bad-example-overview}, has a complicated structure. As is easily seen $\cH^2$ has a large number of automorphisms of order 2 and 3, not all of which have a simple representation mentioned above. 

In \cite{DBLP:conf/stoc/BulatovK22} one of the important applications of the structural theorem for automorphisms of graph products is that it allows one to prove that $\ghomk p{\cH+\rel)}$, where $\cH+\rel$ denotes the expansion of $\cH$ by a relation $\rel$ pp-definable in $\cH$, is polynomial time reducible to $\ghomk p\cH$ (see below). The example above indicates that this result may no longer be true for general relational structures.

\begin{figure}[ht]
    \centering
    \includegraphics[height=3.5cm]{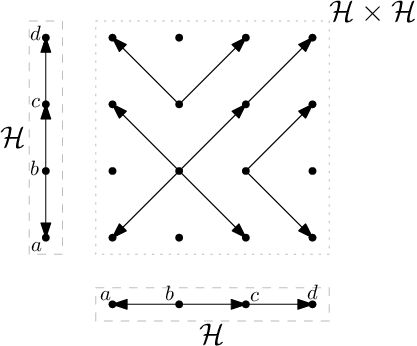}
    \caption{The structure of $\cH$ and $\cH^2$}
    \label{fig:bad-example-overview}
\end{figure} 

\smallskip

\textbf{Rectangularity, permutability, and Mal'tsev polymorphisms.}
The key properties of relational structures heavily exploited in \cite{DBLP:journals/jacm/Bulatov13,effective-Dyer-doi:10.1137/100811258,cai-complex} to obtain characterizations of the complexity of exact counting are rectangularity, balancedness, and the presence of a Mal'tsev polymorphism. 

Recall that a binary relation $\rel\subseteq H_1 \times H_2$ is said to be \emph{rectangular} if $(a, c),(a, d),(b, c) \in\rel$ implies $(b, d) \in\rel$ for any $a, b \in H_1, c, d \in H_2$. A relation $\rel\subseteq H^n$ for $n \geq 2$ is rectangular if for every $I\subsetneq[n]$, the relation $R$ is rectangular when viewed as a binary relation, a subset of $\pr_I\rel\tm\pr_{[n]-I}\rel$. A relational structure $\cH$ is \emph{strongly rectangular} if every relation $\rel\in \langle \cH \rangle$ of arity at least 2 is rectangular.

An $n$-by-$m$ matrix $M$ is said to be a \emph{rank-1 block matrix} if by permuting rows and columns it can be transformed to a block-diagonal matrix (not necessarily square), where every nonzero diagonal block with has rank at most 1. 

Finally, let $\rel(x, y, z)$ be a ternary relation,a subset of $H_1 \times H_2 \times H_3$. We call $\rel$ \emph{balanced} if the matrix $\sM_\rel \in\bbZ^{|H_1|\tm|H_2|}$, where $\sM_\rel[x, y] = |\{z \in H_3\mid (x, y, z) \in \rel\}|$   such that $x \in H_1$ and $y \in H_2$ is a rank-1 block matrix. A relation $\rel$ of arity $n > 3$ is balanced if every representation of $\rel$ as a ternary relation, a subset of $H^k \times H^\ell \times H^{n-k-\ell}$, is balanced. A relational structure $\cH$ is called \emph{strongly balanced} if every relation $\rel \in \SP \cH$ is balanced. 

These three concepts are closely related and proved to be very useful for exact counting. It is well known \cite{Hagemann73:permutable} that strong rectangularity of a relation structure is equivalent to it having a Mal'tsev polymorphism. It is also straightforward that every balanced relation is rectangular. According to \cite{effective-Dyer-doi:10.1137/100811258} $\#\CSP(\cH)$ is polynomial time solvable if and only if $\cH$ is strongly balanced. However, in order for the solution algorithm to work, it requires a Mal'tsev polymorphism to be applied over and over again to construct a \emph{frame}, that is, a compact representation of the set of solutions, \cite{DBLP:journals/jacm/Bulatov13,effective-Dyer-doi:10.1137/100811258}. As was observed above, pp-definitions do not quite work in the case of modular counting, and we introduce their modular counterparts by introducing modular quantifier (see below). Using modular pp-definitions, we can change the definitions of rectangularity and balancedness accordingly to obtain the properties of strong $p$-rectangularity and $p$-balancedness. These new properties require that every modular-pp-definable relation is rectangular and balanced.  Modular-pp-definitions preserve the complexity of modular problems, however, they destroy the connections between the concepts above. In Section~\ref{sec:rectangularity-main} we present a series of counterexamples showing that almost all the properties of $p$-rectangularity, $p$-permutability, $p$-balancedness, and having a Mal'tsev polymorphism are independent of each other. While, as we prove the first three properties are necessary for the tractability of the problem, the latter one does not follow from them, and therefore neither of $p$-rectangularity, $p$-permutability, $p$-balancedness help to construct a solution algorithm in the same fashion it is done for exact counting. 

%%%%%%%%%%%%%%%%%%%%%%%%%%%%%%%%%%%%%%%%%%%%%%%%%%%%%
\subsection{Remedies}\label{sec:remedies}

In the category of positive results we propose several ways to overcome the difficulties outlined in Section~\ref{sec:failure}. We hope that these new approaches will lead to a complexity classification of modular counting.

\smallskip 

\textbf{Multi-sorted structures and $p$-rigidity.}
We start with applying a well-known framework of multi-sorted relational structures to modular counting. While multi-sorted structures is a standard tool in the study of decision CSPs, it usually only used to simplify arguments and streamline algorithms. Here we use this framework in a more fundamental way, to strengthen the main (hypothetical) tractability condition.

Let us consider again the structure $\cT_p$ introduced in the beginning of Section~\ref{sec:failure} and discover that when considered as a multi-sorted structure $\cT_p$ is not $p$-rigid, and therefore can be transformed so that the corresponding modular CSP is solvable in polynomial time. The idea is, given an instance $\cG$ of $\ghomk p{\cT_p}$, to use some preprocessing algorithm (such as local propagation or a solution algorithm for the decision CSP) to reduce possible range of vertices from $\cG$ under homomorphisms. We then \emph{refine} the instance by replacing $\cT_p$ with a multi-sorted structure in such a way that every possible range of a vertex from $\cG$ becomes a set of a distinct type. Now, since automorphisms of a multi-sorted structure may act differently on different types, the refined structure may have a richer automorphism group. For instance, for the refinement of $\cT_p$ the sets $T_p$ and $\{0\}$ have different types, and an automorphism of such a structure may act differently on these sets and no longer has to have $0$ as a fixed point when acting on $T_p$.

We identify several ways of how a refinement of a structure or an instance can be constructed and prove that these constructions are parsimonious reductions between modular CSPs. We are not aware of any example of a refined $p$-rigid structure whose complexity differs from that for exact counting.

\smallskip

\textbf{Modular existential quantifiers.}
We follow the approach of \cite{DBLP:conf/stoc/BulatovK22} by expanding the relational structure $\cH$ by adding pp-definable relations, but doing in a manner friendly to modular counting.

We introduce a new form of expansion which is using $p$-modular quantifiers instead of regular existential quantifiers. The semantics of a $p$-modular quantifier is "there are non-zero modulo $p$ values of a variable" rather than "there exists at least one value of a variable" as the regular existential quantifier asserts. The new concept gives rise to new definitions of pp-formulas. If regular existential quantifiers in pp-formulas are replaced with $p$-modular quantifiers, we obtain \emph{$p$-modular primitive positive formulas} (\emph{$p$-mpp-formulas}, for short). The $p$-modular quantifier is denoted $\exists^{\equiv p}$, and so $p$-mpp-formulas have the form 
\[
\exists^{\equiv p}y_1\zd y_{\ell_1}\dots \exists^{\equiv p}y_{\ell_1+\dots+\ell_{s-1}+1}\zd y_k\Phi(\vc zm).
\] 
Note the more complicated form of the quantifier prefix: modular quantification is not as robust as the regular one and quantifying variables away in groups or one-by-one may change the result. For example, let $\rel=\{(1,0,0),(1,1,0),(1,1,1),(2,2,2)\}$ be a relation on $\{0,1,2\}$. Then formulas $\exists^{\equiv3}y \exists^{\equiv3}z\ \  \rel(x,y,z)$ and $\exists^{\equiv3}y,z\ \  \rel(x,y,z)$ define sets $\{1,2\}$ and $\{2\}$, respectively.

Every relational structure is associated with a relational clone $\SP\cH$ that consists of all relations pp-definable in $\cH$.  Then, similar to pp-definitions, a relation $\rel$ is said to be \emph{$p$-mpp-definable} in a structure $\cH$ if there is a $p$-mpp-formula in $\cH$ expressing $\rel$. The \emph{$p$-modular clone} $\SP\cH_p$ associated with $\cH$ is the set of all relations $p$-mpp-definable in $\cH$. Similar to the result of Bulatov and Dalmau \cite{ref:BULATOV_TowardDichotomy}, expanding a structure by a  $p$-mpp-definable relation does not change the complexity of the problem $\#_p\CSP(\cH)$.  

\begin{theorem}
Let $\ovarrow\cH$ be a be a $\sg$-structure with equality, and $p$ a prime. Let $\rel$ be a relation that is defined by
\[
R(\vc xk)=\existspr\vc y\ell \Phi(\vc xk, y) 
\]
then $\#_p\CSP(\ovarrow\cH+\ovarrow\rel)$ is polynomial time reducible to $\#_p\CSP(\cH)$.
\end{theorem}

\smallskip

\textbf{An algorithm for parity.}
The next two results identify a number of conditions under which it is possible to design an algorithm or to prove the hardness of the problem. One such case is $\#_2\CSP(\cH)$ when $\cH$ satisfies both the 2-rectangularity and the usual strong rectangularity conditions (or, equivalently, has a Mal'tsev polymorphism). Note that the case $p=2$ is special and the concepts of 2-rectangularity and 2-balancedness are equivalent as Lemma~\ref{lem:2-rec-vs-2-bal} states. This comes very handy, because it allows us to avoid rectangular but not balanced relations that require a completely different approach. Some sense of what it might require can be seen in \cite{Dyer09:weighted,cai-complex}. 

The other ingredient that we use to develop the algorithm for determining the parity of solutions for CSP is the use of frames. A frame for a relation $\rel \subseteq H^n$ is defined as a subset $\rel' \subseteq \rel$ that satisfies two conditions. Firstly, if $\rel$ includes a tuple where the $i$-th component is $a$, then $\rel'$ must also include at least one such tuple. Secondly, for $1 < i \leq n$, consider a pair of elements $a,b\in H$. We say $a,b$ are \emph{$i$-equivalent} in $\rel$ if $\rel$ contains tuples $\ba,\bb$ that agree on their first $i-1$ components and such that $\ba[i]=a,\bb[i]=b$. Any pair that is $i$-equivalent in $\rel$ must also be $i$-equivalent in $\rel'$, though $\rel'$ may have fewer tuples sharing these common prefixes compared to $\rel$. It is demonstrated in \cite{effective-Dyer-doi:10.1137/100811258} that every $n$-ary relation over $H$ has a compact frame of size at most $n|H|$, despite $\rel$ potentially having up to $|H|^n$ elements. If $\rel$ such that has Mal'tsev polymorphism, it can be reconstructed from any of its frames. Additionally, they provide a method for constructing a frame efficiently provided again that $\rel$ has a Mal'tsev polymorphism.

We apply the known techniques \cite{Bulatov06:simple,effective-Dyer-doi:10.1137/100811258} to find a concise representation or a frame of the set of solutions of a given CSP. In order to do that we consider the set of solutions of a (decision) CSP instance with $n$ variables as an $n$-ary relation $\rel$. First of all the algorithm finds a frame for $\rel$ using the existing techniques \cite{Bulatov06:simple}.Then it looks for a frame for the relation $\WT\rel\sse\pr_{[n-1]}\rel$ consisting of the tuples from $\pr_{[n-1]}\rel$ that have odd number of extensions from $\rel$. This requires careful filtering of tuples obtained from the original frame. This filtering step is the main novelty brought by our algorithm. The procedure then repeats reducing the arity of $\rel$ ever further. The key property of this procedure is that the parity of $\WT\rel$ equals that of $\rel$. Once reduced to a unary relation, determining the parity becomes straightforward.

\begin{theorem}\label{the:parity-algorithm-overview}
Let $\cH$ be a 2-rigid, strongly 2-rectangular, and $\SP \cH_2$ has a Mal'tsev polymorphism\footnote{Note that the latter condition does not follow from $\cH$ having a Mal'tsev polymorphism, because 2-mpp-definitions do not preserve polymorphisms, as we show in Section~\ref{sec:rectangularity-main}.}. Then $\#_2\CSP(\cH)$ can be solved in time $O(n^5)$. 
\end{theorem}

\smallskip

\textbf{Automorphisms of direct products of structures.}
In Section~\ref{sec:automorphisms} we explore what implications of a structural result similar to that in \cite{ref:ProductOfGraphs} that is used in \cite{DBLP:conf/stoc/BulatovK22} about $\Aut(\cH^n)$ can be. A \emph{rectangularity obstruction} is a violation of the rectangularity or $p$-rectangularity property, that is, an ($n$-ary) relation $\rel$ pp- or $p$-mpp-definable in a structure $\cH$, $k\in[n]$, and tuples $\ba,\bb\in\pr_{[k]}\rel,\bc,\bd\in\pr_{[n]-[k]}\rel$ such that $(\ba,\bc),(\ba,\bd),(\bb,\bd)\in\rel$, but $(\bb,\bc)\not\in\rel$. A \emph{generalized rectangularity obstruction} are the relation $\rel$ and sets $A_{1,1},A_{1,2}\sse\pr_{[k]}\rel$, $A_{2,1},A_{2,2}\sse\pr_{[n]-[k]}\rel$ such that $A_{1,1}\cap A_{1,2}=\emptyset$, $A_{2,1}\cap A_{2,2}=\emptyset$, and any $\ba\in A_{1,1},\bb\in A_{1,2},\bc\in A_{2,1},\bd\in A_{2,2}$ form a rectangularity obstruction. 

At the first glance, if such an obstruction exists, it should be possible to prove the hardness of $\#_p\CSP(\cH)$. Indeed, $\rel$ can be viewed as a bipartite graph $\sK_\rel$, whose parts of the bipartition are $\pr_{[k]}\rel$ and $\pr_{[n]-[k]}\rel$, and as the rectangularity obstruction shows, this graph is not complete bipartite implying that $\ghom{\sK_\rel}$ is \#P-complete. However, the hardness of $\ghomk p{\sK_\rel}$ also involves the requirement that $\sK_\rel$ is $p$-rigid. The property of $p$-rigidity is achieved by restricting the problem to induced subgraphs of $\sK_\rel$. However, such a subgraph may avoid the (generalized) rectangularity obstruction rendering it useless.

The obstruction $A_{1,1},A_{1,2}\sse\pr_{[k]}\rel$, $A_{2,1},A_{2,2}\sse\pr_{[n]-[k]}\rel$ is said to be \emph{protected} in $\rel$ if, after applying a $p$-reduction to $\rel$ under a sequence of $p$-automorphisms from $\Aut(\rel)$, for the resulting relation $\tilde\rel$ it holds that $\pr_{[k]}\tilde\rel\cap A_{1,1},\pr_{[k]}\tilde\rel\cap A_{1,2}\ne\emptyset$, $\pr_{[n]-[k]}\tilde\rel\cap A_{2,1},\pr_{[n]-[k]}\tilde\rel\cap A_{2,2}\ne\emptyset$.  In fact, Theorem~4.2 of \cite{DBLP:conf/stoc/BulatovK22} implies, although implicitly, that any $p$-rigid graph that is not a complete bipartite graph contains a protected rectangularity obstruction. One case of a protected rectangularity obstruction is when it is protected in $\sK_\rel$, that is, survives $p$-reductions of $\sK_\rel$ itself. In this case we say that the obstruction is \emph{graph-protected}.

\begin{proposition}\label{pro:rect-obstructions-intro}
Let $\cH$ be a (multi-sorted) relational structure and $p$ a prime number. If $\cH$ has a graph protected generalized rectangularity obstruction modulo $p$, $\#_p\CSP(\cH)$ is $\#_pP$-complete.
\end{proposition}

We consider a special case of graph-protected generalized rectangularity obstructions, standard hardness gadgets, that have to satisfy the additional condition $A_{1,1}\cup A_{1,2}=\pr_{[k]}\rel, A_{2,1}\cup A_{2,2}=\pr_{[n]-[k]}\rel$. It is proved in Section~\ref{sec:rectangularity-obstruction} that a standard hardness gadget is indeed a graph-protected obstruction.

Standard hardness gadgets provide a fairly limited condition for the hardness of $\#_p\CSP(\cH)$. In fact, it is possible to prove that $\NpCSP(\cH)$ is $\#_pP$-complete whenever $\cH$ has any protected rectangularity obstruction, not necessarily a standard gadget. However, it cannot be done using Theorem~\ref{the:CountingGraphHom} as a black box, and is outside of the scope of this paper.

\smallskip
%%%%%%%%%%%%%
\textbf{Binarization.}
While studying the structure of $\Aut(\cH^k)$ for a relational structure $\cH$ and an integer $k$ may be a difficult problem, in Section~\ref{sec:binarization} we make a step forward by reducing the class of structures $\cH$ for which such a characterization is required. In particular, we show that it suffices to obtain a characterization for structures with only binary rectangular relations. More precisely, for any relational structure $\cH=(H;\vc\rel k)$ we construct its binarization $b(\cH)$ as follows. The structure $b(\cH)$ is multi-sorted, and the domains are the relations $\vc\rel k$ viewed as sets of tuples, thus, $b(\cH)$ has $k$ domains. For every pair $i,j\in[k]$ ($i,j$ are allowed to be equal) and any $s\in[\ell_i],t\in[\ell_j]$, where $\ell_i,\ell_j$ are the arities of $\rel_i,\rel_j$, the structure $b(\cH)$ contains a binary relation $\relo^{ij}_{st}=\{(\ba,\bb)\mid \ba\in\rel_i,\bb\in\rel_j, \ba[s]=\bb[t]\}$. We prove that $\cH$ and $b(\cH)$ share many important properties.

\begin{theorem}\label{the:bipartization-intro}
Let $\cH$ be a relational structure. Then $\cH$ is strongly rectangular ($p$-strongly rectangular, $p$-rigid, has a Mal'tsev polymorphism) if and only if $b(\cH)$ is strongly rectangular ($p$-strongly rectangular, $p$-rigid, has a Mal'tsev polymorphism).
\end{theorem}

In addition to Theorem~\ref{the:bipartization-intro} every relation of $b(\cH)$ is binary and rectangular. This makes such structures somewhat closer to graphs and the hope is that it will be easier to study the structure of $\Aut(b(\cH)^n)$ than $\Aut(\cH^n)$ itself.

%%%%%%%%%%%%%%%%%%%%%%%%%%%%%%%%%%%%%%%
%%%%%%%%%%%%%%%%%%%%%%%%%%%%%%%%%%%%%%%
\section{Modular primitive-positive definitions}\label{sec:modular-pp}

In this section we modify the concept of pp-definitions in such a way that adding a definable relation does not change the complexity of the modular counting problem. We start with introducing a different type of quantifier. For a prime number $p$ the \emph{$p$-modular} quantifier $\existspr$ is interpreted as follows
\begin{align}
\rel(\vc xk)=&\existspr\vc ys \Phi(\vc xk,\vc ys) \text{ is true }\label{equ:p-mpp} \\ 
&\Leftrightarrow  \ext_\Phi(\vc xk) \not \equiv 0 \pmod{p}
\Leftrightarrow \extp_\Phi(\vc xk) \not = 0 .   \nonumber
\end{align}

Note that modular quantifiers are not as robust as regular ones, and the result may depend on whether variables are quantified in groups or one by one, as the following example shows. We will allow quantifiers both ways.

\begin{example}\label{exa:quantifier-order}
Let $\rel=\{(1,0,0),(1,1,0),(1,1,1),(2,2,2)\}$.
Then the formula 
$
\exists^{\equiv3}y \exists^{\equiv3}z\ \  \rel(x,y,z)
$
defines the set $\{1,2\}$, since in every step the number of extensions is not divisible by 3. However, the formula
$
\exists^{\equiv3}y,z\ \  \rel(x,y,z)
$
results in the set $\{2\}$. 
\end{example}

A primitive positive formula that uses modular existential quantifiers mod $p$ instead of regular ones is called a \emph{$p$-modular primitive positive}. The relation it defines is said to be \emph{$p$-mpp-definable}, and the definition itself is called a \emph{$p$-mpp-definition}. If for the $p$-mpp-definition~\eqref{equ:p-mpp} $\extp_\Phi(\ba)=1$ for all $\ba\in\rel$, the $p$-mpp definition is said to be \emph{strict}. By Proposition~5.6 from \cite{DBLP:conf/stoc/BulatovK22} if $\rel$ has a $p$-mpp-definition, it has a strict $p$-mpp definition.

Finally, we show that adding $p$-mpp-definable relations does not change the complexity of a CSP.

%% \begin{lemma}\label{lem:conjunctive-multi-sorted}
%% Let $\ovarrow\cH$ be a multi-sorted relational structure with signature $\sg$, $\rel$ conjunctive definable in $\ovarrow\cH$, and let $\ovarrow\cH+\rel$ denote the expansion of $\ovarrow\cH$ by the predicate symbol $\rel$ that is interpreted as the relation $\rel$ in $\ovarrow\cH$. Then $\NpCSP(\ovarrow\cH+\rel)\le_T\NpCSP(\ovarrow\cH)$.
%% \end{lemma}
%% 
%% \begin{proof}
%% Let $\rel$ be defined by a conjunctive formula 
%% \[
%% \relo_1(y_{11}\zd y_{1\ell_1})\wedge\dots\wedge \relo_s(y_{s1}\zd y_{s\ell_s}),
%% \] 
%% where $\vc\relo s\in\sg$ and $y_{ij}\in\{\vc xk\}$ for all $i,j$. We use the standard definition of the CSP. Let $\cP=(V,\cC)$ be an instance of $\NpCSP(\ovarrow\cH+\rel)$. If $\cC$ contains a constraint of the form $\ang{(\vc xk),\rel}$, replace it with constraints $\ang{(y_{11}\zd y_{1\ell_1}),\relo_1},\zd\ang{(y_{s1}\zd y_{s\ell_s}),\relo_s}$. As is easily seen, the resulting instance has exactly the same solutions as $\cP$. Repeat this procedure while constraints containing $\rel$ remain. The resulting instance $\cP'$ has the same solutions as $\cP$, and is an instance of $\NpCSP(\ovarrow\cH)$.
%% \end{proof}
%% 
\begin{theorem}\label{pro:GadgetExists}
Let $\ovarrow\cH$ be a be a $\sg$-structure with equality, and $p$ a prime. Let $\rel$ be a relation that is defined by
\[
R(\vc xk)=\existspr\vc y\ell \Phi(\vc xk, y) 
\]
then $\#_p\CSP(\ovarrow\cH+\ovarrow\rel)\le_T\#_p\CSP(\cH)$.
\end{theorem}

\begin{proof}
Without loss of generality we may assume that the $p$-mpp-definition of $\rel$ given in the theorem is strict. We use the standard definition of the CSP. 

Take a problem instance $\cP = (V,\cC)$ of $\#_p\CSP(\cH + \rel)$ where $\cC = \{ C_1,\dots,C_t\}$. Without loss of generality we may assume that $C_1,\dots, C_m$, $m \leq t$, are the constraints containing $\rel$. Then the arity of $C_i$ is $k$ for $i\in [m]$.

Let $\cP'$ be the problem instance from $\#_p\CSP(\cH)$ in which each constraint $C_i = \constCSP{(x_{i1},\dots,x_{ik})}{\rel}$ is replaced with formula $\Phi(x_{i1},\dots,x_{ik},\vc {y^i}\ell)$ where all the variables $y^i_j$ are different and do not occur in $\cP$.

For a solution $\vf$ to $\cP$, as for every $i\in[m]$, $\rel(x_{i1},\dots,x_{ik})$ is true, $\vf$ can be extended to a solution $\vf'$ of $\cP'$. Since the $p$-mpp-definition of $\rel$ we use is strict, the number $\ext(\vf)$ of such extensions is equal to $1\pmod p$. Therefore, the number of solutions of $\cP$ and that of $\cP'$ are equal $\pmod p$. The result follows.
\end{proof}

%% 
%% \begin{theorem}[see also \cite{DBLP:conf/stoc/BulatovK22}]\label{the:conjunctCSP}
%% Let $\ovarrow\cH$ be a relational structure and $p$ prime.\\[1mm]
%% Let $\rel$ be $\meet$-definable in $\ovarrow\cH$.
%% %% , and $\ovarrow\cH+\rel$ denote the expansion of $\cH$ by a predicate symbol $\rel$ that is interpreted as the relation $\rel$ in $\cH$. 
%% Then $\NpCSP(\ovarrow\cH+\rel)$ is Turing reducible to $\NpCSP(\ovarrow\cH)$.
%% \end{theorem}

For a structure $\ovarrow\cH$ the set of all relations pp-definable in $\ovarrow\cH$ is called the relational clone of $\ovarrow\cH$ and denoted by $\langle \ovarrow\cH \rangle$. The set of relations  $p$-mpp-definable in $\cH$ is called the \emph{$p$-modular clone} and  denoted by $\SP{\cH}_p$.

%%%%%%%%%%%%%%%%%%%%%%%%%%%%%%%%%%%%%%%
%%%%%%%%%%%%%%%%%%%%%%%%%%%%%%%%%%%%%%%
\section{Rigidity}\label{sec:rigidity}

%% \input{Data/d-rigidity}

%%%%%%%%%%%%%%%%%%%%%%%%%%%%%%%%%%%%%%%
\paragraph*{Rigid graphs and rigid structures}\label{sec:rigid-non-rigid}
% p-rigidity 
% counterexample (the structure with good polynomials)
% algorithm based on polynomials
%% In this section it will be convenient to use the \emph{standard} definition of the CSP, see \cite{ref:POlymorphismAndUsethem_barto2017polymorphisms} or Appendix~\ref{sec:two-views}.

Recall that the group of automorphisms of a multi-sorted structure $\ovarrow\cH$ is denoted by $\Aut(\ovarrow\cH)$. An automorphism $\ovarrow\pi\in\Aut(\ovarrow\cH)$ has order $p$ if every component of $\ovarrow\pi$ has order $p$ or is the identity permutation, and at least one of the components has order $p$.
%% $\pi^p$ is the identity mapping, while $\pi^\ell$ is not identity for any $\ell<p$. 
Relational structures without order $p$ automorphisms are called \emph{$p$-rigid}. By the result of \cite{DBLP:conf/stoc/BulatovK22} if $\cH$ is a $p$-rigid graph then $\CSPp \cH$ is polynomial time if and only if $\#CSP(\cH)$ is. The following example shows that it is not the case for relational structures.

\begin{example}\label{exa:bad-triangle}
Fix a prime number $p$ and let $T_p=\{0\zd p-1,p,p+1\}$. A relational structure $\cT_p$ has the base set $T_p$ and predicates $\rel,C_0\zd C_{p+1}$, where $C_a$ is the constant relation $\{(a)\}$ and $\rel=T^2_p-\{(0,p)\zd(p-1,p)\}$. The structure $\cT_p$ is obviously $p$-rigid as it contains all the constant relations. By \cite{Bulatov17:dichotomy,Zhuk20:dichotomy} the decision $\CSP(\cT_p)$ is solvable in polynomial time, while the exact counting problem $\#\CSP(\cT_p)$ is $\#$P-complete by \cite{ref:BULATOV_TowardDichotomy}, as it contains an induced subgraph (or subrelation) $\{(0,0),(p,0),(p,p)\}$ (or alternatively it does not have a Mal'tsev polymorphism). In fact, structures like $\cT_p$ play a very important role in proving that having a Mal'tsev polymorphism or avoiding induced subgraphs as above is a necessary condition for the tractability of $\#\CSP$, see \cite{ref:BULATOV_TowardDichotomy}. We now show that $\NpCSP(\cT_p)$ can be solved in polynomial time. 

We outline the idea of the algorithm here and later will describe a more general algorithm. Let $\cP=(V,\cC)$ be an instance of $\NpCSP(\cT_p)$. Perform the following three steps:\\[2mm]
{\sc Step 1.} Establish the \emph{arc-consistency} of $\cP$. That is, repeat the following procedure until no changes are possible. Pick constraints $C_1=\ang{\bs_1,\relo_1}, C_2=\ang{\bs_2,\relo_2}$ such that $W=\bs_1\cap\bs_2\ne\eps$ and a tuple $\ba\in\relo_1$. Since the arc-consistency algorithm modifies constraint relations, $\relo_1,\relo_2$ represent the current state of the corresponding constraint relations. If there is no tuple $\bb\in\relo_2$ such that $\pr_W\bb=\pr_W\ba$, remove $\ba$ from $\relo_1$.\\[2mm]
{\sc Step 2.} For every variable $v\in V$ there is now a \emph{domain} $D_v$ such that, for any $C=\ang{\bs,\relo}\in\cC$ with $v\in\bs$, it holds that $\pr_v\rel=D_v$. If $D_v=\eps$ for some $v\in V$, output 0, as $\cP$ has no solutions in this case. Otherwise remove all the variables from $V$ such that $|D_v|=1$. Denote the resulting problem $\cP'=(V',\cC')$. The number of solutions to $\cP'$ equals that of $\cP$.\\[2mm]
{\sc Step 3.} It is not difficult to see that if $D_v\cap\{0\zd p-1\}\ne\eps$ for some variable $v\in V'$ then $\{0\zd p-1\}\sse D_v$. Moreover, if $\vf$ is a solution of $\cP'$ and $\vf(v)\in\{0\zd p-1\}$, then any mapping $\psi$ such that $\psi(v)\in\{0\zd p-1\}$ and $\psi(w)=\vf(w)$, $w\in V'-\{v\}$, is also a solution of $\cP'$. For every such variable $v$ remove $\{0\zd p-1\}$ from $D_v$, as the number of solutions $\vf$ to $\cP'$ with $\vf(v)\in\{0\zd p-1\}$ is a multiple of $p$. Then repeat {\sc Step 2}.\\[2mm]
{\sc Step 4.} In the remaining case $D_v=\{p,p+1\}$ for every $v\in V'$. It is not difficult to see that any mapping from $V'$ to $\{p,p+1\}$ is a solution of $\cP'$. Therefore, output $2^{|V'|}\pmod p$.
\end{example}

%%%%%%%%%%%%%%%%%%%%%%%%%%%%%%%%%%%%%%%
\paragraph*{Refining multi-sorted structures}\label{sec:refinement}

In this subsection we introduce a different concept of rigidity that is stronger than the one used before. In particular, it will explain the tractability of the problem from Example~\ref{exa:bad-triangle}.

While Lemma~\ref{lem:aut-reduction-multi-sorted-structures-prelim} allows one to reduce CSPs over non-$p$-rigid structures, it is sometimes possible to go further and reduce a CSP to one with a richer structure, and a richer set of automorphisms. Let $\ovarrow\cH=(\colect{H_i}{i\in[k]};\vc\rel m)$ be a multi-sorted relational structure. 
We say that a structure $\ovarrow\cG$ is a \emph{refinement} of $\ovarrow\cH$ if it satisfies the following conditions: 
\begin{itemize}
\item [(a)] 
$\ovarrow\cG=(\colect{G_i}{i\in[q]};\vc\relo t)$, where the $G_i$'s are pairwise disjoint,
\item [(b)] 
for every $i\in[q]$, there is an injective mapping $\xi_i:G_i\to H_{i'}$ for some $i'\in[k]$,
\item [(c)] 
for every $j\in[t]$ there is $j'\in[m]$ with $\rel_{j'}\sse H_{i_1}\tm\dots\tm H_{i_\ell}$ and $\relo_j=\{(\vc a\ell)\in G_{i'_1}\tms G_{i'_\ell}\mid(\xi_{i'_1}(a_1)\zd\xi_{i'_\ell}(a_\ell))\in\rel_{j'}\}$, where $G_{i'_r}$ is such that $\xi_{i'_r}(G_{i'_r})\sse H_{i_r}\cap\pr_r\rel_{j'}$.
\end{itemize}
Condition (a) is required because the domains of a multi-sorted structure have to be disjoint. Condition (b) basically says that $G_i$ consists of copies of some elements of $H_{i'}$. Condition (c) amounts to saying that $\relo_j\sse\rel_{j'}$, except it uses copies of the elements of $\cH$. In the notation of item (c) we use $\xi(\vc a\ell)$ to denote $(\xi_{i'_1}(a_1)\zd\xi_{i'_\ell}(a_\ell))$, we use $\xi(\relo_j)$ to denote $\{\xi(\vc a\ell)\mid (\vc a\ell)\in\relo_j\}$, and $\xi^{-1}(\rel_{j'})$ for $\{(\vc a\ell)\in G_{i'_1}\tms G_{i'_\ell}\mid \xi(\vc a\ell)\in\rel_{j'}\}$.

It is possible that while $\cH$ is $p$-rigid, its refinement is not and Lemma~\ref{lem:aut-reduction-multi-sorted-structures-prelim} may apply.

\begin{example}[Continued from Example~\ref{exa:bad-triangle}]\label{exa:bad-triangle2}
Consider the structure $\cT_p$ from Example~\ref{exa:bad-triangle}. Let $G_i=\{i'\}$ for $i\in\{0\zd p+1\}$, $G_{p+2}=\{p'',(p+1)''\}$, and $G_{p+3}=\{0^*\zd (p-1)^*,(p+1)^*\}$. For convenience we use $G_{-1}$ to denote $T_p$. The mappings $\xi_i:G_i\to H$ simply erase dashes and stars of elements from $G_i$. Then for every $i,j\in\{-1\zd p+3\}$ we introduce relation $\relo_{ij}=\xi^{-1}(\rel\cap(\xi_i(G_i)\tm \xi_j(G_j))$. For instance, $\relo_{-1p+2}=(\{0\zd p+1\}\tm\{p+1\})\cup\{(p,p),(p+1,p)\}$. Also, by construction the only constant relations needed to be introduced are $C_{i'}$ on $G_i$ for $i\in\{0\zd p+1\}$. The structure 
\[
\ovarrow\cT^*_p=(G_{-1}\zd G_{p+3};\relo_{ij}, i,j\in\{-1\zd p+3\}: C_{i'}, i\in\{0\zd p+1\})
\]
is a refinement of $\cT_p$. The structure $\ovarrow\cT*_p$ is not $p$-rigid, as the collection of mappings $\pi_i$, $i\in\{-1\zd p+3\}$, defined as follows is an automorphism: The mappings $\pi_0\zd\pi_{p+2}$ are identity mappings. The mapping $\pi_{-1}$ is given by $\pi_{-1}(p)=p,\pi_{-1}(p+1)=p+1$, and $\pi_{-1}(a)=a+1\pmod p$ for $a\in\{0\zd p-1\}$. Finally, $\pi_{p+3}$ is the restriction of $\pi_{-1}$ on $G_{p+3}$ (with stars added to the elements of $G_{p+3}$).
\end{example}

In order to relate refinement structures with reductions between counting problems we introduce two special types of refinement. First of all, we will need an alternative approach to pp-definitions based on homomorphisms, see \cite{Feder98:computational,ref:kolaitis2004constraint}.

\begin{lemma}\label{lem:pp-gadget}
A predicate $\rel(\vc xk)$ is pp-definable in a multi-sorted structure $\ovarrow\cH$ containing the equality predicate if and only if there exists a similar structure $\ovarrow\cG_R$ containing vertices $\vc xk$ such that for any $(\vc{a}{k})\in R$ it holds that $(\vc ak)\in\rel$ if and only if there is a homomorphism from $\ovarrow\cG_\rel$ to $\ovarrow\cH$ that maps $x_i$ to $a_i$, $i\in[k]$.
We will say that $\ovarrow\cG_\rel$ \emph{defines} $R$ in $\ovarrow\cH$.
\end{lemma}

Let $\ovarrow\cH=(\colect{H_i}{i\in[k]};\vc\rel m)$ be a multi-sorted relational structure. The \emph{Gaifman graph} of $\ovarrow\cH$ is the graph $G(\ovarrow\cH)=(V,E)$, where $V=\bigcup_{i\in[k]}H_i$ and $(a,b)\in E$ if and only if there is $j\in[m]$ and $\ba\in\rel_j$ such that $\ba[s]=a,\ba[t]=b$ for some coordinate positions $s,t$ of $\rel_j$. The structure $\ovarrow\cH$ has \emph{treewidth $d$} if $G(\ovarrow\cH)$ has treewidth $d$. 

We say that a refinement $\ovarrow\cG=(\colect{G_i}{i\in[q]};\vc\relo t)$ of $\ovarrow\cH$ is \emph{width $d$ definable} if for every $i\in[q]$ there is a structure $\ovarrow\cG_i$ of treewidth at most $d$ that defines $\xi_i(G_i)$ in $\ovarrow\cH$. In a similar way, we say that $\ovarrow\cG$ is a \emph{definable} refinement if for every $i\in[q]$ the set $\xi_i(G_i)$ is pp-definable in $\ovarrow\cH$. Finally, we say that $\ovarrow\cG$ is \emph{the full width $d$ definable refinement} (respectively, \emph{full definable refinement}) if it satisfies the following conditions. 
\begin{itemize}
\item[(1)] 
It is a width $d$ definable refinement (respectively, a definable refinement). 
\item[(2)] 
For every unary relation $U$ definable in $\cH$ by a structure of treewidth at most $d$ (respectively, pp-definable unary relation) there is $i\in[q]$ such that $\xi_i(G_i)=U$. 
\item[(3)] 
For every relation $S$ obtained from some relation $\rel_j$ by restricting it to domains definable by structures of width $d$ (respectively by pp-definable domains), there is $\relo_{j'}$ such that $\xi(\relo_{j'})=S$. 
\end{itemize}
Since the original domains $H_i$, $i\in[k]$, have trivial pp-definitions, they (or rather their copies) are always among the $G_j$'s, and copies of the original relations are also among the $\relo_j$'s, although they may be over different, smaller domains than the $\rel_i$'s. 

Next, we extend refinements to CSP instances. Let $\ovarrow\cG=(\colect{G_i}{i\in[q]};\vc\relo t)$ be a refinement of $\ovarrow\cH$ and $\cP=(V,\cC)$ an instance of $\CSP(\ovarrow\cH)$. Recall that every variable $v\in V$ is assigned a sort $\sg(v)\in[k]$. Let $\sg':V\to[q]$ be such that $\xi_{\sg'(v)}:G_{\sg'(v)}\to H_{\sg(v)}$ for each $v\in V$. The instance $\cP^{\sg'}=(V,\cC^{\sg'})$ is said to be a \emph{refinement for $\ovarrow\cG$} of $\cP$ with the sort function $\sg'$ if it satisfies the following two conditions
\begin{itemize}
\item[(a)] 
every $v\in V$ is assigned the sort $\sg'(v)$;
\item[(b)]
for every $C=\ang{\bs,\rel}\in\cC$, $\bs=(\vc v\ell)$, it holds that $\xi_{\sg'(v)}(G_{\sg'(v_i)})\sse\pr_i\rel$, $i\in[\ell]$, and there is $C'=\ang{\bs,\xi^{-1}(\rel)}\in\cC^{\sg'}$.
\end{itemize}
The refinement $\cP^{\sg'}$ is \emph{lossless} if for every solution $\vf$ of $\cP$ the mapping $\xi^{-1}\circ\vf$ is a solution of $\cP^{\sg'}$. Suppose $\ovarrow\cG$ is full pp-definable. The refinement $\cP^{\sg'}$ is \emph{minimal lossless} if it is lossless and for each $v\in V$, $\sg'(v)$ is minimal (with respect to inclusion of $\xi_{\sg'(v)}(G_{\sg'(v)}$ in the original domain). If $\ovarrow\cG$ is full of width $d$, the definition is bit more complicated. As we saw in Section~\ref{sec:two-views}, the instance $\cP$ can also be viewed as a structure $\ovarrow\cF$ with vertex set $V$ and such that the solutions of $\cP$ are exactly the homomorphisms from $\ovarrow\cF$ to $\ovarrow\cH$. Then $\cP^{\sg'}$ is \emph{minimal lossless of width $d$} if it is lossless and for every $v\in V$, $\xi_{\sg'(v)}(G_{\sg'(v)})$ is minimal with respect to inclusion among unary relations defined by a structure $\ovarrow\cG_v$ of treewidth $d$ with a designated variable $x$ and such that there is a homomorphism from $\ovarrow\cG_i$ to $\ovarrow\cF$ mapping $x$ to $v$.

\begin{proposition}\label{pro:treewidth-refinement}
Let $\ovarrow\cH$ be a multi-sorted relational structure.\\[2mm]
(1) Let $\ovarrow\cG$ be the full width $d$ definable refinement of $\ovarrow\cH$. For any instance $\cP$ of $\CSP(\ovarrow\cH)$, its minimal lossless refinement of width $d$ for $\ovarrow\cG$ can be found in polynomial time.\\[2mm]
(2) Let $\ovarrow\cH$ be a relational structure containing all the constant relations and such that $\CSP(\ovarrow\cH)$ is solvable in polynomial time, and $\ovarrow\cG$ the full definable refinement of $\ovarrow\cH$. Then for any instance $\cP$ of $\CSP(\ovarrow\cH)$, its minimal lossless refinement for $\cG$ can be found in polynomial time.
\end{proposition}

\begin{proof}
The method we use in the proof of this proposition is, given an instance $\cP=(V,\cC)$ of $\CSP(\ovarrow\cH)$, to, first, identify the correct domains for each variable, and then restrict $\cP$ to those domains and assign to each variable the type associated with its domain.

(1) In this case we first run the $d$-consistency algorithm on $\cP$ (for definitions and properties of this algorithm see \cite{Feder98:computational,ref:kolaitis2004constraint}). This reduction is lossless. Let $G_{i_v}$ denote the domain of $v\in V$ after establishing $d$-consistency. By the results of \cite{Feder98:computational,ref:kolaitis2004constraint} every domain $G_{i_v}$ is definable by a structure of treewidth at most $d$ and is minimal with this property. Therefore, each $G_{i_v}$ is a domain of $\ovarrow\cG$. 

(2) By solving instances of the form $\cP_{v=a}=(V,\cC\cup\{\ang{(v),C_a(v)})$ for each $v\in V$ and $a$ from the domain of $v$ we can determine the sets $D_v=\{a\mid \vf(v)=a, \text{ $\vf$ is a solution of $\cP$}\}$. Since the $D_v$'s are defined through a CSP instance, they are pp-definable in $\ovarrow\cH$, and therefore every $D_v$ is a domain in $\ovarrow\cG$. We assign $v$ the corresponding type. Clearly, the refinement is lossless and minimal.
\end{proof}

\begin{example}[Continued from Example~\ref{exa:bad-triangle}]\label{exa:bad-triangle3}
Consider again the structure $\cT_p$ from Example~\ref{exa:bad-triangle} and its refinement $\ovarrow\cT^*_p$ from Example~\ref{exa:bad-triangle2} to solve $\NpCSP(\cT_p)$. It is not hard to see that $\ovarrow\cT^*_p$ is the full refinement of $\cT_p$ of treewidth 1. Indeed, the set $\xi_i(G_i)$, $i\in\{0\zd p+1\}$, is defined by the pp-formula $C_i(x)$, and 
\[
\xi_{p+2}(G_{p+2})(x)=\exists y(\rel(x,y)\wedge C_p(y)),\qquad 
\xi_{p+3}(G_{p+3})(x)=\exists y(\rel(y,x)\wedge C_0(y)).
\]
Therefore, by Proposition~\ref{pro:treewidth-refinement} $\#_pCSP(\cT_p)$ is polynomial time reducible to $\#_pCSP(\ovarrow\cT^*_p)$. As was observed in Example~\ref{exa:bad-triangle2}, the structure $\ovarrow\cT^*_p$ has an automorphism $\pi$ of order $p$. The reduced structure $(\ovarrow\cT^*_p)^\pi$ has the domains $G_0\zd G_{p+1},G_{p+2}$ unchanged, $G_{-1}$ reduced to $\{p,p+1\}$, $G_{p+3}$ reduced to $\{(p+1)^*\}$, and all the relations are Cartesian products of the corresponding domains. Such a problem is obviously solvable in polynomial time. 
\end{example}

%%%%%%%%%%%%%%%%%%%%%%%%%%%%%%%%%%%%%%%
%%%%%%%%%%%%%%%%%%%%%%%%%%%%%%%%%%%%%%%
\section{Maltsev polymorphisms, $p$-Rectangularity, and $p$-Permutability}\label{sec:rectangularity-main}

The concepts and results from this section can be generalized to multi-sorted structures. However, the single-sorted case is sufficient to communicate the main ideas.

%%%%%%%%%%%%%%%%%%%%%%%%%%%%%%%%%%%%%%%%%%%%%%%%%%%%%%%%%%%%%%%%%%%%%%%%%%%%%%%%%%
\subsection{The four properties}

We use the four properties of relational structures that are strongly related to the complexity of counting problems.

\noindent\textbf{Mal'tsev polymorphisms.}
Let $H$ be a set. A mapping $\FUNC{\vf}{H^3}{H}$ is said to be a \emph{Mal'tsev} polymorphism of a relation $\rel \subseteq H^n$ if for any $\vv{u}_1,\vv{u}_2, \vv{u}_3 \in\rel$, it holds that $\vf(\vv{u}_1,\vv{u}_2, \vv{u}_3)\in \rel$ (applied component-wise), and $\vf$  satisfies the equations $\vf(a,b,b)=\vf(b,b,a)=a$, for all  $a,b\in H$. We say that the relational structure $\cH$ has Mal'tsev polymorphism if all relations $\rel \in \langle \cH \rangle$ admit the same Mal'tsev polymorphism.
%% Similarly, we say that the relational structure $\cH$ has Mal'tsev polymorphism is all relations $\rel \in \langle \cH \rangle _p$ admits to the same Mal'tsev polymorphism.

Let $\rel$ be an $n$-ary relation over $H^n$. In general $|\rel|$ can be exponentially large in $n$. However if $\rel$ has a Mal’tsev polymorphism, $\rel$ admits a concise representation, whose size is linear in $n$. Bulatov and Dalmau \cite{Bulatov06:simple}, and Dyer and Richerby~\cite{effective-Dyer-doi:10.1137/100811258} introduced slightly different concepts of such concise representations called “compact representation” in \cite{Bulatov06:simple} and “frames” in \cite{effective-Dyer-doi:10.1137/100811258}. We will return to these constructions in Section~\ref{sec:parity}.

\noindent\textbf{$p$-Rectangularity.} 
A binary relation $\rel\subseteq A_1 \times A_2$ is called \emph{rectangular} if $(a, c),(a, d),(b, c) \in\rel$ implies $(b, d) \in\rel$ for any $a, b \in A_1, c, d \in A_2$. 
A relation $\rel\subseteq H^n$ for $n \geq 2$ is rectangular if for every $I\subsetneq[n]$, the relation $R$ is rectangular when viewed as a binary relation, a subset of $\pr_I\rel\tm\pr_{[n]-I}\rel$. Note that rectangularity is a structural property, and it has nothing to do directly with the size of the relation or its parts. A relational structure $\cH$ is \emph{strongly rectangular} if every relation $\rel\in \langle \cH \rangle$ of arity at least 2 is rectangular. 
The following lemma provides a connection between strong rectangularity and Mal'tsev polymorphisms and was first observed in \cite{Hagemann73:permutable} although in a different language.

\begin{lemma}[\cite{Hagemann73:permutable}, see also \cite{effective-Dyer-doi:10.1137/100811258}]
A relational structure is strongly rectangular if and only if it has a Mal’tsev polymorphism.
\end{lemma}

We introduce a modular version of this concept, strongly $p$-rectangular relations.
A relational structure $\cH$ is said to be strongly \emph{$p$-rectangular}, if every $\rel \in \SP \cH _p$ is rectangular. The following example shows that a relational structure can be strongly rectangular, but not strongly $p$-rectangular.

\begin{example}\label{exa:bad-one}
Let $H=\{0,1,2,3,4\}$, $p=2$, and $\cH=(H;\rel)$, where $\rel$ is the following ternary relation, 
%(triples are written vertically), 
$\rel = \{ (0,0,0), (0,1,1), (1,0,2), (1,0,3), (1,1,4) \}$.
%\[
%\rel=\left(
%\begin{array}{ccccc}
%0&0&1&1&1\\ 0&1&0&0&1\\ 0&1&1&2&0
%\end{array}\right).
%\]
As is easily seen, $\exists^{\equiv2} z\rel(x,y,z)$ is the relation $\{(0,0),(0,1),(1,1)\}$, which is not rectangular, and so, $\cH$ is not strongly 2-rectangular. We now show that $\cH$ is strongly rectangular by presenting a Mal'tsev polymorphism of $\cH$. Let $f(x,y,z)=x+y+z$, where addition is modulo 2. For $\ba=(a_1,a_2,a_3)\in H^3$, let also $\ba'\in\{0,1\}^3$ denote the triple obtained by replacing 2's and 3's with 0's, and 4's with 1's. Then let $g(x,y,z)$ be given by
\[
g(\ba)=\left\{\begin{array}{ll}
b,& \text{if } \ba\in\{(b,a,a),(a,a,b)\mid a\in H, b\in\{2,3,4\}\},\\
f(\ba') & \text{otherwise}.
\end{array}\right.
\]
It is straightforward that $g$ is a Mal'tsev operation and is a polymorphism of $\cH$.
\end{example}

\noindent\textbf{$p$-Balancedness.}
An $n$-by-$m$ matrix $M$ is said to be a rank-1 block matrix if by permuting rows and columns it can be transformed to a block-diagonal matrix (not necessarily square) such that every nonzero block has rank at most 1.

Let $\rel(x, y, z)$ be a ternary relation, a subset of $A_1 \times A_2 \times A_3$. We call $\rel$ \emph{balanced} if the matrix $\sM_\rel \in\bbZ^{|A_1|\tm|A_2|}$, where $\sM_\rel[x, y] = |\{z \in A_3\mid (x, y, z) \in \rel\}|$   such that $ x \in A_1$ and $ y \in A_2$ is a rank-1 block matrix. A relation $\rel$ of arity $n > 3$ is balanced if every representation of $\rel$ as a ternary relation, a subset of $H^k \times H^\ell \times H^{n-k-\ell}$ is balanced. A relational structure $\cH$ (a constraint language $\Gm$) is called \emph{strongly balanced}, if every at least ternary relation $\rel \in \SP \cH $ ($\rel \in \SP\Gm$) is balanced.

The following lemma relates  rectangularity and balancedness.

\begin{lemma}[\cite{effective-Dyer-doi:10.1137/100811258}, Lemma 29,30]
Strong balancedness implies strong rectangularity, but strong rectangularity does not imply strong balancedness.
\end{lemma}

Let $\rel(x, y, z)$ be a ternary relation on $A_1 \times A_2 \times A_3$. The relation $\rel$, is $p$-balanced if the matrix $\sM_\rel \in\bbZ_p^{|A_1|\tm|A_2|}$ given by
\begin{equation}\label{eq:balancedness-def}
\sM_\rel[x, y] \equiv |\{(x,y,z) \in H^3\mid (x, y, z) \in \rel\}| \pmod{p},
\end{equation}
for $x\in A_1$ and $y\in A_2$, is a rank-1 block matrix (the rank is computed in $\bbZ_p$). A relation $\rel$ on $H$ of arity $n > 3$ is $p$-balanced if every representation of $\rel$ as a ternary relation, a subset of $H^k \times H^\ell \times H^{n-k-\ell}$, is $p$-balanced. A relational structure $\cH$ (a constraint language $\Gm$) is called strongly $p$-balanced if every relation $\rel \in \SP \cH _p$ ($\rel \in \SP \Gm _p$) is $p$-balanced. 

The case $p=2$ is somewhat special.

\begin{lemma}\label{lem:2-rec-vs-2-bal}
A relational structure $\cH$ is strongly 2-rectangular if and only if it is strongly 2-balanced.
\end{lemma}
\begin{proof}
Suppose that $\SP\cH_2$ contains a non-rectangular $n$-ary relation $\rel$. Since $\SP\cH_2$ is closed under renaming coordinates, we may assume that for some $k<n$ there are $\ba,\bb\in\pr_{[k]}\rel$ and $\bc,\bd\in\rel_{[n]-[k]}\rel$ such that $(\ba,\bc),(\ba,\bd),(\bb,\bc)\in\rel$, but $(\bb,\bd)\not\in\rel$. Consider the relation 
\[
\relo(\vc xn,x_{n+1})=\rel(\vc xn)\wedge(x_n=x_{n+1}).
\]
Clearly, $\relo\in\SP\cH_2$. For the matrix $\sM_\relo$ constructed for $k$ and $\ell=n$ we have $\sM_\relo[\ba,\bc]=\sM_\relo[\ba,\bd]=\sM_\relo[\bb,\bc]=1$ and $\sM_\relo[\bb,\bd]=0$, showing that $\sM_\relo$ is not rank-1 block matrix.

Conversely, suppose that an $n$-ary relation $\rel\in\SP\cH_2$ is not $2$-balanced, that is, for some $k<\ell<n$ the matrix $\sM_\rel$ is not rank-1 block. Since the entries of $\sM_\rel$ are from $\bbZ_2$, it means that for some $\ba,\bb\in\pr_{[k]}\rel$ and $\bc,\bd\in\pr_{[\ell]}\rel$, it holds that $\sM_\relo[\ba,\bc]=\sM_\relo[\ba,\bd]=\sM_\relo[\bb,\bc]=1$ and $\sM_\relo[\bb,\bd]=0$. Therefore, the tuples $(\ba,\bc),(\ba,\bd),(\bb,\bc)\in H^\ell$ have an odd number of extensions to a tuple from $\rel$, while $(\bb,\bd)\in H^\ell$ has an even number of such extensions (including 0). Therefore the relation
\[
\relo(\vc x\ell)=\exists^{\equiv2} x_{\ell+1}\zd x_n\; \rel(\vc xn)
\]
belongs to $\SP\cH_2$ and is not rectangular.
\end{proof}

\vspace*{3mm}

\noindent\textbf{$p$-Permutability.}
A \emph{congruence} of a relational structure $\cH$ is an equivalence relation $\theta$ on $H$ that is definable by a pp-formula in $\cH$. More generally, let $\rela\in\SP\cH$ be a $k$-ary relation. A congruence of $\rela$ is a $2k$-ary relation pp-definable in $\cH$ that is an equivalence relation on $\rela$. 
We denote the set of all congruences of $\cH$ (of $\rela$) by $\Con(\cH)$ (respectively, by $\Con(\rela)$). By $\circ$ we denote the product of binary relations: $(a,b) \in \rel \circ \relo$ if and only if there is $c$ such that $(a, c) \in \rel$ and $(c, b) \in \relo$. In other words, the product of $\rel$ and $\relo$ is given by the pp-definition $\exists z\ \rel(x,z)\meet\relo(z,y)$. We say $\cH$ is \emph{congruence permutable} if for all $\alpha, \beta \in \Con(\rel)$, where $\rel\in\langle \cH \rangle$,  it holds that $\alpha \circ \beta = \beta \circ \alpha$. 

Congruence $p$-permutability is defined as follows: A \emph{$p$-congruence} of $\cH$ or of $\rela\in\SP\cH_p$ is an equivalence relation on $H$ or $\rela$, respectively, that is $p$-mpp-definable in $\cH$. We denote the set of all $p$-congruences of $\cH$ (respectively, $\rela$) by $\Con_p(\cH)$ (respectively, $\Con_p(\rela)$). By $\circp$ we denote the product of binary relations: $(a, b) \in \rel \circp \relo$ if and only if the number of elements $c$ such that $(a, c) \in \rel$ and $(c, b) \in \relo$ is not a multiple of $p$. The product of $\rel$ and $\relo$ is given by the pp-definition $\exists^{\equiv p} z\ \rel(x,z)\meet\relo(z,y)$. We say that $\cH$ is \emph{congruence $p$-permutable} if for all $\alpha, \beta \in \Con_p(\rela)$, where $\rela\in\langle \cH \rangle_p$, it holds that $\alpha \circp \beta = \beta \circp \alpha$. In other words, for any $\ba,\bb \in\rela$, if the number of tuples $\bc\in\rela$ such that $(\ba,\bc)\in\al, (\bc,\bb)\in\beta$ is not equal to $0\pmod p$, then so is the number of tuples $\bd\in\rela$ with $(\bb,\bd)\in\al, (\bd,\ba)\in\beta$.

\vspace*{3mm}

The regular versions of strong rectangularity, congruence permutability, and the existence of Mal'tsev polymorphism are known to be equivalent.

\begin{proposition}[\cite{Hagemann73:permutable}, see also \cite{effective-Dyer-doi:10.1137/100811258,DBLP:journals/jacm/Bulatov13}]\label{pro:3-properties-equivalence}
For a relational structure $\cH$ the following are equivalent:\\[2mm]
(1) $\cH$ has a Mal'tsev polymorphism;\\[1mm]
(2) $\cH$ is strongly rectangular;\\[1mm]
(3) $\cH$ is congruence permutable.
\end{proposition}

For the modular versions of the same properties the picture is much more complicated. 

%%%%%%%%%%%%%%%%%%%%%%%%%%%%%%%%%%%%%%%%%%%%%%%%%%%%%%%%%%%%%%%%%%%%%
\subsection{Contrasting Properties: Modular vs. Exact Counting.}
Although Proposition~\ref{pro:3-properties-equivalence} provides significant insights for exact counting, the modular of those conditions are almost unrelated. The following lemma represents the only connection between these notions.

\begin{lemma}\label{lem:p-per-to-p-rect}
    If $\cH$ is congruence $p$-permutable, then it is strongly $p$-rectangular.
\end{lemma}

\begin{proof}
Suppose $\cH$ is congruence $p$-permutable and consider a $p$-mpp-definable relation $\rel \subseteq H^r \times H^s$. We need to show that $\rel$ is $p$-rectangular. Define a relation $\sim_1$ on $\rel$ by $(\vv x_1, \vv y_1) \sim_1 (\vv x_2, \vv y_2)$ if and only if  $\vv x_1 = \vv x_2$. This relation is $p$-mpp-definable in $\cH$, by 
\[
(\vv x_1, \vv y_1) \sim_1 (\vv x_2, \vv y_2) := \rel(\vv x_1, \vv y_1) \wedge \rel(\vv x_2, \vv y_2) \wedge (\vv x_1 =\vv x_2)
\]
It is easy to check that this is also an equivalence relation and its equivalence classes correspond to tuples $\vv x\in\pr_{[r]}\rel$. Hence it is a $p$-congruence of $\rel$. In a similar way, define a congruence $\sim_2$ on $\rel$ by $(\vv x_1, \vv y_1) \sim_2 (\vv x_2, \vv y_2)$ if and only if $\vv y_1 = \vv y_2$. Let $\alpha = \sim_1 \circp \sim_2$. 

Suppose $((\vv a, \vv b), (\vv c, \vv d)) \in \alpha$. Then there exists $(\vv u, \vv v) \in \rel$ such that $(\vv a, \vv b) \sim_1 (\vv u, \vv v)$ and $ (\vv u, \vv v) \sim_2 (\vv c, \vv d)$. Thus, $\vv a = \vv u $ and $ \vv d = \vv v$. Therefore, $(\vv a, \vv b), (\vv a, \vv d), (\vv c, \vv d) \in \rel$. Now, congruence $p$-permutability implies that $\sim_1 \circp \sim_2 = \sim_2 \circp \sim_1$. Hence $((\vv c, \vv d), (\vv a, \vv b)) \in \alpha$. Thus, there exists $(\vv u', \vv v') \in \rel$ such that $(\vv c, \vv d) \sim_1 (\vv u', \vv v')$ and $ (\vv u', \vv v') \sim_2 (\vv a, \vv b)$. Thus, $\vv c = \vv u' $ and $ \vv b = \vv v'$. Therefore, $(\vv a, \vv c), (\vv c, \vv b), (\vv b, \vv d) \in \rel$. The strong $p$-rectangularity follows.
\end{proof}

To demonstrate the lack of connections between Maltsev polymorphisms, rectangularity, and permutability in the modular case, we present three examples. 

\begin{example}\label{exa:p-rect--not-p-perm}
Let $H = \{a_1, a_2, a_3, a_4, a_5, a_6, a_7 \}$. We define a relational structure $\cH$ over the set $H$ with the following relations:
\begin{itemize}
   \item All constant relations $C_{a_i}$ for all $i \in [7]$,
   \item The equivalence relation $\rel$ with equivalence classes ${a_1}/_\rel = \{a_1, a_2, a_3 \}$ and ${a_4}/_\rel = \{ a_4, a_5, a_6, a_7 \}$,
   \item The equivalence relation $\relo$ with equivalence classes ${a_1}/_\relo = \{a_1, a_2, a_4, a_5 \}$ and ${a_3}/_\relo = \{ a_3, a_6, a_7 \}$.
\end{itemize}
The structure $\cH = (H; \rel, \relo, C_{a_1}, ..., C_{a_7})$ is $2$-rigid due to the presence of all constant relations. It is not congruence $2$-permutable, since $(a_1, a_6) \in \rel \circp \relo$ but $(a_1, a_6) \not \in \relo \circp \rel$. We now show that $\cH$ is strongly $2$-rectangular.

Suppose an $n$-ary $\rela\in\SP\cH_2$, $k\in[n]$, and $\ba,\bb\in\pr_{[k]}\rela, \bc,\bd\in\pr_{[n]-[k]}\rela$ witness that $\rel$ is not rectangular, that is, $(\ba,\bc),(\ba,\bd),(\bb,\bc)\in\rela$, but $(\bb,\bd)\not\in\rela$. Consider a 2-mpp-definition of $\rela$ 
\[
\rela(\vc xn)=\exists^{\equiv2}\vc ym\Phi(\vc xn,\vc ym).
\]
Note that all the equivalence classes of $\rel\cap\relo$ except for $\{a_3\}$ contain 2 elements. This means that as $\ext_\Phi(\ba,\bc)$ is odd, it has to be that $\Phi(\ba,\bc,a_3\zd a_3)$ is true. The same holds for $\Phi(\ba,\bd,a_3\zd a_3),\Phi(\bb,\bc,a_3\zd a_3)$, but $\Phi(\bb,\bd,a_3\zd a_3)$ is false. This is however impossible. Indeed, as all the predicates of $\Phi$ of the form $\rel(x_i,x_j),\relo(x_i,x_j),\rel(y_i,y_j)$, and $\relo(y_i,y_j)$ are true on $(\bb,\bd,a_3\zd a_3)$ (in the former two cases due to the rectangularity of $\rel,\relo$), it has to be that $\Phi(\bb,\bd,a_3\zd a_3)$ is false implies that some predicate $\rel(x_i,y_j)$ or $\relo(x_i,y_j)$ is false. Suppose $i\in[k]$ and $\rel(x_i,y_j)$ is false. This means that $\rel(\bb[i],a_3)$ is false, which contradicts the assumption $\Phi(\bb,\bc,a_3\zd a_3)$ is true.
\end{example}

\begin{example}[Example~\ref{exa:bad-one} re-stated]\label{exa:maltsev--not-p-perm-rect}
Let $H=\{0,1,2,3,4\}$, and $\cH=(H;\rel, C_0, C_1, C_2, C_3,C_4)$, where $\rel$ is the ternary relation defined as
\[
\rel = \{ (0,0,0), (0,1,1), (1,0,2), (1,0,3), (1,1,4) \}.
\]
The constants $C_0\zd C_4$ correspond to $0\zd 4$ respectively. Thus, $\cH$ is 2-rigid.

The formula $\exists^{\equiv 2} z\rel(x,y,z)$ define the relation $\{(0,0),(0,1),(1,1)\}$, which is not rectangular. Hence, $\cH$ is not strongly 2-rectangular. As shown in Example~\ref{exa:bad-one}, $\cH$ has a Mal'tsev polymorphism. Therefore, by Proposition~\ref{pro:3-properties-equivalence}, it is strongly rectangular.

\end{example}

\begin{example}\label{exa:p-perm-rect--not-maltsev}
Let $H = \{a_1, a_2, a_3, a_4, a_5, a_6 \}$. We define a relational structure $\cH$ over the set $H$ with the following relations:
\begin{itemize}
    \item All constant relations $C_{a_i}$ for all $i \in [6]$,
    \item The equivalence relation $\rel$ with equivalence classes $a_1/_\rel = \{a_1, a_2, a_3 , a_4\}$ and $a_5/_\rel = \{ a_5, a_6 \}$,
    \item The equivalence relation $\relo$ with equivalence classes $a_1/_\relo = \{a_1, a_2, a_5, a_6 \}$ and $a_3/_\relo = \{ a_3, a_4 \}$.
\end{itemize}
The structure $\cH = (H; \rel, \relo, C_{a_1}, \ldots, C_{a_6})$ is $2$-rigid due to the presence of all constant relations. According to Proposition~\ref{pro:3-properties-equivalence}, it does not have a Mal'tsev polymorphism because $\rel \circ \relo$ is not rectangular. We only need to demonstrate that it is congruence 2-permutable. By Lemma~\ref{lem:p-per-to-p-rect}, this would imply that $\cH$ is also strongly 2-rectangular. Observe that $a_1$ and $a_2$, $a_3$ and $a_4$, $a_5$ and $a_6$ are interchangeable. That is, if $\rel' \in \SP \cH _2$ such that $\rel' \subseteq H^n$, $n>1$, and $(a_i, \vv{b}) \in \rel'$, where $\vv{b} \in H^{n-1}$, then $(a_{i+1}, \vv{b}) \in \rel'$ for all $i \in [3]$, unless $|\pr_1\rel'|=1$. Now, let $\relo'$ and $\rela'$ be congruences of $\rel'\in\SP \cH _2$. Then 
\[
(\relo' \circp \rela')(\bx, \by) = \exists^{\equiv2} \bz \, (\relo'(\bx, \bz) \wedge \rela'(\bz, \by)).
\]
By the simple fact stated earlier, for any $(\bx, \by)$, the number of extensions is always even which is zero modulo $2$. Hence, $\relo' \circp \rela' = \rela' \circp \relo' = \emptyset$, unless $|\rel'|=1$.
\end{example}

\begin{figure}[t]
    \centering
    \begin{subfigure}[t]{\textwidth}
        \centering
        \includegraphics[width=0.65\textwidth]{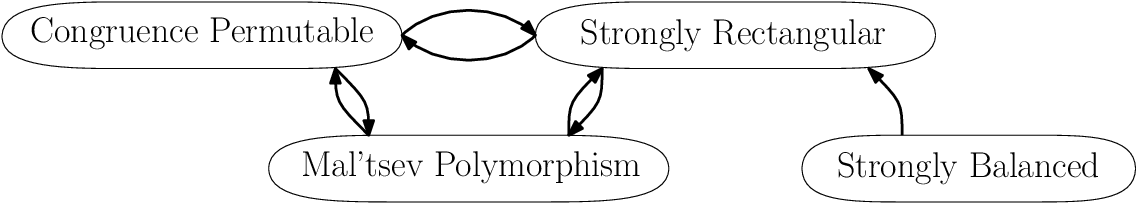}
        \caption{congruence permutability, strong rectangularity, and the existence of a Mal'tsev polymorphism are equivalent (see \cite{Hagemann73:permutable, effective-Dyer-doi:10.1137/100811258,DBLP:journals/jacm/Bulatov13}). Also, strong balancedness implies strong rectangularity (see \cite{effective-Dyer-doi:10.1137/100811258}).}
        \label{fig:a}
    \end{subfigure}
    \begin{subfigure}[t]{\textwidth}
        \centering
        \includegraphics[width=0.65\textwidth]{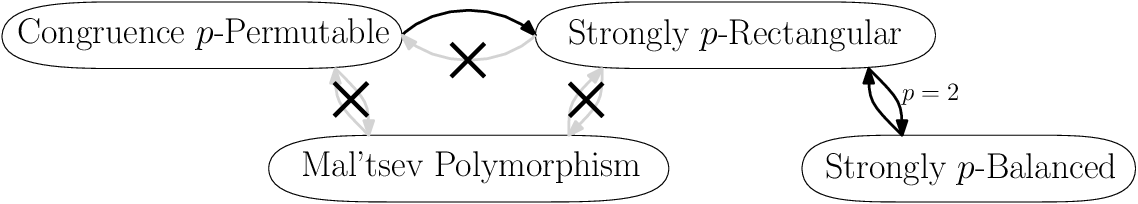}
        \caption{The only connection that is preserved for the modular case is that congruence $p$-permutability implies strong $p$-rectangularity. Also, strong 2-rectangularity is equivalence to strong 2-balancedness.}
        \label{fig:b}
    \end{subfigure}

    \caption{The connection between the four properties is shown above. Figure~(\ref{fig:a}) shows the connection in the case of exact counting. Figure~(\ref{fig:b}) shows the the connection for the modular counterparts.}
    \label{fig:relations-cp-sr-mp}
\end{figure}

%%%%%%%%%%%%%%%%%%%%%%%%%%%%%%%%%%%%%%%
\section{An Algorithm for Parity}\label{sec:parity}

In this section we show how to solve $\#_2\CSP(\cH)$ for a relational structure $\cH$ that is 2-rigid, strongly 2-rectangular, and has a Mal'tsev polymorphism $\varphi$. Observe that an instance of $\#_2\CSP(\cH)$ can be viewed as a conjunctive formula over $\cH$, and the set of solutions is a conjunctive definable relation from $\SP \cH _2$. Therefore, problem we solve is: given a conjunctive definition of a relation $\rel$, find the parity of the number of elements in $\rel$. 

The main idea is to reduce the arity of $\rel$. In other words, we attempt to find a relation $\WT\rel$, still 2-mpp-definable in $\cH$, such that $\WT\rel \subseteq H^{n-1}$ and $|\WT\rel| \equiv |\rel| \pmod 2$. To accomplish this, we use witness functions first introduced in \cite{effective-Dyer-doi:10.1137/100811258}.

%%%%%%%%%%%%%%%%%%%%%%%%%%%%%%%%%%%%%%%%%%%%%%
\subsection{Frames and Witness Functions}
Suppose that $\rel$ is an $n$-ary relation with a Mal'tsev polymorphism $\vf$. For each $i \in [n]$ we define the following relation $\sim_i$ on $\proj_i \rel$: $a \sim_i b$ if there exist tuples $\vv x \in H^{i-1}$ and $\vv y _a, \vv y _b \in H^{n-i}$ such that $(\vv x ,a ,\vv y_a) \in \rel \text{ and } (\vv x ,b ,\vv y_b) \in \rel$.
For the case $i = 1$, we have $a \sim_1 b$ for all $a, b \in \proj_1 \rel$ because they share the common empty prefix $\varepsilon$. 
The following two results are straightforward corollaries of the rectangularity of $\rel$ and were used in~\cite{effective-Dyer-doi:10.1137/100811258}.

\begin{lemma}[Folklore]
$\sim_i$ is an equivalence relation for all $i \in [n]$.
\end{lemma}
 
Let $\cE_i$ denote the collection of the equivalence classes of $\sim_i$, and $\cE_i=\{\cE_{i,1}\zd\cE_{i,\ell_i}\}$, $\cE_{i,j}\sse\proj_i \rel$, where $j\in[\ell_i]$. We often refer to these classes as \emph{frame classes}. 

\begin{lemma}[Folklore]
If $a \sim_i b$ and $\vv x \in \rel$ with $x_i = a$, then there is a $\vv y \in \rel$ with $\vv y_i = b$ and $\proj_{[i-1]} \vv x = \proj_{[i-1]} \vv y$.
\end{lemma}

%%%%%%%%%%%%%%%%%%%%%%%%%%%%%%%%%%%%%%%%%%%%%%%%%%%%%%%%%%%%%%%%%%%%%%%%%

A mapping $\FUNC{\omega}{[n] \times H}{ H^n \cup \{\bot \}}$ is called a witness function of $\rel$ if
\begin{itemize}
    \item[i.] For any $i \in [n]$ and $a \not \in \proj_i \rel$,  $\omega(i, a) = \bot$;
    \item[ii.] For any $i \in [n]$ and $a \in \proj_i \rel$,  $\omega(i, a) \in \rel$ is a witness for $(i, a)$, i.e., $\proj_i \omega(i,a) = a$;
    \item[iii.] For any $i \in [n]$ and $a, b \in \proj_i \rel$ with $a \sim_i b$, we have $\proj_{[i-1]} \omega(i, a) = \proj_{[i-1]} \omega (i, b)$.
\end{itemize}
A witness function $\omega$ provides a concise representation of $\rel$. Let $F=\{\omega(i,a)\mid i\in[n],a\in\pr_i\rel\}$. Such a set of tuples is called a \emph{frame} of $\rel$.
A witness function (or a frame) can be found in polynomial time given a conjunctive definition of a relation in a relational structure with a Mal'tsev polymorphism. This is the property that makes them essential for solving CSPs. 

\begin{proposition}[\cite{effective-Dyer-doi:10.1137/100811258}]
Let $\cH$ be a relational structure. If $\cH$ has a Mal'tsev polymorphism, then the witness function for the relation $\rel(x_1\zd x_n)=\bigwedge_{i\in[m]} \rel_i(x_{i_1}\zd x_{i_t})$ can be computed in $O(mn^4)$. 
\end{proposition}

We will also use other algorithmic properties of frames.

\begin{proposition}[\cite{effective-Dyer-doi:10.1137/100811258}]\label{pro:Frame_constant_existance}
Given a frame $F$ for $\rel(x_1,x_2,...,x_n)$, a frame for $\rel(x_1, x_2,..., x_n)\wedge C_a(x_s)$, i.e., $x_s$, $s\in[n]$, is the constant $a$, can be constructed in $O(n^2)$ time.
\end{proposition}

We denote the relation $\rel(x_1, ..., x_n) \wedge C_a(x_s)$ by $\rel^{s \leftarrow a}$, its $\sim_j$ relations by $\sim^{s \leftarrow a}_j$, its frame classes by $\cE^{s \leftarrow a}_{s,t}$, and its witness function by $\omega^{s\leftarrow a}$.

\begin{proposition}[\cite{effective-Dyer-doi:10.1137/100811258}]\label{pro:Frame_constant_string_existance}
Let $I \subseteq [n]$. Given a frame $F$ for $\rel(x_1,x_2,...,x_n)$, a frame for $\rel(x_1, x_2,..., x_n)\wedge ( \bigwedge_{s \in I} C_{a_s}(x_s))$, i.e., $x_s$ is the constant $a_s \in H$, $s\in I$, can be constructed in $O(n^3)$ time.
\end{proposition}

\begin{lemma}\label{lem:Frame_add_const}
Let $i\in [n-1]$, $a,b,c \in H$. For all $s+1 \leq i\leq n$, if $b \sim^{s \leftarrow a}_i c$, then $b \sim_i c$.
\end{lemma}

\begin{proof}
The proof is straightforward from the definition. Suppose $b \sim^{s\leftarrow a}_i c$. Then, there exist $\vv x\in H^{i-1}$ and $\vv y_b, \vv y_c \in H^{n-i}$ such that $(\vv x, b, \vv y_b), (\vv x, c, \vv y_c)\in \rel^{s\leftarrow a}\sse\rel$ and $\pr_s \vv x=a$. By the definition, $b \sim_i c$.
\end{proof}

\begin{proposition}[\cite{effective-Dyer-doi:10.1137/100811258}, \cite{DBLP:journals/jacm/Bulatov13}]\label{pro:Frame_projection}
Given a frame $F$ for $\rel(x_1,x_2,...,x_n)$, a frame for $\relo(x_1, x_2..., x_{n-1}) = \exists y \rel(x_1, ..., x_{n-1}, y)$, can be constructed in $O(n)$ time.
\end{proposition}

%%%%%%%%%%%%%%%%%%%%%%%%%%%%%%%
\subsection{The Algorithm}
In order to find $\WT\rel$ as explained in the beginning of Section~\ref{sec:parity} we use a witness function and frame of $\rel$ to compute a frame and witness function of $\WT\rel$. If we can find such a $\WT\rel$, we repeat this process $n-1$ times. Eventually, we obtain $\WT\rel^* \subseteq H$, whose cardinality can be easily found. 

First, we assume that the witness function for $\rel$ and $\WT\rel$ is given. We prove we prove the following.

\begin{proposition}\label{pro:parity-algorithm}
Let $\cH$ be a 2-rigid, strongly 2-rectangular, $\SP \cH_2$ has Mal'tsev polymorphism , and let $\rel \in \SP\cH _2$ be an $n$-ary relation. Given a frame $F$ and a witness function $\omega$ for $\rel$, the parity of $\rel$ can be computed in $O(n)$ time.
\end{proposition}

In Section~\ref{sec:calculate-witness-function} we propose an efficient method to compute a witness function in $O(n^4)$. Consequently, the main outcome of this section is summarized as follows:

\begin{theorem}[Theorem~\ref{the:parity-algorithm-overview} re-stated]
Let $\cH$ be a 2-rigid, strongly 2-rectangular, and $\SP \cH_2$ has Mal'tsev polymorphism. And let $\rel \in \SP\cH _2$ be an $n$-ary relation. Then $\#_2\CSP(\cH)$ can be solved in time $O(n^5)$. 
\end{theorem}

%%%%%%%%%%%%%%%%%%%%%%%%%%%%%%%%%%%%
\subsubsection{Auxiliary relations}

We define the following relation as a tool in out proofs. For $\rel \in \SP \cH _2$, let $\PAR\rel$ be defined as follows:
\begin{equation} \label{equ:PAR_definition}
\PAR\rel(\vv x, y) = \rel(\vv x, y) \wedge (\exists^{\equiv2} z \;\; \rel(\vv x, z))
\end{equation}
Note that the relation $\PAR\rel$ is 2-mpp-definable in $\cH$. Thus, $\PAR\rel \in \SP \cH _2$, and it is rectangular. 
Also, note that the second part of the conjunct expresses the fact that the number of extensions of $\vv x$ should be odd. Hence, $(\vv x, y) \in \PAR\rel$ if and only if $(\vv x, y) \in \rel$ and $\vv x$ has an odd number of extensions, i.e., $\ext_\rel(\vv x) \equiv 1 \pmod 2$. 

\begin{lemma}\label{lem:PAR_size}
Let $\rel \in \SP \cH _2$. Then $|\rel| \equiv |\PAR\rel| \pmod 2$.
\end{lemma}

\begin{proof}
Let $\sgn$ be the characteristic function for a relation $\relo\sse H^n$: 
\begin{equation*}
\sgn_\relo(\vv x)=
\left\{ 
\begin{array}{ll}
   1  & \vv x\in \relo, \\
   0  & \vv x\in H^n-\relo. 
\end{array}\right. 
\end{equation*}
We start by calculating the size of $\rel$ as follows. Recall that $\prpar_{[n-1]}\rel$ denotes the relation defined by $\exists^{\equiv2}y\, \rel(\vv x,y)$.
\begin{align*}
    |\rel| &= \sum_{\vv x\in H^n} \sgn_\rel(\vv x)\\
     &= \sum_{\vv y \in H^{n-1}} \sum_{z \in H} \sgn_\rel(\vv y, z) \\
    &= \sum_{\vv y \in H^{n-1}} \ext_\rel (\vv y)\\
    &= \sum_{\substack{\vv y \in H^{n-1} \\ \vv y \in \prpar_{[n-1]}\rel}} \ext_\rel (\vv y) + \sum_{\substack{\vv y \in H^{n-1} \\ \vv y \not \in \prpar_{[n-1]}\rel}} \ext_\rel(\vv y)\\
    &\equiv \bigoplus_{\substack{\vv y \in H^{n-1} \\ \vv y \in \prpar_{[n-1]}\rel}} \ext_\rel(\vv y) \pmod 2\\ 
    &\equiv |\PAR\rel| \pmod 2
\end{align*}
The first and second equalities are trivial by the definition of the function $\sgn$. The third equality holds because for a $\vv y \in H^{n-1}$ 
\[
    \ext_\rel(\vv y) = \sum_{z \in H} \sgn((\vv y, z) \in \rel).
\]
This equation indicates that the number of extensions of $\vv y \in H^{n-1}$ to a tuple from the $n$-ary relation $\rel$ is determined by appending each possible $z \in H$ to $\vv y$ and verifying whether $(\vv y, z)$ belongs to $\rel$.  
The forth equality is true, because we can split all $\vv y \in \pr_{[n-1]} \rel$ into two disjoint sets: one of tuples with an odd number of extension, and the other of tuples with an even number of extension.
Then, the first equivalence holds, because if $\vv y\in \pr_{[n-1]}\rel$ and $\vv y\not \in \prpar_{[n-1]}\rel$, then, $\ext_\rel(\vv y) \equiv 0 \mod 2$.
Finally, the last equivalence is valid, because $(\vv x, y)\in \PAR\rel$ if and only if $\ext_\rel(\vv x) \equiv \ext_{\PAR\rel} (\vv x) \equiv 1 \pmod 2$. 
\end{proof}

Now, the relation $\WT \rel\subseteq H^{n-1}$ we have been talking about since the beginning Section~\ref{sec:parity} is given by 
\begin{equation}\label{TILDE_definition}
    \WT\rel(\vv x) = \exists^{\equiv2} y\;\; \PAR\rel(\vv x, y)
\end{equation}
Note that the relation $\WT\rel$ is 2-mpp-definable in $\cH$. Thus, $\WT\rel \in \SP \cH _2$, and it is rectangular. Also, by the definition of $\PAR\rel$, $\vv x \in \WT\rel$ if and only if $\vv x$ has odd number of extensions into $\rel$. Therefore, the following lemma is straightforward.

\begin{lemma}\label{lem:PAR-exists}
Let $\relo(\vv x) = \exists y \; \; \PAR\rel(\vv x, y)$. Then, $\relo = \WT\rel$. 
\end{lemma}

\begin{proof}
%Let $\vv x \in \pr_{[n-1]} \PAR \rel$, if and only if, there exists $y \in H$ such that $(\vv x, y) \in \PAR\rel$, if and only if, $\vv x$ has odd number of extensions into $\PAR\rel$, if and only if $\vv x$ has odd number of extensions into $\rel$, if and only if $\vv x \in \prpar_{[n-1]}\rel$. 
We have that $\vv x \in \pr_{[n-1]} \PAR \rel$ if and only if there exists $y \in H$ such that $(\vv x, y) \in \PAR\rel$. Furthermore, by the definition of $\PAR\rel$, it means $\vv x$ has an odd number of extensions to a tuple from $\PAR\rel$, which is equivalent to having an odd number of extensions to a tuple from $\rel$. In other words, $\vv x$ belongs to $\prpar_{[n-1]}\rel$. The result follows.
\end{proof}

\begin{lemma}\label{lem:TILDE_size}
Let $\rel \in \SP \cH _2$. Then $|\WT \rel| \equiv  |\PAR\rel| \pmod 2$.
\end{lemma}
\begin{proof}
The cardinality of $\WT\rel$ is as follows:
\begin{align*}
    |\WT\rel| &= \sum_{\vv x\in H^{n-1}} \sgn_{\WT\rel}(\vv x)\\
%%     &= \sum_{\substack{\vv x \in H^{n-1} \\ y \in H}} \sgn_{\rel(\vv x, y \\
    &\equiv \bigoplus_{\vv x \in H^{n-1}} \bigoplus_{y \in H} \sgn_\rel(\vv x, y) \pmod2\\
&\equiv \bigoplus_{\vv x \in H^{n-1}} \bigoplus_{y \in H} \sgn_{\PAR\rel}(\vv x, y )\pmod2\\
    &\equiv |\PAR\rel|\pmod2
\end{align*}
The first equality is straightforward. The first congruence is true as it follows directly from the definition of the relation $\WT\rel$: $\vv x \in \WT\rel$ if and only if $\vv x$ has odd number of extensions in $\rel$. 
%% , which is if and only if there exists $y \in \prpar_n \rel$ such that $(\vv x, y) \in \rel$. 
%% The first congruence holds because $y \in \prpar_n\rel$ if and only if $ \bigoplus_{y\in H} \sgn( (\vv x, y) \in \rel) \equiv 1 \pmod 2$.
The second congruence is valid because $\vv x\in\WT\rel$ has odd number of extensions, hence $(\vv x, y)$ should belong to $\PAR\rel$ for each $(\vv x, y)\in\rel$.
\end{proof}

\begin{proof}[Proof of Proposition~\ref{pro:parity-algorithm}]
Our goal is to reduce the arity of $\rel$ from $n$ to $1$ step by step while keeping the parity of $|\rel|$ unchanged. To achieve this, 
we define $\rel^{(n)}=\rel$ and $\rel^{(k-1)}=\WT\rel^{(k)}$ for $1 \le k \le n$. 
By applying Lemma~\ref{lem:PAR_size} and Lemma~\ref{lem:TILDE_size} to $\rel^{(k)}$ we have the following congruences:
\begin{equation*}
|\rel| \equiv |\rel^{(n)}| \equiv |\rel^{(n-1)}| \equiv ... \equiv |\rel^{(1)}| \pmod 2.
\end{equation*}
By Proposition~\ref{pro:calculate-witness} we can compute a witness function of $\rel^{(k-1)}$ denoted by $\omega^{(k-1)}$ for $n \geq k \geq 1$ from a witness function $\omega^{(k)}$ of $\rel^{(k)}$.
Finally, the cardinality of $\rel^{(1)}\subseteq H$ can be computed using its witness function $\omega^{(1)}$: We simply check if $\omega^{(1)}(1, a)\neq\bot$ for each $a\in H$. Note that the whole process takes just $O(n)$ times, not considering the complexity of calculating witness functions.
\begin{algorithm}[H]
\caption{$\mathsf{CalculateSize}$($\omega, \cE , n$)}\label{alg:FindHandle}
    \begin{algorithmic}[1]
    \If {$n=1$}
    \State $\ell \gets $ Count how many $a \in H$, $\omega(1,a) \neq \bot$
    \State \textbf{Return} $\ell\pmod2$
    \EndIf
    \If {for all $t \in [|\cE_n|]$, $|\cE_{n, t}|\equiv 0 \pmod 2$}
    \State \textbf{Return} 0
    \EndIf
    \State $\WT\omega, \WT\cE \gets \mathsf{FindWitnessFunction}(\omega, \cE)$
    \State \textbf{Return} $\mathsf{CalculateSize}$($\WT\omega, \WT\cE , n-1$) 
    \end{algorithmic}
\end{algorithm}
\end{proof}

%%%%%%%%%%%%%%%%%%%%%%%%%%%%%%%%%%%%%%%%%%%%%%%%%%%%
\subsubsection{Calculating witness function for $\WT \rel$}\label{sec:calculate-witness-function}
In this subsection, we show how to calculate a frame $\WT F$, the frame classes $\WT \cE$, and a witness function $\WT\omega$ for $\WT\rel$. Since $\PAR\rel \in \SP \cH_2$ is rectangular, we can consider its frame and witness function, denoted as $\cE'$ and $\omega'$. Suppose we have already calculated the frame and witness function for $\PAR\rel$. By applying Lemma~\ref{lem:PAR-exists} and Proposition~\ref{pro:Frame_projection}, we can obtain a witness function and frame for $\WT\rel$ in $O(n)$ time.

To calculate a witness function for $\PAR\rel$, we use two auxiliary lemmas. Once we have a witness function for $\PAR\rel$, we can proceed with calculating a witness function for $\WT\rel$.

\begin{lemma}\label{lem:tool_PAR}
Let $(\vv x, y)\in \rel$  where $\vv x \in H^{n-1}$ and $y \in H$. If there exists $s\in [|\cE_n|]$ such that $|\cE_{n, s}| \equiv 1 \pmod 2$ and $y \in \cE_{n, s}$, then $(\vv x, y) \in \PAR\rel$.
\end{lemma}

\begin{proof}
We need to calculate the number of extensions of $\vv x$.
\begin{align*}
    \ext_\rel(\vv x) &= \sum_{z \in H} \sgn_\rel(\vv x, z)\\
    &= \sum_{z \in \pr_n \rel} \sgn_rel(\vv x , z)\\
    &= \sum_{i \in [|\cE_{n}|]} \sum_{ z\in \cE_{n, i}} \sgn_\rel(\vv x , z)\\
    &= \sum_{ z\in \cE_{n, s}} \sgn_\rel(\vv x , z)\\
    &\equiv |\cE_{n,s}| \equiv 1 \pmod 2.
\end{align*}
The first and second equalities are straightforward, as they follow directly from the definition of the function $\sgn$. The third equality holds because we can partition the set $\pr_n\rel$ into equivalence classes given by the frame $F$. The fourth equality is valid because the relation $\rel$ is rectangular, and so $\vv x$ can be extended to only one class $\cE_{n,s}$.
Since there exists $s \in [|\cE_n|]$ such that $y \in \cE_{n,s}$, the number of extensions of $\vv x$ should be the cardinality of $\cE_{n,s}$. As $|\cE_{n,s}|\equiv 1$, the number of extensions of $\vv x$ is $1$ modulo $2$. Therefore, $(\vv x, y) \in \PAR\rel$.
\end{proof}

\begin{lemma}
Let $a, b \in H$, $k\in[n]$, and $\vv x \in H^n$ be such that $\vv x \in \PAR\rel$ and $\pr_k \vv x = a$. We can check in $O(n^3)$ if $a \sim'_k b$, where $\sim'_k$ is the frame equivalence of $\PAR\rel$.
\end{lemma}

\begin{proof}
As $\PAR\rel \in \SP \cH_2$, we know that $\pr_{[k]}\PAR\rel$ is rectangular. Therefore, if $(\pr_{[k-1]}\vv x, a) \in \pr_{[k]} \PAR\rel$ and $a \sim'_k b$, then $(\pr_{[k-1]} \vv x, b) \in \pr_{[k]}\PAR\rel$. This implies the existence of an extension $(\vv y, c) \in H^{n-k}$ such that $(\pr_{[k-1]}\vv x, a, \vv y, c) \in \PAR\rel$. According to Lemma~\ref{lem:tool_PAR}, we can check if $(\pr_{[k-1]}\vv x, a, \vv y, c) \in \PAR\rel$ by verifying if $c$ belongs to an $\cE_n$ class of $\rel$ with odd cardinality.

We begin by defining the relation $\rela(x_1, ..., x_n)$ as follows:
\[
    \rela(x_1, ..., x_n) = \rel(x_1, ..., x_n ) \wedge ( \bigwedge_{i \in [k-1]} x_s \in C_{d_i}) \wedge (x_k \in C_b)
\]
Essentially, $\rela$ is the same as $\rel$ with its first $k-1$ coordinates fixed to the corresponding coordinates of $\vv x$, and the $k$th coordinate fixed to $b$. We can compute the witness function and frame classes of $\rela$ using Proposition~\ref{pro:Frame_constant_string_existance}, in time $O(n^3)$. We denote the resulting frame classes by $\cE''$. Next, we check if there exists $t \in [|\cE''_{n}|]$ such that $|\cE''_{n,t}|\equiv 1 \pmod 2$. If such a $t$ exists, then by Lemma~\ref{lem:Frame_add_const} and Lemma~\ref{lem:tool_PAR}, we can conclude that $a\sim'_k b$.

This process requires only $O(n^3)$ time to compute the witness function for $\rela$, and then $O(1)$ time to detect whether there exists such a $t$. Therefore, the overall running time of the algorithm is $O(n^3)$.
\begin{algorithm}[H]
\caption{$\mathsf{CheckEpsilonClass}$($\vv x \in H^n, a, b\in H, k \in [n]$)}\label{alg:EpsilonClass}
    \begin{algorithmic}[1]
    \State Construct $\rela$ as $\rela(x_1, ..., x_n) = \rel(x_1, ..., x_n ) \wedge ( \bigwedge_{i \in [k-1]} x_s \in C_{d_i}) \wedge (x_k \in C_b)$
    \State Calculate $\cE''$ classes of $\rela$ with its witness function $\omega''$
    \If {there exists $t \in [|\cE''_{n}|]$ such that $|\cE''_{n,t}|\equiv 1 \pmod 2$}
    \State set $\omega'(k, b) = \omega''(n, c)$, for $c \in \cE''_{n,t}$
    \State \textbf{Return} $\omega'(k, b)$
    \Else
    \State \textbf{Return} $\bot$
    \EndIf
    \end{algorithmic}
\end{algorithm}
Note that if the algorithm $\mathsf{CheckEpsilonClass}$ returns $\vv y \in H^n$ for the input $\vv x \in H$ and $a, b \in H$, then it follows that $\pr_k \vv y = b$ and $\pr_{[k-1]} \vv x = \pr_{[k-1]} \vv y$. This is because, in line (4) of the algorithm, $\omega'(k,b)$ is set to be $\omega''(n,c)$ and we know that $\pr_k \omega''(n,c) = b$ and $\pr_{[k-1]} \omega''(n,c )= \pr_{[k-1]} \vv x$, as per the definition of $\rela$.
\end{proof}

\begin{proposition}\label{pro:calculate-witness}
Given a witness function $\omega$ for $\rel$, a witness function $\WT\omega(k, a)$ for $\WT\rel$ can be found in time $O(n^4)$ for any $k\in[n]$ and $a\in H$,.
\end{proposition}

\begin{proof}
In the first step, we calculate a witness function $\omega'$ and frame $F'$ for $\PAR\rel$. Using Proposition~\ref{pro:Frame_projection} and Lemma~\ref{lem:PAR-exists}, we can then compute the witness function and frame for $\WT\rel$.

In order to compute $\omega'$, we first need to compute $\omega'(n,a)$ for all $a\in H$. This step is relatively simple. Since $a \in \pr_n \PAR\rel$ if and only if there exists $\vv x\in H^{n-1}$ such that $(\vv x, a) \in \PAR\rel$, we can check whether $\ext_\rel(\vv x) \equiv 1 \pmod 2$. By Lemma~\ref{lem:tool_PAR}, $a \in \pr_n \PAR\rel$ if and only if $a \in \cE_{n, s}$ where $s\in [|\cE_{n}|]$ and $|[\cE_{n,s}]| \equiv 1 \pmod 2$.
Clearly, if $a \sim_n b$, then $a \sim'_n b$. Therefore, we can set $\omega'(n,a) = \omega(n,a)$ if $|\cE_{n,s}| \equiv 1 \pmod 2$, and $\omega'(n,a) = \bot$ if $|\cE_{n,s}| \equiv 0 \pmod 2$. Note that this step only takes $O(1)$ time, since $\omega$ is already given to us.

Let $k\in [n-1]$. According to Proposition~\ref{pro:Frame_constant_existance}, we can compute $\omega^{k\leftarrow a}$ and $\cE^{k \leftarrow a}$ in $O(n^2)$ time. Next, we examine $\cE^{k \leftarrow a}_{n, i}$ for $i \in [|\cE^{k \leftarrow a}_{n}|]$. We check whether there exists an $s \in [|\cE^{k\leftarrow a}_{n}|]$ such that $|\cE^{k\leftarrow a}_{n, s}| \equiv 1 \mod 2$. If we find such an $s$, we select $b \in \cE^{k\leftarrow a}_{n, s}$ and set $\omega'(k, a) = \omega^{k\leftarrow a}(n, b)$. Otherwise, we assign $\omega'(k, a) = \bot$. Note that searching for $s$ only requires $O(1)$ time since we only need to look up $|\cE^{k \leftarrow a}_n|$ options, which is less than $|H|$. 
Since $\omega'(k, a) = \omega^{k \leftarrow a}(n,b)$, it follows that $\pr_k \omega'(k, a) = \pr_k\omega^{k \leftarrow a} (n,b)= a$. Also, because of the way we construct $\omega'(k,a)$ and Lemma~\ref{lem:tool_PAR}, we have that $\omega'(k, a) \in \PAR\rel$.

Next, we use Algorithm~$\mathsf{CheckEpsilonClass}(\omega'(k,a), a, b)$ for all $b \in H$ to derive the equivalence relation $\sim'_k$ defined by $\PAR\rel$ for the $k$th coordinate. If the algorithm returns $\vv y$, it means that $a \sim'_k b$ and they are in the same class. Therefore, we set $\omega'(k,b) = \vv y$. Otherwise, $a \not \sim'_k b$. We repeat this process for each $a \in H$ whose witness function has not yet been associated with anything, and at the end, we have the witness function for $\PAR\rel$. By Proposition~\ref{pro:Frame_projection}, we have the witness function for $\WT\rel$, as well.

\begin{algorithm}[H]
\caption{$\mathsf{FindWitnessFunction}(\omega, \cE)$}\label{alg:WitnessFunction}
    \begin{algorithmic}[1]
    \State Compute $\omega$ and $\cE$ for $\rel$.
    \For{\textbf{each} $a \in H$}
        \State Find $s \in \cE_{n}$ such that $a \in \cE_{n,s}$
        \If{ $|\cE_{n,s}| \equiv 1 \pmod 2$}
            \State $\omega'(n,a) \gets \omega(n,a)$
        \Else
            \State Set $\omega'(n,a) \gets \bot$
        \EndIf
    \EndFor
    
    \For{\textbf{each} $k \in [n-1]$}
        \State Set $D \gets H$
        \While{$D \not = \emptyset$}
            \For{\textbf{each} $a \in D$}
                \State Calculate $\omega^{k\gets a}$ and $\cE^{k \gets a}$
                \If{ there exists $s \in [|\cE^{k\gets a}_{n}|]$ such that $|\cE^{k\gets a}_{n,s}| \equiv 1 \pmod 2$}
                    \State Pick $b\in \cE^{k \gets a}_{n, s}$
                    \State Set $\omega'(k,a) \gets \omega(n, b)$
                    \For{\textbf{each} $ c \in D\setminus \{a \}$}
                        \State Set $\vv y \gets \mathsf{CheckEpsilonClass}(\omega'(k,a), a, c, k)$
                        \If{ $\vv y \not = \bot$}
                            \State Set $\omega'(k,b) \gets \vv y$
                            \State Set $D \gets D \setminus \{c\}$
                        \EndIf
                    \EndFor
                \Else
                    \State Set $\omega'(k,a) \gets \bot$
                \EndIf
            \EndFor
        \EndWhile
    \EndFor
    \State Calculate $\WT\omega$ and $\WT\cE$ based on $\omega'$ and $\cE'$ by projection \Comment{Proposition~\ref{pro:Frame_projection}}
    \State \textbf{Return} $\WT\omega$ and $\WT\cE$
    \end{algorithmic}
\end{algorithm}
\end{proof}

\begin{remark}
    Note that at no point the algorithm explicitly uses the condition of 2-rigidity, although it is among the conditions of Proposition~\ref{pro:parity-algorithm}. However, it uses constant relations in a very essential way, they are used to compute the witness function and frames. In general relational structures, the availability of constant relations is not guaranteed. However, by Theorem~\ref{the:ConstantCSP-main}, we can add constant relations to $\cH$ if $\cH$ is 2-rigid.
\end{remark}

%%%%%%%%%%%%%%%%%%%%%%%%%%%%%%%%%%%%%%%
%%%%%%%%%%%%%%%%%%%%%%%%%%%%%%%%%%%%%%%
\section{Hardness and Automorphisms}\label{sec:automorphisms}

%% \input{Data/h-automorphisms}
%%%%%%%%%%%%%%%%%%%%%%%%%%%%%%%%%%%%%%%
An important ingredient of the complexity classification of $\#_pCSP$ for graphs \cite{DBLP:conf/stoc/BulatovK22} is the structure of automorphisms of graph products, see \cite{ref:ProductOfGraphs} for non-bipartite graphs and \cite{DBLP:conf/stoc/BulatovK22} for bipartite graphs. We compare the situation for relational structures with the result for non-bipartite graphs stated below. 

%%%%%%%%%%%%%%%%%%%%%%%%%%%%%%%%%%%%%%%%%%%
\subsection{Automorphisms of structures and automorphisms of products}\label{sec:poduct-auto}
% the result for graphs
Let $G=(V,E)$ be a graph. The graph $G$ is said to be \emph{$\RT$-thin} if it does not have \emph{twin vertices}, that is, vertices with equal neighbourhoods. The graph $G$ is \emph{prime} if for any representation $G=G_1\tm G_2$ one of $G_1,G_2$ is a 1-element graph. A representation $G=G_1\tm\dots\tm G_r$ is said to be \emph{prime factorization of} of $G$ if all the $G_i$'s are prime.

\begin{theorem}[\cite{ref:ProductOfGraphs}, \cite{DBLP:conf/stoc/BulatovK22}]\label{the:splitting_automorphisms}
 Let $G=G_1\tm\dots\tm G_r$ be a prime factorization of a connected $\RT$-thin non-bipartite graph $\sG$, where for $v\in V(G)$, we have $v=(\vc{v}{r})$. Then for any automorphism $\psi$ of $\sG^{\ell}$, there is a permutation $\pi$ of $[\ell]\times [r]$ such that $\psi$ can be split into $r\ell$ automorphisms:
\begin{align*}
    \psi([v_1]\zd [v_\ell])= \left( \vvecad{\psi_{1,1}([v_{\pi(1,1)}])}{\psi_{1,r}([v_{\pi(1,r)}])},...,\vvecad{\psi_{\ell,1}([v_{\pi(\ell,1)}])}{\psi_{\ell,r}([v_{\pi(\ell,r)}])} \right),
\end{align*}
for $(i',j')=\pi(i,j)$, $\psi_{i,j}$ is an isomorphism from the $i'$th copy of $G_{j'}$ to the $i$th copy of $G_j$. 
\end{theorem}

% counterexample for structures
As the following example shows, an analogous result is not true for relational structures.

\begin{example}\label{exa:bad-graph}
Let $\cH = (V, E)$ be a directed graph where $V = \{a, b,c ,d \}$ with directed edges $E = \{ (b, a), (b,c), (c,d) \}$. If we view $\cH$ as a relational structure, then it is 2-rigid. However, the automorphism group of $\cH^2$ has a complicated structure. As is seen in Figure~\ref{fig:bad-example}, $\cH^2$ has the following automorphism: $\pi(a,d)=(c,d), \pi(c,d)=(a,d)$, and $\pi(x,y)=(x,y)$ otherwise.
This automorphism does not have the structure given in Theorem~\ref{the:splitting_automorphisms}.

\begin{figure}[H]
    \centering
    \includegraphics[height=3.5cm]{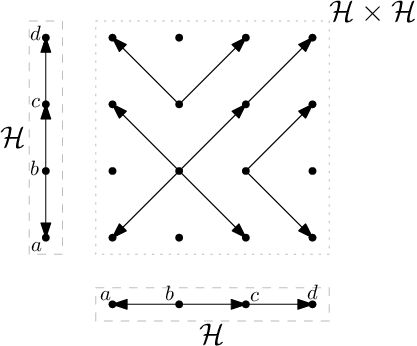}
    \caption{The structure of $\cH$ and $\cH^2$}
    \label{fig:bad-example}
\end{figure}
\end{example}

%%%%%%%%%%%%%%%%%%%%%%%%%%%%%%%%%%%%%%%
\subsection{Rectangularity obstruction}\label{sec:rectangularity-obstruction}

One of the implication of Theorem~\ref{the:splitting_automorphisms} is that some subsets from $G_1\tm\dots\tm G_r$ survive $p$-reductions by automorphism of order $p$. We explore what the existence of such sets entails.

Let $\cH$ be a (multi-sorted) relational structure with base set $H$, $n\geq 1$, and $\rel\sse H^n$. A subset $S \subseteq\rel$ is called \emph{$p$-protected} in $\rel$ if, after applying a $p$-reduction to $\rel$ under a sequence of $p$-automorphisms from $\Aut(\rel)$, the subset $S$ remains non-empty. Formally, $S$ is $p$-protected if for any sequence of automorphisms $\pi_1, \pi_2, \ldots, \pi_k \in \Aut(\rel)$, the image of $S$ under the $p$-reduction remains non-empty:
\[
S \cap \Fix(\pi_1 \circ \pi_2 \circ \cdots \circ \pi_k) \neq \emptyset.
\]

A \emph{rectangularity obstruction} in a relational structure $\cH$ is a violation of the rectangularity or $p$-rectangularity property. Let $\rel \subseteq H^n$ be an $n$-ary relation that is either pp-definable or $p$-mpp-definable in $\cH$. For some $k \in [n]$, let $\ba, \bb \in \pr_{[k]}\rel$ and $\bc, \bd \in \pr_{[n]-[k]}\rel$. The relation $\rel$ together with the triple $(\ba,\bc),(\ba,\bd),(\bb,\bc)$ is a rectangularity obstruction if $(\ba,\bc),(\ba,\bd),(\bb,\bc)\in\rel$ and $(\bb, \bc) \not\in \rel$. 
A \emph{generalized rectangularity obstruction} for the relation $\rel$ are sets $A_{1,1},A_{1,2}\sse\pr_{[k]}\rel$, $A_{2,1},A_{2,2}\sse\pr_{[n]-[k]}\rel$ such that $A_{1,1}\cap A_{1,2}=\emptyset$, $A_{2,1}\cap A_{2,2}=\emptyset$, and any $\ba\in A_{1,1},\bb\in A_{1,2},\bc\in A_{2,1},\bd\in A_{2,2}$ form a rectangularity obstruction. A rectangularity obstruction is $p$-protected if each of the sets $\{(\ba,\bc)\},\{(\ba,\bd)\},\{(\bb,\bc)\}$ is $p$-protected in $\rel$. A generalized rectangularity obstruction $A_{1,1},A_{1,2}\sse\pr_{[k]}\rel$, $A_{2,1},A_{2,2}\sse\pr_{[n]-[k]}\rel$ is $p$-protected if each of the sets $\rel\cap(A_{1,1}\tm A_{2,1}), \rel\cap(A_{1,1}\tm A_{2,2}), \rel\cap(A_{1,2}\tm A_{2,1})$ is $p$-protected in $\rel$.

A special case of a generalized rectangularity obstruction is a \emph{standard hardness gadget}. For a relational structure $\cH$, a standard hardness gadget is a tuple $(\rel, A_{1,1}, A_{1,2}, A_{2,1}, A_{2,2})$ such that
\begin{itemize}
   \item 
   $\rel \subseteq H^n$, $n \geq 2$, $p$-mpp-definable in $\cH$, i.e., $\rel \in \langle \cH \rangle_p$;
    \item  
    for some $s\in[n]$, $0<s<n$, $A_{1,1}, A_{1,2}\sse\pr_{[s]}\rel, A_{2,1},A_{2,2}\sse\pr_{[n]-[s]}\rel$ are non-empty;
    \item 
    $A_{1,1} \cup A_{1,2} = \pr_{[s]} \rel$, $A_{1,1} \cap A_{1,2} = \emptyset$, and 
    $A_{2,1} \cup A_{2,2} = \pr_{[n] - [s]} \rel$, $A_{2,1} \cap A_{2,2} = \emptyset$;
    \item 
    $A_{1,1}\tm A_{2,2}, A_{1,2}\tm A_{2,1}, A_{1,2}\tm A_{2,2}\sse\rel$;
    \item 
    $(A_{1,2}\tm A_{2,2})\cap\rel=\emptyset$. 
\end{itemize}
Thus, a standard hardness gadget is a generalized rectangularity obstruction with the extra condition $A_{1,1} \cup A_{1,2} = \pr_{[s]} \rel$, $A_{2,1} \cup A_{2,2} = \pr_{[n] - [s]} \rel$.

By the definition of a standard hardness gadget, the following lemma is straightforward.

\begin{lemma}\label{lem:struct-standard-gadget}
Let $\vv x_1, \vv x_2 \in A_{1,i}$ for $i \in [2]$. Then, if there exists $\vv y \in \pr_{[n]-[k]}\rel$ such that $(\vv x_1, \vv y) \in \rel$, then $(\vv x_2, \vv y) \in \rel$. Also, let $\vv y_1, \vv y_2 \in A_{2,i}$ for $i \in [2]$. Then, if there exists $\vv x \in \pr_{[k]}\rel$ such that $(\vv x, \vv y_1) \in \rel$, then $(\vv x, \vv y_2) \in \rel$.
\end{lemma}

For a standard hardness gadget we have the following.

\begin{proposition}\label{pro:standard-hardness-gadget}
Let $\cH$ be a $p$-rigid relational structure. If there exists a $p$-protected standard hardness gadget $(\rel, A_{1,1}, A_{1,2}, A_{2,1}, A_{2,2})$ for $\cH$, then $\NpCSP(\cH)$ is $\#_p$-complete.
\end{proposition}

In order to prove Proposition~\ref{pro:standard-hardness-gadget}, we use the result on the hardness of counting graph homomorphisms modulo prime numbers, see Theorem~\ref{the:CountingGraphHom}. We use it in the case of bipartite graphs, in which case homomorphisms also preserve the left and right side of the partitions of the graphs. 

We define the bipartite graph $\sK_\rel = (L \cup R, E)$ using the hardness gadget $(\rel, A_{1,1}, A_{1,2}, A_{2,1}, A_{2,2})$ as follows:
\begin{itemize}
\item $L = \{ u^1_{\vv x} \mid \vv x \in \pr_{[s]} \rel \}$,
\item $R = \{ u^2_{\vv y} \mid \vv y \in \pr_{[n]-[s]} \rel \}$,
\item $E = \{ (u^1_{\vv x}, u^2_{\vv y}) \mid (\vv x, \vv y)\in \rel \}$.
\end{itemize}
In order to streamline the notation, we define a mapping $\mathfrak{F}$ from $\pr_{[s]} \rel \cup \pr_{[n]-[s]} \rel$ to $L \cup R$ by $\mathfrak{F}(\vv x) = u^1_{\vv x}$ or $\mathfrak{F}(\vv x) = u^2_{\vv x}$ depending on whether $\vv x\in L$ or $\vv x\in R$. It is straightforward to observe that $\mathfrak{F}$ is bijective.

Based on the definition of $\sK_\rel$, we can define the sets $A'_{i,j} = \{u^i_{\vv x} \mid \vv x \in A_{i,j}\}$ for $i,j\in[2]$. Note that, by the definition of a standard hardness gadget, $A'_{1,1} \cap A'_{1,2} = A'_{2,1} \cap A'_{2,2} = \emptyset$.

Before proceeding further, we introduce some notation regarding graphs. Let the neighbor set of a vertex in a graph $\sG=(V,E)$ be defined as $N_\sG(v) = \{ u \in V \mid (u,v) \in E \}$. We say that graphs $(\sG, a)$ and $(\sG,b)$ (with a distinguished vertex) are isomorphic if there is an automorphism of $\sG$ mapping $a$ to $b$. The following lemma is straightforward.

\begin{lemma}\label{lem:auto-by-neighbour}
Let $\sG=(V,E)$ be a graph and $a,b \in V$. The graphs $(\sG, a)$ and $(\sG,b)$ are isomorphic if and only if $N_\sG(a) = N_\sG(b)$.
\end{lemma}

\begin{lemma}\label{lem:k-graph}
Let $u^i_{\mathbf{x}}, u^i_{\mathbf{x'}} \in A'_{i,j}$ and $u^i_{\mathbf{z}} \in A'_{i,j'}$ where  $j \neq j'$. Then, $(\sK_\rel , u^i_{\mathbf{x}}) \cong (\sK_\rel, u^i_{\mathbf{x'}})$ and $(\sK_\rel , u^i_{\mathbf{x}}) \not \cong (\sK_\rel, u^i_{\mathbf{z}})$.
\end{lemma}

\begin{proof}
We prove the lemma only for the case of the set $i=j=1$. The proof for the remaining sets is similar. Let $u^1_{\mathbf{x}}, u^1_{\mathbf{x'}} \in A'_{1,1}$, then $\mathbf{x}, \mathbf{x'} \in A_{1,1}$. Let $(u^1_{\mathbf{x}}, u^2_{\mathbf{y}}) \in E$. Thus, by the definition of $\sK_\rel$, $(\mathbf{x}, \mathbf{y}) \in \rel$. By Lemma~\ref{lem:struct-standard-gadget}, $(\mathbf{x'} , \mathbf{y}) \in \rel$. So, $(u^1_{\mathbf{x'}}, u^2_{\mathbf{y}}) \in E$, therefore $N_\sK(u^1_{\mathbf{x}}) = N_\sK(u^1_{\mathbf{x'}})$. Hence, by Lemma~\ref{lem:auto-by-neighbour}, $(\sK_\rel, u^1_{\mathbf{x}}) \cong (\sK_\rel, u^1_{\mathbf{x'}})$.

By the definition of $\sK_\rel$, if $u^1_{\mathbf{x}}\in A'_{1,1}$ and $u^1_{\mathbf{z}} \in A'_{1,2}$ then there exists $u^2_{\mathbf{y}}$ such that $(u^1_{\mathbf{z}}, u^2_{\mathbf{y}}) \in E$, but $(u^1_{\mathbf{x}}, u^2_{\mathbf{y}}) \not \in E$. Hence, by Lemma~\ref{lem:auto-by-neighbour}, $(\sK_\rel , u^1_{\mathbf{x}}) \not \cong (\sK_\rel, u^1_{\mathbf{z}})$.
\end{proof}

We need one more ingredient before proving Proposition~\ref{pro:standard-hardness-gadget}. We need to show that the $p$-reduced form of $\sK_\rel$ is not trivial. Let $\pi_{i,j}$ be an automorphism of $\sK_\rel$ such that $\pi_{i,j}(x) = x$ for all $x \not \in A'_{i,j}$ and it is a permutation on $A'_{i,j}$. Since $A'_{i,j}$ is $p$-protected, any $p$-automorphism on $A'_{i,j}$ has a fixed point. Let $\widetilde{A'_{i,j}}$ denote the set $A'_{i,j}$ after removing all automorphisms of order $p$ from $\sK_\rel$. Clearly $|\widetilde{A'_{i,j}}| \equiv |A'_{i,j}| \not \equiv 0 \pmod p$. Thus, the reduced form $\sK_\rel$ is not trivial, and in fact, based on the properties of $\sK_\rel$, it is not a complete bipartite graph.

Thus, the following lemma can be concluded by Theorem~\ref{the:CountingGraphHom}.

\begin{lemma}\label{lem:hardness-for-K}
    The problem $\#_p\Hom(\sK_\rel)$ is $\#_pP$-hard.
\end{lemma}

Now, we are ready to prove the key lemma of this section.

\begin{lemma}\label{lem:standard-gadget-to-csp}
    Let $\cH$ have a $p$-protected standard hardness gadget $(\rel, A_{1,1}, A_{1,2}, A_{2,1}, A_{2,2})$. Then $\#_p\Hom(\sK_\rel) \leq \NpCSP(\cH)$.
\end{lemma}
\begin{proof}
    Let $\sG=(V,E)$ be an instance of the problem $\#_p\Hom(\sK_\rel)$.    
    We define the following instance $\cP = (W, \cC)$ of $\CSP(\cH)$. Note that in this instance, we only use the $p$-mpp-definable relation $\rel$ from the standard hardness gadget.
    \begin{itemize}
        \item For each $u \in L$, we define $\vv x _u = (\vv x_{u,1}, ... , \vv x_{u,s})$.
        \item For each $v \in R$, we define $\vv y _v = (\vv y_{v,1}, ... , \vv y_{v,t})$.
        \item $W = \bigcup_{u \in L} \{ \vv x_{u,1} , ..., \vv x_{u,s} \} \cup \bigcup_{v \in R} \{ \vv y_{v,1} , ..., \vv y_{v,t} \}$.
        \item For each $(u,v) \in E$, introduce $C = \langle (\vv x_{u,1}, ... , \vv x_{u,s}, \vv y_{v,1}, ... , \vv y_{v,t}), \rel \rangle$ into $\cC$.
    \end{itemize}
    We show that each solution $\psi$ for $\cP$ gives rise to a homomorphism from $\sG$ to $\sH$. Let $\psi$ be a solution for $\cP$, then $((\psi(\vv x_{u,1}), ... , \psi(\vv x_{u,s}), \psi(\vv y_{v,1}), ... , \psi(\vv y_{v,t}))) \in \rel$. For all $u \in L$, set $\vf(u) = (\psi(\vv x_{u,1}), ... , \psi(\vv x_{u,s}))$, and for all $v \in R$, set $\vf(v) = (\psi(\vv y_{v,1}), ... , \psi(\vv y_{v,t}))$. Clearly, if $(u,v) \in E_G$, then $(\vf(u), \vf(v)) \in E_{\sK}$ since $\psi$ is a solution for $\cP$.

    Conversely, let $\vf$ be a homomorphism from $\sG$ to $\sK$. Define the following solution for $\cP$: $\psi(\vv x_{u,i}) = \mathfrak{F}^{-1}(\vf(u))[i]$. Since $\vf \in \Hom(\sG, \sK_\rel)$, if $(u,v) \in E_\sG$, then $(\vf(u), \vf(v)) \in E_{\sK_\rel}$. Note that $\vf(u), \vf(v)$ are nodes in $\sK_\rel$. So, we can apply the reverse mapping of $\mathfrak{F}$. Hence, we will have $(\mathfrak{F}^{-1} (\vf(u)), \mathfrak{F}^{-1} (\vf(v))) \in \rel$. Thus, $\psi$ is also a solution for $\cP$.

    Therefore, the number of solutions of $\sG$ in $\hom(\sG, \sK_\rel)$ is the same as the number of solutions of $\cP$. The result follows.
\end{proof}

\begin{proof}[Proof of Proposition~\ref{pro:standard-hardness-gadget}]
    By Lemma~\ref{lem:standard-gadget-to-csp}, we have the reduction $\#_p\Hom(\sK_\rel) \leq \NpCSP(\cH)$. Additionally, from Lemma~\ref{lem:hardness-for-K}, we know that $\#_p\Hom(\sK_\rel)$ is $\#_pP$-hard. Therefore, if $\cH$ has a rectangularity obstruction, it is $\#_pP$-complete.
\end{proof}

Standard hardness gadgets provide a fairly limited condition for the hardness of $\#_p\CSP(\cH)$. In fact, it is possible to prove that $\NpCSP(\cH)$ is $\#_pP$-complete whenever $\cH$ has any $p$-protected rectangularity obstruction, not necessarily a standard gadget. However, it cannot be done using Theorem~\ref{the:CountingGraphHom} as a black box, and is outside of the scope of this paper.

%%%%%%%%%%%%%%%%%%%%%%%%%%%%%%%%%%%%%%%
%%%%%%%%%%%%%%%%%%%%%%%%%%%%%%%%%%%%%%%
\section{Binarization}\label{sec:binarization}

In this section we will introduce a transformation of a multi-sorted relational structure $\ovarrow\cH$ to a one consisting of binary rectangular relations. This transformation, although hugely increasing the sizes of domains, preserves many of the useful properties of $\ovarrow\cH$ including certain types of polymorphisms and automorphisms. Also, the counting CSP over the new language is parsimoniously interreducible with that over $\ovarrow\cH$. The hope therefore is that such transformation reduces a complexity classification of $\#_p\CSP(\cH)$ to those with only binary rectangular relations, which may be more accessible. This transformation is not new, it appeared in a different context in \cite{Bulatov12:weighted} and \cite{Barto14:local}. By $\ar(\rel)$ we denote the arity of relation $\rel$.

Let $\ovarrow\cH=\{\{H_i\}_{i\in[k]};\vc\relo n)$ be a finite multi-sorted relational structure. We assume that every $H_i$ is among $\vc\relo n$ as a unary relation. The structure $b(\ovarrow\cH)$ is constructed as follows:
\begin{itemize}
\item The domains are $\vc\relo n$, the relations from $\ovarrow\cH$.
\item For $\relo_i,\relo_j$, $i\le j$, and $s\in[\ar(\relo_i)], t\in[\ar(\relo_j)]$, the structure $b(\ovarrow\cH)$ contains the binary relation $\rel^{ij}_{st}\sse\relo_i\tm\relo_j$ given by
$
\rel^{ij}_{st}=\{(\ba,\bb)\mid \ba\in\relo_i,\bb\in\relo_j,\ba[s]=\bb[t]\}.
$
\end{itemize}

First we show that the $\CSP$ over $b(\ovarrow\cH)$ is interreducible with that over $\cH$.

\begin{proposition}[see also Theorem~4 of \cite{Bulatov12:weighted}]\label{pro:binarization-main-body}
(1) For any structure $\ovarrow\cH$, $\NCSP(\ovarrow\cH)$ is parsimoniously interreducible with $\NCSP(b(\ovarrow\cH))$.\\[1mm]
(2) For any structure $\ovarrow\cH$, $\#_p\CSP(\ovarrow\cH)$ is parsimoniously interreducible with  $\#_p\CSP(b(\ovarrow\cH))$.
\end{proposition}

\begin{proof}
We use the standard definition of the CSP. Let $\cP$ be any instance of $\NCSP(\cH)$, with variable set $V$ and constraint set $\cC$. We assume that for every $v\in V$ with domain $A\in H$, $\cP$ contains a constraint $\ang{(v),A}$. We construct an instance $\cP'$ of $\NCSP(b(\cH))$, which has variable set $V'$ and constraint set $\cC'$, as follows.\\[3mm]
(i) For each $C\in\cC$, we introduce a variable $v_C\in V'$. Thus $V'=\cC$.\\[2mm]
(ii) For all $C_1,C_2\in\cC$, $C_1=\ang{(v^1_1\zd v^1_k), \rel_i}, C_2=\ang{(v^2_1\zd v^2_\ell), \rel_j}$, such that $v^1_s=v^2_t$ we introduce a constraint $\ang{(v_{C_1},v_{C_2}),\rel^{ij}_{st}}\in\cC'$.\\

Let $\vf:V\to\bigcup_{i\in[k]}H_i$ be any solution of $\cP$. Then define $\vf':V'\to\bigcup_{\rel\in\cH}\rel$ as follows. For every $v_C\in V'$ with $C=\ang{(\vc s\ell),\rel}$, set $\vf'(v_C)=(\vf(s_1)\zd\vf(s_\ell))$. As is easily seen, $\vf'$ is a solution of $\cP'$. Conversely, if $\vf'$ is a solution of $\cP'$, then, as every $A\in H$ is a relation from $\cH$, the action of $\vf'$ on the domains from $H$ is well defined, let us denote it $\vf$. We would like to argue now that the action of $\vf'$ on relations $\rel_i$ coincides with the component-wise action of $\vf$. In other words that for any $\ba\in\rel_i$, where $\ba=(a_1\zd a_\ell)$, it holds that $\vf'(\ba)=(\vf(a_1)\zd\vf(a_\ell))$. Let the $m$th coordinate of $\rel_i$ be $H_j$, which is also the $j$th domain for $b(\cH)$. Then $\vf'$ preserves the relation $\rel^{ij}_{m1}$, which means that for any $(\ba,b)\in\rel^{ij}_{m1}$ we have $a_m=b$ and 
\[
\vf'\cp{\ba}b=\cp{\vf'(\ba)}{\vf(b)}\in\rel^{ij}_{m1}.
\]
Then if $\vf'(\ba)=(\vc c\ell)$, then $c_m=\vf(b)$, proving the result.

Conversely, suppose $\cP'$ is any instance of $\NCSP(b(\cH))$ with variable set $V'$ and constraint set $\cC'$. We construct an instance $\cP$ of $\NCSP(\cH)$, with variable set $V$ and constraint set $\cC$, as follows. Recall that every variable $v\in V'$ has an associated domain $\rel^v\in\cH$. Let $r_v$ be the arity of $\rel^v$. We now create a relation $\sim$ on the set $V^*=\bigcup_{v\in V'}(\{v\}\tm[r_v])$, as follows.  For $u,v\in V'$, let $(u,s)\sim'(v,t)$ if there is a constraint $\ang{(u,v),\rel^{ij}_{st}}\in\cC'$, where the domain of $u$ is $\rel_i$, the domain of $v$ is $\rel_j$, and $s\in[r_u], t\in[r_v]$. Then $\sim$ is the symmetric-transitive closure of $\sim'$. Clearly, $\sim$ is an equivalence relation, which “identifies” the variables $(u,s)$ and $(v,t)$. Now, suppose $\vf'$ is any solution of $\cP'$. Then, for any $v\in V'$, we have $\vf'(v)=(\vc a{r_v})\in\rel^v$. Let us write $\vf'_s(v)=a_s$ ($s\in[r_v]$). Now, define $\vf: V^*\to\bigcup_{i\in[k]} H_i$ from $\vf'$ by $\vf(v,s)=\vf'_s(v)$ for all $(v,s)\in V^*$, $s\in[r_v]$. We will write $\vf=\zeta(\vf')$ for this function. If, for any $u,v\in V'$, we have $(u,s)\sim(v,t)$, then we must have $(\vf'(u),\vf'(v))\in\rel^{ij}_{st}$. This implies that $\vf'_s(u)=\vf'_t(v)$, and hence $\vf(u,s)=\vf(v,t)$, where $\vf=\zeta(\vf')$. 
Let the variable set $V$ of $\cP$ will be the set of equivalence classes $V^*/\sim$. For any $(v,s)\in V^*$, we write  $\ov{(v,s)}$ for its equivalence class. Let $\vf:V^*\to\bigcup_{i\in[k]} H_i$ be such that $\vf=\zeta(\vf')$ for some $\vf':V\to\{\vc Hk\}$. Then we can define $\ov\vf:V\to\bigcup_{i\in[k]} H_i$ by $\ov\vf(\ov{(v,s)})=\vf(v,s)$ for all $(v,s)\in\ov{(v,s)}$. Thus we have constructed a bijection between the mappings $\ov\vf:V\to\{\vc Hk\}$ and mappings $\vf':V'\to\bigcup H$ that may potentially be solutions of $\cP'$. We will write $\ov\vf=\xi(\vf')$ for this bijection. Now, $\cP$ will have constraint set
\[
\cC=\{\ang{\bv,\rel}\mid \bv=(\ov{(v,1)}\zd\ov{(v,r_v)}), v\in V', \text{ and $\rel\sse H_{i_1}\tms H_{i_{r_v}}$ is the domain of $v$}\}.
\]
Then, $\ov\vf=\xi(\vf')$ is a solution of $\cP$ if and only if $\vf'$ is a solution of $\cP'$, and we have a parsimonious reduction from $\NCSP(b(\cH)$ to $\NCSP(\cH)$. 
\end{proof}

The structure $b(\ovarrow\cH)$ inherits other properties of $\ovarrow\cH$.

\begin{proposition}\label{pro:inherit-main}
Let $\ovarrow\cH=\{\{H_i\}_{i\in[k]};\vc\relo n)$ be a multi-sorted structure. Then \\[1mm]
(1) $\ovarrow\cH$ has a Mal'tsev polymorphism if and only if $b(\ovarrow\cH)$ does;\\[1mm]
(2) $\ovarrow\cH$ is $p$-rigid if and only if $b(\ovarrow\cH)$ is;\\[1mm]
(3) If $\ovarrow\cH$ is strongly $p$-rectangular then so is $b(\ovarrow\cH)$.
\end{proposition}

The first two items of Proposition~\ref{pro:inherit-main} allow for a fairly straightforward proof using the algebraic approach. However, this would require the introduction of more advanced algebraic machinery. Thus, we give a more elementary proof here.

\begin{proof}
(1) Suppose first that $\cH$ has a Mal'tsev polymorphism $f$. Then for any $\rel\in\cH$ and any $\ba_1,\ba_2,\ba_3\in\rel$, the tuple $f(\ba_1,\ba_2,\ba_3)\in \rel$, and this action defines an operation $f^b$ on the domains of $b(\cH)$ (i.e.\ relations from $\cH$). Moreover, $f^b$ is Mal'tsev, because in the notation above $f(\ba_1,\ba_1,\ba_2)=f(\ba_2,\ba_1,\ba_1)=\ba_2$. It remains to make sure that $f^b$ preserves $\rel^{ij}_{st}$. Let $(\ba_1,\bb_1),(\ba_2,\bb_2),(\ba_3,\bb_3)\in\rel^{ij}_{st}$, where $\ba_i=(a_{i1}\zd a_{i\ell}),\bb_i=(b_{i1}\zd b_{im})$. Then we have $a_{1s}=b_{1t},a_{2s}=b_{2t},a_{3s}=b_{3t}$, and therefore, for 
\[
\cp\bc\bd=f^b\left(\cp{\ba_1}{\bb_1},\cp{\ba_2}{\bb_2},\cp{\ba_3}{\bb_3}\right)
\]
we obtain $\bc[s]=\bd[t]$, and so $(\bc,\bd)\in\rel^{ij}_{st}$.

Suppose now that there is a Mal'tsev polymorphism $f$ of $b(\cH)$. Since $H_i\in\{\vc\relo m\}$ for each $i\in[k]$, the action of $f$ on the domains from $H$ is well defined, let us denote it $f'$. We would like to argue now that the action of $f$ on relations $\rel_i$ coincides with the component-wise action of $f'$. In other words that for any $\ba_1,\ba_2,\ba_3\in\rel_i$, where $\ba_j=(a_{j1}\zd a_{jk})$
\[
f(\ba_1,\ba_2,\ba_3)=\cpp{f'(a_{11},a_{21},a_{31})}{\vdots}{f'(a_{1k},a_{2k},a_{3k})}.
\]
Let the $\ell$th coordinate of $\rel_i$ be $H_j$, which is also the $j$th domain for $b(\cH)$. Then $f$ preserves the relation $\rel^{ij}_{\ell1}$, which means that for any $(\ba_1,b_1),(\ba_2,b_2),(\ba_3,b_3)\in\rel^{ij}_{\ell1}$ we have $a_{1\ell}=b_1,a_{2\ell}=b_2, a_{3\ell}=b_3$ and 
\[
f\left(\cp{\ba_1}{b_1},\cp{\ba_2}{b_2},\cp{\ba_3}{b_3}\right)=\cp{f(\ba_1,\ba_2,\ba_3)}{f'(b_1,b_2,b_3)}.
\]
Denoting $\bc=f(\ba_1,\ba_2,\ba_3)$, we obtain $\bc[\ell]=f'(b_1,b_2,b_3)$, implying that $f$ acts as $f'$ coordinate-wise.

(2) Let $\pi$ be a $p$-automorphism of $\cH$. Then $\pi$ is a unary polymorphism of $\cH$, and therefore as in item (1) $\pi^b$ is a unary polymorphism of $b(\cH)$. That it is a $p$-automorphism is straightforward. The reverse claim is as in (1), as well.

(3) Suppose that $\ang{b(\cH)}_p$ is not p-rectangular. This means that there is $\relo'(\bx)=\exists^2\by\Phi'(\bx,\by)$, where $\Phi'$ is a conjunctive formula, such that $\relo'$ is not rectangular. We consider $\Phi'$ as an instance of $\NCSP(b(\cH))$ and apply the transformation from the proof of Proposition~\ref{pro:binarization-main-body} (from a $\NCSP(b(\cH))$ instance to a $\NCSP(\cH)$ instance). Let $V'=V'_1\cup V'_2$, where $V'_1=\{\vc x\ell\}, V'_2=\{\vc ym\}$, and $\bx=(\vc x\ell),\by=(\vc ym)$. The set $\cC'$ of constraints is the set of all atoms in $\Phi'$. Let the set of variables $V$ and the set $\cC$ of constraints be constructed as in the proof of Proposition~\ref{pro:binarization-main-body}. Let also $V_1=\{\ov{(x,s)}\mid x\in V'_1\}$, $V_2=V-V_1$, denote $V_1=\{\vc zq\}, V_2=\{\vc tr\}$. Set
\[
\Phi(\vc zq,\vc yr)=\bigwedge_{\ang{\bs,\rel}\in\cC}\rel(\bs),
\]
and 
\[
\relo(\vc zq)=\exists^{\equiv2}\vc tr\Phi(\vc zq,\vc tr).
\]
We prove that $\relo$ is not rectangular.

Without loss of generality we may assume that a witness of non-$p$-rectangularity of $\relo'$ is as follows: $J\sse[m], K=[m]-J$, $\vf_1,\vf_2\in\pr_J\relo',\psi_1,\psi_2\in\pr_K\relo'$ are such that $(\vf_1,\psi_1),(\vf_1,\psi_2),(\vf_2,\psi_1)\in\relo'$, but $(\vf_2,\psi_2)\not\in\relo'$. Let $W_1\sse V_1$ be given by $W_1=\{\ov{(v,s)}\mid v\in J\}$ and $W_2=V_1-W_1$. We first observe that $W_2\ne\eps$. Indeed, by construction if $(v,s)\sim(u,t)$ then for any solution $\vf:V'\to\bigcup_{i\in[n]}\relo_i$, where $\vf(v)=(\vc ak),\vf(u)=(\vc b\ell)$ we have $a_s=b_t$. That $W_2=\eps$ means that for any $x_i\in K$ and any $s\in[r_{x_i}]$ it holds that $(x_i,s)\sim(x_j,t)$ for some $x_j\in J$ and $t\in[r_{x_j}]$. Therefore for any $\vf\in\relo$, the assignment $\pr_K\vf$ is determined by the assignment $\pr_J\vf$, a contradiction with the choice of a counterexample of $p$-rectangularity.

That $(\vf_1,\psi_1),(\vf_1,\psi_2),(\vf_2,\psi_1)\in\relo$, but $(\vf_2,\psi_2)\not\in\relo$ means that the first three assignments have the number of extensions to an assignment $V'\to\bigcup_{i\in[n]}\relo_i$ not divisible by $p$, while $(\vf_2,\psi_2)$ has a number of such extensions that is a multiple of $p$ (including zero). Observe now that we can apply the bijection $\zeta$ to assignments $\vf_i,\psi_i$ as follows: take any $\sg'$ that extends one of them to an assignment $V'\to\bigcup_{i\in[n]}\relo_i$, take $\sg=\zeta(\sg')$ and set $\sg^*=\pr_{W_1}\sg$ or $\sg^*=\pr_{W_2}\sg$ depending on whether we deal with $\vf_1,\vf_2$ or $\psi_1,\psi_2$. Then since $\zeta$ is a bijection, the number of extensions of $(\zeta(\vf_i),\zeta(\psi_j))$ to an assignment of $V$ equals that of $(\vf_i,\psi_j)$ to an assignment of $V'$. Thus, $(\zeta(\vf_1),\zeta(\psi_1)),(\zeta(\vf_1),\zeta(\psi_2)),(\zeta(\vf_2),\zeta(\psi_1))$ have the number of extensions not divisible by $p$ and so belong to $\relo$, while $(\zeta(\vf_2),\zeta(\psi_2))$ has the number of extensions divisible by $p$, and so does not belong to $\relo$, witnessing that $\relo$ is not rectangular.
\end{proof}

%%%%%%%%%%%%%%%%%%%%%%%%%%%%%%%%%%%%%%%
%%%%%%%%%%%%%%%%%%%%%%%%%%%%%%%%%%%%%%%
\bibliographystyle{plainurl}
%% \bibliography{refrences.bib}

\newpage
\appendix

%\input{Data/a-intro}
%%%%%%%%%%%%%%%%%%%%%%%%%%%%%%%
%%%%%%%%%%%%%%%%%%%%%%%%%%%%%%%
\section{Products of relational structures and the uniqueness of $\protect\ovarrow\cH^{*p}$}\label{sec:appendix-new}
%% \input{Data/q-appendix-new}

%%%%%%%%%%%%%%%%%%%%%%%%%%%%%%%%%%%%%%%%%%%%%%%%%
The goal of this subsection is to prove Proposition~\ref{pro:mult-uniqueness-main}.

\renewcommand{\theprop}{2.1}
\begin{prop}
Let $\ovarrow\cH$ be a multi-sorted  structure and $p$ a prime. Then up to an isomorphism there exists a unique $p$-rigid multi-sorted  structure $\ovarrow\cH^{*p}$ such that $\ovarrow\cH\rightarrow_p^* \ovarrow\cH^{*p}$.
\end{prop}

However, in order to do this we first need to introduce some auxiliary tools. In particular we use expansions of multi-sorted structures vertices and their homomorphisms. 
Note that a multi-sorted relational structure $(\ovarrow\cG,\vv a)$ with distinguished vertices can be viewed as an expansion of $\cG$ with $k$ additional unary symbols, one for each distinguished vertex. Also, note that these distinguished vertices may come from different sorts. In such an interpretation a homomorphism of multi-sorted structures with distinguished vertices is just a homomorphism between the corresponding expansions.

Next we introduce two types of products of multi-sorted structures. Recall that the \emph{direct product} of similar multi-sorted $\sg$-structures $\ovarrow\cH,\ovarrow\cG$, denoted $\ovarrow\cH\tm\ovarrow\cG$ is the multi-sorted $\sg$-structure with the base set $\colect{H_i \times G_i}{i\in[k]}$ and such that the interpretation of $\rel\in\sg$ is given by $\rel^{\ovarrow\cH\tm\ovarrow\cG}((a_1,b_1)\zd(a_k,b_k))=1$ if and only if $\rel^{\ovarrow\cH}(\vc ak)=1$ and $\rel^{\ovarrow\cG}(\vc bk)=1$. By $\ovarrow\cH^\ell$ we will denote the \emph{$\ell$th power} of $\ovarrow\cH$, that is, the direct product of $\ell$ copies of $\ovarrow\cH$. The direct product $(\ovarrow\cG,\vv x)\tm(\ovarrow\cH,\vv y)$ of structures $(\ovarrow\cG,\vv x),(\ovarrow\cH,\vv y)$, where $\vv x, \vv y$ are $r$-tuples, is defined to be $(\ovarrow\cG\tm\ovarrow\cH,(\bx[1]_1,\by[1])\zd(\bx[r],\by[r]))$. 

 For $a\in\bigcup_{i\in[k]}H_i$ (or $a\in\bigcup_{i\in[k]}G_i$) we write $\type(a)=i$ if $a\in H_i$ (respectively, $a\in G_i$) Two $r$-tuples $\vv{x} $ and $\vv{y}$ have the same \emph{equality type} if $(x_i, \type(x_i)) = (x_j, \type(x_j))$ if and only if $(y_i, \type(y_i) )= (y_j, \type(y_j))$ for $i,j\in[r]$. 
 Let $(\ovarrow\cG,\vv x)$ and $(\ovarrow\cH,\vv y)$ be multi-sorted structures with $r$ distinguished vertices and such that $\vv x$ and $\vv y$ have the same equality type. Then $(\ovarrow\cG,\vv x)\odot(\ovarrow\cH,\vv y)$ denotes the structure that is obtained by taking the disjoint union of $\ovarrow\cG$ and $\ovarrow\cH$ and identifying every $x_i$ with $y_i$, $i\in[r]$. The distinguished vertices of the new structure are $\vc xr$.

The following statement is straightforward.

\begin{proposition}
Let $(\ovarrow\cG,\vv x),(\ovarrow\cH,\vv y),(\ovarrow\cK,\vv z)$ be similar multi-sorted relational structures with $r$ distinguished vertices. Then
\begin{align*}
\hom((\ovarrow\cG,\vv x)\odot(\ovarrow\cH,\vv y),(\ovarrow\cK,\vv z)) &= \hom((\ovarrow\cG,\vv x),(\ovarrow\cK,\vv z))\cdot\hom((\ovarrow\cH,\vv y),(\ovarrow\cK,\vv z)); \\  
\hom((\ovarrow\cK,\vv z),((\ovarrow\cG,\vv x)\times(\ovarrow\cH,\vv y)) &= \hom((\ovarrow\cK,\vv z),(\ovarrow\cG,\vv x))\cdot\hom((\ovarrow\cK,\vv z),(\ovarrow\cH,\vv y)).   
\end{align*}
Moreover, $\ovarrow\vf\in\Hom((\ovarrow\cK,\vv z),((\ovarrow\cG,\vv x)\times(\ovarrow\cH,\vv y))$ if and only if there exist $\ovarrow\vf_1 \in \Hom((\ovarrow\cK,\vv z),(\ovarrow\cG,\vv x))$ and $\ovarrow\vf_2\in\Hom((\ovarrow\cK,\vv z),(\ovarrow\cH,\vv y))$ such that $\ovarrow\vf(v)=(\ovarrow\vf_1(v),\ovarrow\vf_2(v))$ for $v\in K$.
\end{proposition}

We will also need another simple observation. By $\inj((\ovarrow\cG,\vv{x}),(\ovarrow\cH,\vv{y}))$ we denote the number of injective homomorphisms from $(\ovarrow\cG,\vv x)$ to $(\ovarrow\cH,\vv y)$.

\begin{lemma}[cf. \cite{ref:CountingMod2ToSquarefree}]
Let $(\ovarrow\cG,\vv{x})$ and $(\ovarrow\cH,\vv{y})$ be similar multi-sorted relational structures with $r$ distinguished vertices. If $\vv{x},\vv{y}$ do not have the same equality type, then $\inj((\ovarrow\cG,\vv{x}),(\ovarrow\cH,\vv{y}))=0$.
\end{lemma}
\begin{proof}
    If there are $i, j \in [r]$, where $r$ is the arity of $\vv x$ and $\vv y$, such that $(x_i , \type(x_i))= (x_j, \type(x_j))$ but $(y_i, \type(y_i) )\neq (y_j, \type(y_j))$, then there are no homomorphisms (injective or otherwise) from $(\ovarrow{\cG}, \vv{x})$ to $(\ovarrow{\cH}, \vv{y})$, since $x_i$ cannot be mapped simultaneously to both $y_i$ and $y_j$. Otherwise, there must be $i, j \in [r]$ such that $(x_i , \type(x_i)) \neq (x_j, \type(x_j))$ but $(y_i, \type(y_i) ) = (y_j, \type(y_j))$. Then no homomorphism $\vf$ can be injective because we must have $\vf(x_i) = \vf(x_j) = y_i$.
\end{proof}

Finally, we will use factor structures. Let $\ovarrow\cH$ be a multi-sorted $\sg$-structure and $\ovarrow\th = \colect{\th_i}{i\in[k]}$ a collection of equivalence relations on $H = \colect {H_i}{i\in[k]}$, such that $\th_i\sse H_i^2$ for all $i\in [k]$. By $\ovarrow\cH/_{\ovarrow\th}$ we denote the \emph{factor structure} defined as follows. 
\begin{itemize}
    \item
    $\ovarrow\cH/_{\ovarrow\th}$ is a multi-sorted $\sg$-structure.
    \item 
    Let $H_i/_{\th_i}=\{a/_{\th_i}\mid a\in H_i\}$, where $a/_{\th_i}$ denotes the $\th_i$-class containing $a$. The base set of $\ovarrow\cH/_{\ovarrow\th}$ is $H/_{\ovarrow\th}=\colect{H/_{\th_i}}{i\in[k]}$
    \item
    For any $\rel\in\sg$, say, $s$-ary, $\rel^{\ovarrow\cH/_{\ovarrow\th}}=\{(a_1/_{\th_1}\zd a_s/_{\th_s})\mid (\vc as)\in\rel^{\ovarrow\cH}\}$.
\end{itemize}
The set of all equivalence relations of a multi-sorted set $H$ is denoted $\Part(H)$ and the equivalence relation that is the equality on each set from $H$ is denoted by $\zz$.

Factor structures often appear in relation with homomorphisms. If $\ovarrow\vf$ is a homomorphism from a multi-sorted structure $\ovarrow\cG$ to a multi-sorted structure $\ovarrow\cH$, then the \emph{kernel} $\ovarrow\th$ of $\ovarrow\vf$, denoted $\ker(\ovarrow\vf)$, is a collection of equivalence relations on $G = \colect{G_i}{i\in[k]}$ given by 
\[
\text{For } a,b \in G_i, \; (a,b) \in\th_i \text{ if and only if } \vf_i(a) =\vf_i(b).
\]
For an equivalence relation $\ovarrow\th$ on $G = \colect{G_i}{i\in[k]}$ by $\hom_{\ovarrow\th}(\ovarrow\cG,\ovarrow\cH)$ and $\Hom_{\ovarrow\th}((\ovarrow\cG,\vv x),(\ovarrow\cH,\vv y))$ we denote the number of homomorphisms from $\ovarrow\cG$ to $\ovarrow\cH$ (from $(\ovarrow\cG,\vv x)$ to $(\ovarrow\cH,\vv y)$) with kernel $\ovarrow\th$.
The homomorphism $\ovarrow\vf$ gives rise to a homomorphism $\ovarrow\vf/_{\ovarrow\th}$ from $\ovarrow\cG/_{\ovarrow\th}$ to $\ovarrow\cH$, where $\vf_{i}/_{\th_i}(a/_{\th_i})=\vf_i(a)$, $a\in G_i$. The homomorphism $\ovarrow\vf/_{\ovarrow\th}$ is always injective since all $\vf_i/_{\th_i}$ are injective.

We also define factor structures for structures with distinguished vertices as follows. Let $(\ovarrow\cH,\vv a)$ be a structure with $s$ distinguished vertices and $\ovarrow\th$ an equivalence relation on $H$. Then $(\ovarrow\cH,\vv a)/_{\ovarrow\th}=(\ovarrow\cH/_{\ovarrow\th}, \vv a/_\th)$, where $\vv a/_{\ovarrow\th} = (a_1/_{\th_{\type(a_1)}}\zd a_s/_{\th_{\type(a_s)}})$.

\begin{lemma}\label{lem:lovasz-multi-structure}
Let $(\ovarrow\cG, \vv{x})$ and $(\ovarrow\cH,\vv y)$ be similar $p$-rigid multi-sorted relational structures with $r$ distinguished vertices. Then, $(\ovarrow\cG, \vv{x}) \cong (\ovarrow\cH,\vv y)$ if and only if 
\begin{equation}
\hom((\ovarrow\cK, \vv{z}) , (\ovarrow\cG, \vv{x})) \equiv \hom((\ovarrow\cK, \vv{z}),(\ovarrow\cH, \vv{y})) \pmod{p} \label{equ:lovasz}
\end{equation} 
for all relational structures $(\ovarrow\cK, \vv z)$ with $r$ distinguished vertices.
\end{lemma}

\begin{proof}
The proof goes along the same lines as that in \cite{ref:CountingModPToTrees_gbel_et_al_LIPIcs}. If $(\ovarrow\cG, \vv{x})$ and $(\ovarrow\cH,\vv y)$ are isomorphic, then \eqref{equ:lovasz} obviously holds for all multi-sorted relational structures $(\ovarrow\cK, \vv{z})$. 

For the other direction, suppose that \eqref{equ:lovasz} is true for all $(\ovarrow\cK, \vv z)$.
First, we claim that this implies that $\vv{x} = (x_1 , \ldots, x_r)$ and $\vv{y} = (y_1, \ldots, y_r)$ have the same equality type. Indeed, if they do not, then without loss of generality there are $i,j\in[r]$ such that $(x_i, \type(x_i)) = (x_j, \type(x_j))$ but $(y_i, \type(y_i))\ne (y_j, \type(y_j))$. 
Let $s = \type(\ovarrow\cH)$, $ X= \{ x_1,\ldots, x_r \}$, and $\ovarrow\cK$ be the relational structure with the base multi-sorted set $K = \colect{K_i}{i\in[k]}$ where $K_i = \{ x \in X\mid \type(x) = i  \}$  with empty predicates, and $(\vc xr)$ as distinguished vertices. Note that there might exist $i$ such that $K_i$ is empty. Then $\hom((\ovarrow\cK, \vv{x}) , (\ovarrow\cG, \vv{x}))=1 \ne 0=  \hom((\ovarrow\cK, \vv{x}),(\ovarrow\cH, \vv{y}))$, a contradiction with the assumption that \eqref{equ:lovasz} holds for all $(\ovarrow\cK,\vv{x})$.

We show by induction on the number of vertices in $\ovarrow\cK$ that 
\begin{equation}\label{lavasEquInj}
\inj((\ovarrow\cK, \vv{z}),(\ovarrow\cG, \vv{x})) \equiv \inj((\ovarrow\cK, \vv{z}),(\ovarrow\cH, \vv{y})) \pmod{p},  
\end{equation}
for all $(\ovarrow\cK, \vv{z})$. Let $n_0 = |\{(x_1, \type(x_1)) , ..., (x_r, \type(x_r)) \}| = |\{(y_1, \type(y_1)) , ..., (y_r, \type(y_r))\}|$ be the number of distinct elements in $\vv x,\vv{y}$. For the base case of the induction, consider a relational structure $(\ovarrow\cK, \vv{z})$  such that $|K| \leq n_0$. If $\vv{z}$ does not have the same equality type as $\vv x,\vv{y}$, then $\inj((\cK, \vv{z}),(\cG, \vv{x})) = \inj((\cK, \vv{z}),(\cH, \vv{y})) = 0$.
If $\vv{x}$ has the same equality type as $\vv x,\vv{y}$, the only homomorphisms from $(\ovarrow\cK,\vv z)$ to $(\ovarrow\cG,\vv x),(\ovarrow\cH,\vv y)$ are the ones that map $z_i$ to $x_i,y_i$, respectively. Therefore, $\inj((\ovarrow\cK, \vv{z}),(\ovarrow\cG,\vv{x}))=\inj((\ovarrow\cK, \vv{z}) , (\cH, \vv{y}))$.

For the inductive step, let $n > n_0$ and assume that (\ref{lavasEquInj}) holds for all $(\ovarrow\cK, \vv{z})$ with $|K| < n$. Let $(\ovarrow\cK,\vv{z})$ be a relational structure with $|K|=n$, and let $\ovarrow\th = \colect{\th_i}{i\in[k]}$ be a collection of equivalence relations on $K$. Then, as is easily seen 
\begin{align*}
    \hom_{\ovarrow\th}((\ovarrow\cK,\vv z),(\ovarrow\cG,\vv x))&= \inj((\ovarrow\cK,\vv z)/_{\ovarrow\th},(\ovarrow\cG,\vv x)),\\
    \hom_{\ovarrow\th}((\ovarrow\cK,\vv z),(\ovarrow\cH,\vv y))&=\inj((\ovarrow\cK,\vv z)/_{\ovarrow\th},(\ovarrow\cH,\vv y)).
\end{align*}
Then
\begin{align*}
\hom((\ovarrow\cK,\vv{z}),(\ovarrow\cG,\vv{x})) &= \inj((\ovarrow\cK,\vv{z}),(\ovarrow\cG,\vv{x})) + \sum_{\ovarrow\th\in\Part(K)-\{\zz\}} \inj((\ovarrow\cK,\vv{z})/_{\ovarrow\th},(\ovarrow\cG,\vv{x})),\\
\hom((\ovarrow\cK,\vv{z}),(\ovarrow\cH,\vv{y})) &= \inj((\ovarrow\cK,\vv{z}),(\ovarrow\cH,\vv{y})) + \sum_{\ovarrow\th\in\Part(K)-\{\zz\}} \inj((\ovarrow\cK,\vv{z})/_{\ovarrow\th},(\ovarrow\cH,\vv{y})).
\end{align*}
Since 
\[
\hom((\ovarrow\cK,\vv{z}),(\ovarrow\cG,\vv{x}))\equiv\hom((\ovarrow\cK,\vv{z}),(\ovarrow\cH,\vv{y}))\pmod p,
\] 
and 
\[
\inj((\ovarrow\cK,\vv{z})/_\th,(\ovarrow\cG,\vv{x}))\equiv\inj((\ovarrow\cK,\vv{z})/_{\ovarrow\th},(\ovarrow\cH,\vv{y}))\pmod p,
\] 
for all $\ovarrow\th\in\Part(K)-\{\zz\}$, it implies 
\[
\inj((\ovarrow\cK,\vv{z}),(\ovarrow\cG,\vv{x}))\equiv\inj((\ovarrow\cK,\vv{z}),(\ovarrow\cH,\vv{y}))\pmod p.
\]
Finally, we prove that \eqref{lavasEquInj} for $(\ovarrow\cK, \vv{z})=(\ovarrow\cG,\vv x)$ implies $(\ovarrow\cG, \vv{x}) \cong (\ovarrow\cH, \vv y)$. An injective homomorphism from a relational structure to itself is an automorphism. Since $(\ovarrow\cG, \vv{x})$ is $p$-rigid, $|\Aut(\ovarrow\cG,\vv x)|=\inj((\ovarrow\cG,\vv x),(\ovarrow\cG,\vv x))\not\equiv0\pmod p$. Therefore,  $\inj((\ovarrow\cG,\vv{x}),(\ovarrow\cH,\vv y))\not\equiv0\pmod p$, meaning there is an injective homomorphism from $(\ovarrow\cG,\vv x)$ to $(\ovarrow\cH,\vv y)$. In a similar way, there is an injective homomorphism from $(\ovarrow\cH,\vv y)$ to $(\ovarrow\cG,\vv x)$. Thus, the two structures are isomorphic.
\end{proof}

\begin{proof}[Proof of Proposition~\ref{pro:mult-uniqueness-main}]
Suppose the contrary, that $\ovarrow\cH_1$ and $\ovarrow\cH_2$ are two different non-isomorphic $p$-reduced forms of $\ovarrow\cH$. By Proposition~\ref{pro:mult-aut-reduction} for any structure $\cG$
\[
\hom(\ovarrow\cG, \ovarrow\cH_1) \equiv \hom(\ovarrow\cG,\ovarrow\cH)\equiv \hom(\ovarrow\cG, \ovarrow\cH_2) \pmod{p}.
\]
By Lemma~\ref{lem:lovasz-multi-structure} we can conclude $\cH_1 \cong \cH_2$. The result follows.
\end{proof}

%%%%%%%%%%%%%%%%%%%%%%%%%%%%%%%%%%%%%%%%%%%%%%
%%%%%%%%%%%%%%%%%%%%%%%%%%%%%%%%%%%%%%%%%%%%%%
\section{Proof of Theorem~\ref{the:ConstantCSP-main}}\label{sec:appendix-const}

The goal of this subsection is to prove Theorem~\ref{the:ConstantCSP-main}. 

\renewcommand{\thether}{2.6}
\begin{ther}
Let $\cH$ be a multi-sorted relational structure and $p$ prime.\\[1mm]
(1) $\NpCSP(\ovarrow\cH^=)$ is Turing reducible to $\NpCSP(\ovarrow\cH)$;\\[1mm]
(2) Let $\ovarrow\cH$ be $p$-rigid. Then $\NpCSP(\ovarrow\cH^\const)$ is Turing reducible to $\NpCSP(\ovarrow\cH)$. 
\end{ther}

We start by introducing some auxiliary construction. For a relational structure $\cH$ the \emph{indicator problem} $\cI_n(\ovarrow\cH)$ \cite{Jeavons99:expressive} is defined as follows. Let $\ovarrow\cH=(\colect{H_i}{i\in[k]};\vc\rel m)$ be a (multi-sorted) relational structure and $n\in\nat$.
The $n$th indicator problem $\cI_n(\ovarrow\cH)$ for $\ovarrow\cH$ is defined to be  an instance of $\CSP(\ovarrow\cH)$ given by:
\begin{itemize}
    \item The set of variables $V$ is the set $V=\bigcup_{i\in [k]}H_i^n$;
    \item The domains are from $\ovarrow\cH$;
    \item For every (say, $\ell$-ary) relation $\rel_s\sse H_{i_1}\tm\dots\tm H_{i_\ell}$ of $\ovarrow\cH$ and every $\bc^j=(c_{j1}\zd c_{jn})\in H_{i_j}^n$, $j\in[\ell]$,  such that $\{\vc\bc n\}\sse\rel_s$, $\bc_i=(c_{1i}\zd c_{\ell i})$,
$\cI_n(\ovarrow\cH)$ contains the constraint $\ang{(\bc^1\zd\bc^k),\rel_s}$.
\end{itemize}
Let also $RG_n(\ovarrow\cH)$ denote the set of solutions of $\cI_n(\ovarrow\cH)$.

\begin{lemma}[\cite{Bodnarchuk69:Galua1,Geiger68:closed,Jeavons99:expressive}]\label{lem:indicator}
(1) $RG_n(\ovarrow\cH)$ is the set of all $n$-ary polymorphisms of $\ovarrow\cH$.\\[2mm] 
(2) $RG_n(\ovarrow\cH)$ is conjuctive-definable in $\ovarrow\cH$.
%% (3) $\ovarrow\cH$ has a Mal'tsev polymorphism if and only if $RG_3(\ovarrow\cH)$ contains a tuple $\vf$ such that $\vf(c_1,c_1,c_2)=\vf(c_2,c_1,c_1)=c_2$ for any $c_1,c_2\in H_i$ for every $i\in[k]$.
\end{lemma}

\begin{proof}[Proof of Theorem~\ref{the:ConstantCSP-main}]
(1) Let $\cP$ be an instance of $\NpCSP(\ovarrow\cH^=)$. The goal is to eliminate all equality relations $=_{H_i}$ for each domain $H_i$, where $i \in [k]$. To achieve this, we apply a straightforward procedure: For every constraint of the form $=_{H_i}(u, v)$ in $\cP$, we remove it from $\cP$ and replace every occurrence of $v$ with $u$.

The resulting problem instance, denoted as $\cP^*$, belongs to $\NpCSP(\ovarrow\cH)$. This reduction can clearly be executed in polynomial time. Additionally, although the set of solutions to $\cP^*$ differs from the set of solutions to $\cP$ (as some variables are eliminated), both instances have the same cardinality.

(2) We follow the same line of argument as the proof of a similar result in \cite{ref:BULATOV_TowardDichotomy} for exact counting, and \cite{DBLP:conf/stoc/BulatovK22} for modular counting. Let $H = \colect{H}{i\in[k]}$, where $H_i = \{ a_{i,1}, ..., a_{i,\ell_{i}} \}$, and let $\relo = \{(\ovarrow\vf_1(a_{1,1})\zd, \ovarrow\vf_1(a_{1,\ell_1}), \zd \ovarrow\vf_k(a_{k,1}, \zd \ovarrow\vf_k(a_{k,\ell_{k}}))  \mid \ovarrow\vf \in \Hom(\ovarrow\cH, \ovarrow\cH) \}$ be the relation conjunctive definable by Lemma~\ref{lem:indicator} through the indicator problem. By Theorem~\ref{the:ConstantCSP-main}(1) and Lemma~\ref{lem:conjunctive} we may assume that $\ovarrow\cH$ has $=_H$ and $\relo$ as its predicates.

Let $\cP = (V;\cC)$ be an instance of $\NpCSP(\ovarrow\cH^\const)$. We construct an instance $\cP'=(V',\cC')$ of $\NpCSP(\ovarrow\cH)$ as follows. 
\begin{itemize}
    \item 
    $V'= V \cup \bigcup_{i \in [k]} \{v_{a_{i,j}}\mid a_{i,j} \in H_i \}$;
    \item
    $\cC'$ consists of three parts: $\{C=\ang{\bx,\rel}\in\cC\mid \rel\in\sg\}$, $\{\ang{(v_{a_1}\zd v_{a_n}),\relo}\}$, and\linebreak  $\{\ang{(x,v_{a_{i,j}}),=_{H_i}}\mid \ang{(x),C_{a_{i,j}}}\in\cC\}$.
\end{itemize}

The number of solutions of $\cP$ equals the number of solutions $\ovarrow\vf$ of $\cP'$ such that $\vf_i(v_{a_{i,j}}) = a_{i,j}$ for all $a_{i,j} \in H_i$ where $i \in [j], j\in [\ell_i]$. Let $U$ be the set of all such solutions of $\cP'$ and $T=|U|$. Then $T$ can be computed in two stages.

Let again $\Part(H)$ be the poset of partitions of $H = \colect{H_i}{i \in [k]}$. Note that each partition consists of collection of partitions as $\theta = \colect{\theta_i}{i \in [k]}$ where each $\theta_i$ is a partition for $H_i$. For every partition $\th\in\Part(H)$ we define $\cP'_\th$ as the instance $(V',\cC_\th)$, where 
\[
    \cC_\th=\cC'\cup\{ \constCSP{ (v_{a_{i,j}}, v_{a_{i,j'}})}{=_{H_i}}\mid (a_{i,j},a_{i,j'})\in\th_i \}.
\]
Note that if $\ovarrow\vf$ is a solution of $\cP'$, then $\vf$ is a solution of $\cP'_{\theta}$ if and only if $\vf_i(v_{a_{i,j}}) = \vf_i(v_{a_{i,j'}})$ for every $a_{i,j},a_{i,j'}$ from the same class of $\theta_i$. Let us denote by $M(\th)$ the number of solutions of $\cP'_\th$. The number $M(\th)$ can be computed using the oracle $\NpCSP(\cH)$, since we assume that all $=_{H_i}$'s and $\relo$ are predicates of $\cH$.

Next, we find the number of solutions $\ovarrow\vf$ of $\cP'$ that assign $v_{a_{i,j}}$, $a_{i,j}\in H_i$, pairwise different values. Let $W$ be the set of all such solutions. Let us denote by $N(\th)$ the number of all solutions $\vf$ of $\cP'_\th$ such that $\ovarrow\vf_i(v_{a_{i,j}})=\ovarrow\vf_i(v_{a_{i,j'}})$ if and only if $a_{i,j},a_{i,j'}$ belong to the same class of $\th_i$. In particular, $N(\zz)=|W|$. The number $N(\zz)$ can be obtained using the M\"obius inversion formula for the poset $\mathsf{Part}(H)$. Let $\FUNC{w}{H}{\mathbb{Z}}$ be M\"obius-inversion function. Also, observe that for any $\th\in\Part(H)$
\[
M(\th)=\sum_{\eta\ge\th}N(\eta).
\]
Therefore, 
\[
N(\zz) = \sum_{\th\in\Part(H)} w(\th) M(\th).
\]
Thus $N(\zz)$ can be found through a constant number of calls to $\NpCSP(\ovarrow\cH)$.

Now, we express $T$ via $N(\zz)$. Let $G=\Aut(\ovarrow\cH)$ be the automorphism group of $\cH$. We show that $W= \{ g\circ\vf\mid g \in G, \vf \in U \}$. For every solution $\vf$ in $U$ and every $g \in G$, $g\circ\vf$ is also a solution of $\cP'$. Moreover, since $g$ is one-to-one, $g\circ\vf$ is in $W$. Conversely, for every $\psi \in W$, there exists some $g \in G$ such that $g(a) =\psi(v_a)$, $a \in H$. Note that $g^{-1} \in G$ implies $\vf =g^{-1}\circ\psi \in U$, which witnesses that  $\psi= g\circ\vf$ is of the required form.
Finally, for every $\vf,\vf' \in U$ and every $g,g'\in  G$, if $\vf|_{\{v_a\mid a \in H \}}\ne\vf'|_{\{v_a\mid a \in H \}}$ or $g\ne g'$ then $g\circ\vf\ne g'\circ \vf'$. Thus, $N(\zz)=|G|\cdot T$. Since $\cH$ is $p$-rigid,  $|G|\not \equiv 0\pmod p$. Therefore  $T \equiv |G|^{-1}\cdot N(\zz)\pmod{p}$.
\end{proof}

\end{document}